\documentclass{amsart}
\usepackage{amsmath,amssymb}
\usepackage{tikz-cd}

\usepackage{amsmath}
\usepackage{amsfonts}
\usepackage{amssymb,amsthm}

\def\scal#1#2{\langle #1\bv#2 \rangle}
\def\sscal#1#2{\langle #1\bv\!\!\bv#2 \rangle}
\def\shuffle{\mathop{_{^{\sqcup\!\sqcup}}}}
\def\hshuffle{\mathop{\,\underline{\!\shuffle\!}\,}}
\def\halfshuffle{\hshuffle\limits_2} 
\def\ncp#1#2{#1\langle #2\rangle}
\def\ncs#1#2{#1\langle \!\langle #2\rangle \!\rangle}

\newcommand{\calA}{{\mathcal A}}
\newcommand{\calB}{{\mathcal B}}
\newcommand{\calC}{{\mathcal C}}
\newcommand{\calD}{{\mathcal D}}

\newcommand{\calF}{{\mathcal F}}

\newcommand{\calH}{{\mathcal H}}
\newcommand{\calI}{{\mathcal I}}
\newcommand{\calJ}{{\mathcal J}}
\newcommand{\calK}{{\mathcal K}}
\newcommand{\calL}{{\mathcal L}}
\newcommand{\calM}{{\mathcal M}}

\newcommand{\calR}{{\mathcal R}}

\newcommand{\calT}{{\mathcal T}}
\newcommand{\calU}{{\mathcal U}}
\newcommand{\calV}{{\mathcal V}}

\newcommand{\calY}{{\mathcal Y}}

\newcommand{\N}{{\mathbb N}}

\newcommand{\Q}{{\mathbb Q}}
\newcommand{\R}{{\mathbb R}}
\newcommand{\C}{{\mathbb C}}

\newcommand{\Frac}[2]{\displaystyle \frac{#1}{#2}}
\newcommand{\Sum}[2]{\displaystyle{\sum_{#1}^{#2}}}
\newcommand{\Prod}[2]{\displaystyle{\prod_{#1}^{#2}}}

\def\Lyn{{\mathcal Lyn}}

\def\abs#1{|#1|}
\def\absv#1{\|#1\|}

 \def\shuffle{\mathop{_{^{\sqcup\!\sqcup}}}}

\newtheorem{corollary}{Corollary}
\newtheorem{proposition}{Proposition}
\newtheorem{theorem}{Theorem}
\newtheorem{lemma}{Lemma}
\newtheorem{definition}{Definition}

\newtheorem{example}{Example}

\newtheorem{remark}{Remark}

\newcommand{\Li}{\operatorname{Li}}
\newcommand{\ad}{\operatorname{ad}}

\def\L{\mathrm{L}}

\def\F{\mathrm{F}}
\def\G{\mathrm{G}}

\def\Sum{\displaystyle\sum}
\def\Prod{\displaystyle\prod}
\def\Frac{\displaystyle\frac}
\def\path{\rightsquigarrow}

\def\bv{\mid}

\def\abs#1{\bv\!#1\!\bv}

\begin{document}
\title[Solutions of Universal Differential Equation]{\small On The Solutions of Universal Differential Equation by Noncommutative Picard-Vessiot Theory}

\author{V.C. Bui}
\address{University of Sciences, Hue University, 77 - Nguyen Hue street - Hue city, Vietnam.}
\curraddr{University of Sciences, Hue University, 77 - Nguyen Hue street - Hue city, Vietnam.}
\email{bvchien@hueuni.edu.vn}
\thanks{}

\author{V. Hoang Ngoc Minh}
\address{Universit\'e de Lille, 1 Place D\'eliot, 59024 Lille, France.}
\curraddr{Universit\'e de Lille, 1 Place D\'eliot, 59024 Lille, France.}
\email{vincel.hoang-ngoc-minh@univ-lille.fr}
\thanks{}

\author{Q.H. Ng\^o}
\address{Hanoi University of Science and Technology, 1 Dai Co Viet, Hai Ba Trung, Ha Noi, Viet Nam.}
\curraddr{Hanoi University of Science and Technology, 1 Dai Co Viet, Hai Ba Trung, Ha Noi, Viet Nam.}
\email{hoan.ngoquoc@hust.edu.vn}
\thanks{}

\author{V. Nguyen Dinh}
\address{University of Science and Technology of Hanoi, 18 Hoang Quoc Viet, Cau Giay, Ha Noi, VietNam.}
\curraddr{University of Science and Technology of Hanoi, 18 Hoang Quoc Viet, Cau Giay, Ha Noi, Viet Nam.}
\email{nguyen-dinh.vu@usth.edu.vn}
\thanks{}





\begin{abstract}
Basing on the algebraic combinatorics on noncommutative series with holomorphic coefficients, various recursive constructions of sequences of grouplike series converging to solutions of universal differential equation are proposed.
Basing on monoidal factorizations, these constructions intensively use diagonal series and various pairs of bases in duality, in concatenation-shuffle
bialgebra and in a Loday's generalized bialgebra. As applications, the unique solution, satisfying asymptotic conditions, of  \textit{universal}
Knizhnik-Zamolodchikov equation is provided by \textit{d\'evissage}.
\end{abstract}
\maketitle

\section{Introduction}\label{intro}
Providing more explanations concerning the short text \cite{QTS12} and continuing the work of \cite{CM}, this work consists of expliciting solutions of universal differential equation (see \eqref{NCDE} below, when the solutions exist) using in particular Volterra expansions for the Chen series. Ultimately, applied to the universal Knizhnik-Zamolodchikov (see \eqref{KZn} below, \cite{etingof}), this provides by \textit{d\'evissage} (\textit{i.e.} solutions of\footnote{KZ is an abbreviation of V. Knizhnik and A. Zamolodchikov.} $KZ_n$ are obtained by use of solutions of $KZ_{n-1}$ and the noncommutative generating series of hyperlogarithms \cite{Linz}) the unique grouplike solution satisfying asymptotic conditions. These solutions use a Picard-Vessiot theory of noncommutative differential equations \cite{orlando} and various factorizations of Chen series, for which, in Section \ref{combinatorialframeworks} below, almost notations of formal series, on the noncommutative variables belonging to the alphabet $\calT_n=\{t_{i,j}\}_{1\le i<j\le n}$ and with coefficients in a ring $(\calA,1_{\calA})$, arise in \cite{berstel,lothaire,reutenauer,viennotgerard}. In particular, the rings \footnote{\label{discretetopology}The ring $\ncs{\calA}{\calT_n}$ is additionally endowed with the discrete topology, \textit{i.e.} $\abs{S-T}=2^{-\varpi(S-T)}$, for $S,T\in\ncs{\calA}{\calT_n}$, where $\varpi(S)$ is the valuation of a series $S$ \cite{berstel}.} of (Lie) series and of (Lie) polynomials over $\calT_n$, are denoted, respectively, by ($\ncs{\calL ie_{\calA}}{\calT_n}$ and $\ncp{\calL ie_{\calA}}{\calT_n}$) $\ncs{\calA}{\calT_n}$ and $\ncp{\calA}{\calT_n}$. According to different contexts in Section \ref{universalequation} below, the ring $\calA$ can be incarnated in the ring of complex numbers, $(\C,1)$, or in
the ring of holomorphic functions over $\calV$ (denoted by $(\calH(\calV),1_{\calH(\calV)})$), or in the wedge algebra of holomorphic forms over $\calV$ (denoted by $\Omega(\calV)$),
where $\calV$ is a simply connected differentiable manifold of $\C^n$.

The coefficients of $S$, \textit{i.e.} $\{\scal{S}{w}\}_{w\in\calT_n^*}$ belonging to $\calH(\calV)$, are holomorphic and the partial differentiations $\{\partial_i\scal{S}{w}\}_{1\le i\le n}$ are well defined. So is the differential
$d\scal{S}{w}=\partial_1\scal{S}{w}dz_1+\cdots+\partial_n\scal{S}{w}dz_n$. Hence, in Sections \ref{universalequation}--\ref{free} below, one can study the following first order noncommutative differential equation \cite{orlando}, so-called \textit{universal} differential equation, over $\ncs{\calH(\calV)}{\calT_n}$,
\begin{eqnarray}\label{NCDE}
{\bf d}S=\sum_{w\in\calT_n^*}d\scal{S}{w}w=M_nS,
&\mbox{where}&M_n
=\sum_{1\le i<j\le n}\omega_{i,j}t_{i,j}\in\ncp{\calL ie_{\Omega(\calV)}}{\calT_n}.
\end{eqnarray}
\textit{Universality} can be seen as, replacing each letter $t_{i,j}\in\calT_n$ by a constant matrix $\calM(t_{i,j})$ (resp. a holomorphic vector field $\calY(t_{i,j})$), one obtains a linear (resp. nonlinear) differential equation \cite{deligne,dyson,FPSAC96} (resp. \cite{Linz,ACA,CM}).

In particular, to the partition $\calT_n=T_n\sqcup\calT_{n-1}$, onto $\calT_{n-1}$ and $T_n=\{t_{k,n}\}_{1\le k\le n-1}$, corresponds the split of the universal connection
$M_n$, onto $M_{n-1}$ and $\bar M_n$:
\begin{eqnarray}\label{split}
M_n=\bar M_n+M_{n-1},&\mbox{where}&\bar M_n:=\sum_{k=1}^{n-1}\omega_{k,n}\;t_{k,n}\in\ncp{\calL ie_{\Omega(\calV)}}{T_n}.
\end{eqnarray}

Note that one can use the following intermediate alphabet in bijection with $\calT_n$
\begin{eqnarray}\label{intermediatealphabet}
X=\{x_k\}_{1\le j\le N},&\mbox{with}&N=n(n-1)/2\ge1,
\end{eqnarray}
for which one can use the diffential forms $\{\omega_{i}\}_{1\le i\le N}$ in bijection with $X$ and then (see also \eqref{Mn}--\eqref{Mn2} below)
\begin{eqnarray}
{\bf d}S=M_nS,&\mbox{where}&M_n:=\sum_{i=1}^N\omega_{i}x_j\in\ncp{\calL ie_{\Omega(\calV)}}{X}.
\end{eqnarray}
It follows that (see also \eqref{splitbis} below for example)
\begin{eqnarray}\label{Mn}
M_n=\sum_{1\le i<j\le n}\omega_{i,j}\;t_{i,j}=\sum_{1\le k\le N}F_k\;x_k=\sum_{1\le l\le n}U_l\;dz_l,
\end{eqnarray}
where
\begin{eqnarray}\label{Mn2}
F_k=\sum_{1\le l\le n}f_{l,k}\;dz_j&\mbox{and then}&U_l=\sum_{1\le k\le N}f_{l,k}\;x_k.
\end{eqnarray}

For any $S\neq0$ belonging to the integral ring $\ncs{\calH(\calV)}{\calT_n}$,
if $S$ is solution of \eqref{NCDE} then, by \eqref{Mn}--\eqref{Mn2}, one might have
\begin{eqnarray}\label{dSbis}
{\bf d}S=M_nS=\sum_{1\le l\le n}({\bf\partial}_lS)\;dz_l,&\mbox{with}&{\bf\partial}_lS=U_lS.
\end{eqnarray}
Since ${\bf\partial}_j{\bf\partial}_iS=(({\bf\partial}_jU_i)+U_iU_j)S$
and ${\bf\partial}_i{\bf\partial}_jS={\bf\partial}_j{\bf\partial}_iS$ then $(({\bf\partial}_jU_i)-({\bf\partial}_iU_j)+[U_i,U_j])S=0$ and then
${\bf\partial}_iU_j-{\bf\partial}_jU_i=[U_i,U_j]$, $1\le i,j\le n$. Or equivalently,
${\bf d}M_n=M_n\wedge M_n$ inducing a Lie ideal of relators on $\{t_{i,j}\}_{1\le i<j\le n}$, $\calJ_n$, and \eqref{NCDE} is solved over $\ncs{\calH(\calV)}{\calT_n}$ and then $\ncs{\calH(\calV)}{\calT_n}/\calJ_n$ as explained in Section \ref{Chenseries} below.

According to \cite{drinfeld1}, $M_n$ is said to be \textit{flat} and \eqref{NCDE} is said to be \textit{completely integrable}.

With the discrete topology, solution of \eqref{NCDE}, when exists, can be usually computed by the following convergent Picard's iteration over the topological basis $\{w\}_{w\in\calT_n^*}$
\begin{eqnarray}\label{picard0}
&F_0(\varsigma,z)=1_{\calH(\calV)},
&F_i(\varsigma,z)=F_{i-1}(\varsigma,z)+\int_{\varsigma}^zM_n(s)F_{i-1}(s),i\ge1,
\end{eqnarray}
and the sequence $\{F_k\}_{k\ge0}$ admits the limit, also called Chen series (see \cite{cartier1,Chen1954,Mathieu} and their bibliographie) of the holomorphic $1$-forms $\{\omega_{i,j}\}_{1\le i<j\le n}$ and along a path $\varsigma\path z$ over $\calV$, modulo $\calJ_n$,
is viewed as the fundamental solution of \eqref{NCDE}.

More generally, by a Ree's theorem Chen series is grouplike (see \cite{cartier1,reutenauer}), belonging to $e^{\ncs{\calL ie_{\calH(\calV)}}{\calT_n}}$, and can be put in the MRS\footnote{MRS is an abbreviation of G. M\'elan\c{c}on, C. Reutenauer and M.P. Sch\"utzenberger.} factorization form \cite{hoangjacoboussous,CM} (see Proposition \ref{MRSBTT} and Corollary \ref{Chensolution} below) and, since the rank of the module of solutions of \eqref{NCDE} is at most equals $1$ then, under the action of the Haussdorf group, \textit{i.e.} $e^{\ncs{\calL ie_{\C}}{\calT_n}}$ playing the r\^ole of the differential Galois group of \eqref{NCDE}, any grouplike solution of \eqref{NCDE} can be computed by multiplying on the right of the previous Chen series, modulo $\calJ_n$, by an element of Haussdorf group (containing the monodromy group of \eqref{NCDE}, see \cite{orlando,legsloops}). In practice, infinite solutions of \eqref{NCDE} can be computed using convergent iterations over $\ncs{\calH(\calV)}{\calT_n}$ and then $\ncs{\calH(\calV)}{\calT_n}/\calJ_n$.

A challenge is then to explicitly and exactly compute (and to study) these limits of convergent sequences of (not necessarily grouplike) series on the dual topological ring and over various corresponding dual topological bases. For that, on the one hand, thanks to the algebraic combinatorics on noncommutative series (recalled in Section \ref{combinatorialframeworks} below) and, on the other hand, by means of a noncommutative symbolic calculus (introduced in Section \ref{Iterated integrals} below) and a Picard-Vessiot theory of noncommutative differential equations (outlined in Section \ref{eqdiffnoncommutatice} below), solutions of \eqref{NCDE} are explicitly computed (in Section \ref{Chenseries} below). Applying \eqref{NCDE}--\eqref{split} and \eqref{picard0}, in Section \ref{Knizhnik-Zamolodchikov} below, substituting $t_{i,j}$ by $t_{i,j}/2\mathrm{i}\pi$ and specializing $\omega_{i,j}$ to $d\log(z_i-z_j)$ and then $\calV$ to the universal covering, $\widetilde{\C_*^n}$, of the configuration space of $n$ points on the plane \cite{kohno3,kohno4}, $\C_*^n:=\{z=(z_1,\ldots,z_n)\in\C^n|z_i\ne z_j\mbox{ for }i\neq j\}$, various expansions of Chen series over $\ncs{\calH(\widetilde{\C_*^n})}{\calT_n}$ (see Theorem \ref{Chen_braids} and Corollary \ref{finitefactorization} below) will provide solutions of the following noncommutative differential equation\footnote{So-called $KZ_n$ equation and $\Omega_n$ is called universal KZ connection form with $N$ (determined in \eqref{intermediatealphabet}) logarithmic singularities.} (given in Theorem \ref{KZsol} and Corollaries \ref{P}--\ref{C} below)
\begin{eqnarray}\label{KZn}
{\bf d}F=\Omega_nF,&\mbox{where}&\Omega_n(z):=\sum_{1\le i<j\le n}\frac{t_{i,j}}{2{\rm i}\pi}d\log(z_i-z_j),
\end{eqnarray}
and is splitting as follows (Proposition \ref{crochet} below will examine the flatness $\Omega_n$ and integrability conditions of \eqref{KZn}, see also Lemma \ref{condition} and Remark \ref{BE} below)
\begin{eqnarray}\label{N}
\Omega_n=\bar{\Omega}_n+\Omega_{n-1},&\mbox{where}&
\bar{\Omega}_n(z):=\sum_{k=1}^{n-1}\Frac{t_{k,n}}{2{\rm i}\pi}d\log(z_k-z_n).
\end{eqnarray}
In particular, let $\Sigma_{n-2}=\{z_1,\ldots,z_{n-2}\}\cup\{0\}$ (one puts $z_{n-1}=0$) be the set of singularities and $s=z_n$. For\footnote{\label{NOTE} $z_n$ is variate moving towards $z_{n-1}$ and $z_k=a_k$ is fixed and then $d(z_n-z_k)=dz_n=ds$.} $z_n\to z_{n-1}$, the connection $\bar{\Omega}_n$ behaves as $(2{\rm i}\pi)^{-1}N_{n-1}$, where $N_{n-1}$ is nothing but the connection of the differential equation satisfied by the noncommutative generating series of hyperlogarithms (see \eqref{sghyperlog}--\eqref{DEN} below)
\begin{eqnarray}\label{Nbis}
N_{n-1}(s):=t_{n-1,n}\frac{ds}{s}-\sum_{k=1}^{n-2}t_{k,n}\frac{ds}{z_k-s}
\in\ncp{\calL ie_{\Omega(\widetilde{\C\setminus\Sigma_{n-2}})}}{T_n}.
\end{eqnarray}

\begin{example}\label{$KZ_3$}
\begin{itemize}
\item If $n=2$ then $\calT_2=\{t_{1,2}\}$ and $\Omega_2(z)=(t_{1,2}/2{\rm i}\pi)d\log(z_1-z_2)$.
A solution of ${\bf d}F=\Omega_2F$ is
$F(z_1,z_2)=e^{(t_{1,2}/2{\rm i}\pi)\log(z_1-z_2)}=(z_1-z_2)^{t_{1,2}/2{\rm i}\pi}$
and it belongs to $\ncs{\calH(\widetilde{\C_*^2})}{\calT_2}$,

\item For $n=3,\calT_3=\{t_{1,2},t_{1,3},t_{2,3}\}$ and $\Omega_3(z)=\bar{\Omega}_3+\Omega_2(z)$, where
$\bar{\Omega}_3=(t_{1,3}d\log(z_1-z_3)+t_{2,3}d\log(z_2-z_3))/{2{\rm i}\pi}
\in\ncp{\calL ie_{\Omega(\widetilde{\C\setminus\{0,z_1\}})}}{t_{1,2},t_{2,3}}$, which behaves as $N_2(s)=(t_{1,2}s^{-1}{ds}-t_{2,3}(z_1-s)^{-1}{ds})/2{\rm i}\pi$,
by putting $z_2=0$ and $z_1=1$, see also Appendix \ref{AppendixB}.
\end{itemize}
\end{example}

\begin{example}\label{$KZ_3$bis}
\begin{itemize}
\item Solution of ${\bf d}F=\Omega_3F$ can be computed as limit of the sequence $\{F_l\}_{l\ge0}$,
in $\ncs{\calH(\widetilde{\C_*^3})}{\calT_3}$, by convergent Picard's iteration as in \eqref{picard0}
\begin{eqnarray*}
&F_0(z^0,z)=1_{\calH(\widetilde{\C_*^n})},&F_i(z^0,z)=F_{i-1}(z^0,z)+\int_{z^0}^z\Omega_3(s)F_{i-1}(s),i\ge1.
\end{eqnarray*}

\item Let us compute, by another way, a solution of ${\bf d}F=\Omega_3F$ thanks to the sequence $\{V_l\}_{l\ge0}$,
in $\ncs{\calH(\widetilde{\C_*^3})}{\calT_3}$, satisfying the following recursion\footnote{This recursion is
different with respect to the exposure pattern in \eqref{S_k} below.}
\begin{eqnarray*}
V_0(z)&=&e^{(t_{1,2}/2{\rm i}\pi)\log(z_1-z_2)},\\
V_l(z)&=&V_0(z)\int_0^zV_0^{-1}(s)\Big(\frac{t_{1,3}}{2{\rm i}\pi}d\log(z_1-z_3)+\frac{t_{2,3}}{2{\rm i}\pi}d\log(z_2-z_3)\Big)V_{l-1}(s)\\
&=&e^{(t_{1,2}/2{\rm i}\pi)\log(z_1-z_3)}\int_0^ze^{-(t_{1,2}/2{\rm i}\pi)\log(s_1-s_2)}\bar{\Omega}_3(s)V_{l-1}(s).
\end{eqnarray*}
\end{itemize}
\end{example}

The Chen series, of the holomorphic $1$-forms $\{d\log(z_i-z_j)\}_{1\le i<j\le n}$ and along the path $z^0\path z$
over universal covering $\widetilde{\C_*^n}$, can be used to determine solutions of \eqref{KZn} and depends on the
differences $\{z_i-z_j\}_{1\le i<j\le n}$, as will be treated in Section \ref{free} below to illustrate our purposes.
Furthermore, the universal KZ connection form $\Omega_n$ satisfies the following identity \cite{drinfeld1}
(see also Proposition \ref{crochet} below)
\begin{eqnarray}\label{integrable}
{\bf d}\Omega_n-\Omega_n\wedge\Omega_n=0
\end{eqnarray}
then $\Omega_n$ is flat and \eqref{KZn} is completely integrable. It turns out that \eqref{integrable} induces 
the relators associated to following relations on $\{t_{i,j}\}_{1\le i<j\le n}$ \cite{kohno1,kohno2,kohno3}.
\begin{eqnarray}\label{braidbis}
\qquad
\calR_n=
\left\{\begin{array}{rll}
[t_{i,k}+t_{j,k},t_{i,j}]=0&\mbox{for distinct }i,j,k,&1\le i<j<k\le n,\cr
[t_{i,j}+t_{i,k},t_{j,k}]=0&\mbox{for distinct }i,j,k,&1\le i<j<k\le n,\cr
[t_{i,j},t_{k,l}]=0&\mbox{for distinct }i,j,k,l,&
\left\{{\displaystyle1\le i<j\le n,\atop\displaystyle1\le k<l\le n,}\right.
\end{array}\right.
\end{eqnarray}
generating the Lie ideal $\calJ_{\calR_n}$, of $\ncp{\calL ie_{\calH(\calV)}}{\calT_n}$, seemingly different to
the relators associated to the infinitesimal braid relators on $\{t_{i,j}\}_{1\le i,j\le n}$ \cite{drinfeld1}:
\begin{eqnarray}\label{braid}
\calR'_n=\left\{\begin{array}{rcllll}
t_{i,j}&=&0&\mbox{for }i=j,\cr
t_{i,j}&=&t_{j,i}&\mbox{for distinct }i,j,\cr
[t_{i,k}+t_{j,k},t_{i,j}]&=&0&\mbox{for distinct }i,j,k,\cr
[t_{i,j},t_{k,l}]&=&0&\mbox{for distinct }i,j,k,l.
\end{array}\right.
\end{eqnarray}
Solutions of \eqref{KZn} will be then expected belonging to $\ncs{\calH(\widetilde{\C_*^n})}{\calT_n}/\calJ_{\calR_n}$ and
the logarithm of grouplike solutions will be expected in $\ncs{\calL ie_{\calH(\widetilde{\C_*^n})}}{\calT_n}/\calJ_{\calR_n}$.
These expressions will be explicitly computed (see Section \ref{free} below).

Now, let us explain a strategy for solving \eqref{NCDE} throughout the universal KZ equation \eqref{KZn}. This involves in high energy physics \cite{Weinzierl} and has applications on representation theory of affine Lie algebra and quantum groups, braid groups, topology of hyperplane complements, knot theory \cite{cartier1,cartier,cartier2,drinfeld1,drinfeld2,Furusho,Jimbo1,Jimbo2,kohno1, kohno2,Mathieu}:
\begin{itemize}
\item According to \cite{Chen1954}, the Chen series $C_{\varsigma\path z}$, of $\{d\log(z_i-z_j)\}_{1\le i<j\le n}$ and along the concatenation of the paths $\varsigma\path z^0$ and $z^0\path z$ over $\calV$ is followed
\begin{eqnarray}\label{normalization0}
C_{\varsigma\path z}&=&C_{z^0\path z}C_{\varsigma\path z^0},\quad\mbox{or equivalently},\cr
\forall w\in\calT_n^*,\quad
\scal{C_{\varsigma\path z}}{w}&=&\sum_{u,v\in\calT_n^*,uv=w}\scal{C_{z^0\path z}}{u}\scal{C_{\varsigma\path z^0}}{v}.
\end{eqnarray}
On the other side, the coefficients of the Chen series, along $0\path z$ and of $\{d\log(z_i-z_j)\}_{1\le i<j\le n}$, are not well defined. For example, for any $1\le i<j\le n$, the integral $\int\limits_0^zd\log(z_i-z_j)$ is not defined. In general, strategies that are widely used in the literature are tangential base points\footnote{\textit{i.e.} simply connected regions in the neighborhood of the divisor at infinity.} \cite{deligne}.

Hence, in Section \ref{free} below, as an extension of the treatment on polylogarithms in \eqref{Li}
(resp. hyperlogarithms in \eqref{hyperlogarithms}) we will construct an other grouplike series for
computing solution of \eqref{KZn}, denoted by $\F_{KZ_n}$, such that
\begin{eqnarray}\label{normalization}
\F_{KZ_n}(z)=C_{z^0\path z}\F_{KZ_n}(z^0).
\end{eqnarray}
$\F_{KZ_n}(z)$ will normalize $C_{0\path z}$ (see Definitions \ref{Chenseriesdef} and \ref{F},
Corollaries \ref{P}--\ref{C} below) and, as a counter term, $\F_{KZ_n}(z^0)$ belongs to
$\{e^C\}_{C\in\ncs{\calL ie_{\C}}{\calT_n}}$.
These will be obtained as image, by tensor of morphisms of algebras, of the diagonal series over $\calT_n=T_n\sqcup\calT_{n-1}$ (see Lemma \ref{factorization}, Propositions \ref{dual_algebras}--\ref{expression} and Theorem
\ref{diagonal} below) over $(\ncp{\Q}{T_n},{\tt conc},1_{T_n^*},\Delta_{\shuffle})$ (resp. $(\ncp{\Q}{\calT_{n-1}},{\tt conc},1_{\calT_{n-1}^*},\Delta_{\shuffle})$) endowed pair
of dual bases $\{P_l\}_{l\in\Lyn T_n}$ and $\{S_l\}_{l\in\Lyn T_n}$ (resp. $\{P_l\}_{l\in\Lyn\calT_{n-1}}$ and $\{S_l\}_{l\in\Lyn\calT_{n-1}}$),  indexed by Lyndon words over $T_n$ (resp. $\calT_{n-1}$):
\begin{eqnarray}
\calD_{\calT_n}&=&\calD_{\calT_{n-1}}\Prod_{l=l_1l_2\atop{l_2\in\Lyn\calT_{n-1},l_1\in\Lyn T_n}}^{\searrow}e^{S_l\otimes P_l}
\calD_{T_n}\quad{\mbox{(decreasing lexicographical}\atop\mbox{ordered product)}}\\
&=&\calD_{T_n}\Big(1_{\calT_n^*}\otimes1_{\calT_n^*}\\
&+&\sum_{k\ge1}\sum_{v_1,\ldots,v_k\in T_n^*\atop t_1,\ldots,t_k\in\calT_{n-1}}
a(v_1t_1)\halfshuffle(\cdots\halfshuffle a(v_kt_k)\ldots))\otimes r(v_1t_1)\ldots r(v_kt_k)\Big),\cr
\calD_{T_n}&=&\prod_{l\in\Lyn T_n}^{\searrow}e^{S_l\otimes P_l}\quad{\mbox{(decreasing lexicographical}\atop\mbox{ordered product),}}
\end{eqnarray}
where $\halfshuffle$ is the half-shuffle product \cite{Loday} and, for any $w=t_1\ldots t_m\in\calT_n^*$,
$a(w)=(-1)^mt_m\ldots t_1$ and $r(w)=\ad_{t_1}\circ\cdots\circ\ad_{t_{m-1}}t_m$.

Furthermore, considering $\calI_n$, the sub Lie algebra of $\ncs{\calL ie_{\Q}}{\calT_n}$ generated by
$\{\ad^k_{-T_n}t\}_{t\in\calT_{n-1}}^{k\ge0}$, the enveloping algebra $\calU(\calI_n)$ and its dual
$\calU(\calI_n)^{\vee}$ are generated by the dual bases (see Section \ref{moreabout} below)
\begin{eqnarray}
\calB&=&\{\ad_{-T_n}^{k_1}t_1\ldots\ad_{-T_n}^{k_p}t_p\}_{t_1,\ldots,t_p\in\calT_{n-1}}^{k_1,\ldots,k_p\ge0,p\ge1},\\
\calB^{\vee}&=&\{a(T_n^{k_1}t_1)\halfshuffle(\cdots\halfshuffle a(T_n^{k_p}t_p)\ldots))
\}_{t_1,\ldots,t_p\in\calT_{n-1}}^{k_1,\ldots,k_p\ge0,p\ge1}.
\end{eqnarray}

\item With the previous expressions of the diagonal series $\calD_{\calT_n}$, for $z_n\to z_{n-1}$, grouplike solutions of \eqref{KZn}--\eqref{N} will be of the form
$h(z_n)H(z_1,\ldots,z_{n-1})$ (see Note \ref{NOTE} and Proposition \ref{volterraexp}--\ref{sol},
Theorems \ref{Chen_braids}--\ref{KZsol}, Corollary \ref{P} below) such that
\begin{itemize}
\item $h$ is solution of $df=(2{\rm i}\pi)^{-1}N_{n-1}f$, where $N_{n-1}$ is the connection determined
in \eqref{Nbis}. Hence, $h(z_n)\sim_{z_n\to z_{n-1}}(z_{n-1}-z_n)^{t_{n-1,n}/2\mathrm{i}\pi}$.

\item $H(z_1,\ldots,z_{n-1})$ is solution of ${\bf d}S=\Omega_{n-1}^{\varphi_n}S$, where
\begin{eqnarray}
&\Omega_{n-1}^{\varphi_n}(z)=\sum_{1\le i<j\le n-1}d\log(z_i-z_j)\varphi_n^{(z^0,z)}(t_{i,j})/2\mathrm{i}\pi,&\\
&\varphi_n^{(z^0,z)}(t_{i,j})\sim_{z_n\to z_{n-1}}e^{\ad_{-\log(z_{n-1}-z_n))t_{n-1,n}/2\mathrm{i}\pi}}t_{i,j}\mod\calJ_n.&
\end{eqnarray}
\end{itemize}

\item With the discrete topology, an explict computation of the actual solution, $\F_{KZ_n}$, uses the following recursion
\begin{eqnarray}\label{S_k}
V_k(\varsigma,z)=V_0(\varsigma,z)\sum_{t_{i,j}\in\calT_{n-1}}
\int_{\varsigma}^z\omega_{i,j}(s)V_0^{-1}(\varsigma,s)t_{i,j}V_{k-1}(\varsigma,s)
\end{eqnarray}
and considers two different cases of starting condition, $V_0$, for \eqref{S_k}:
\begin{itemize}
\item as the grouplike series $(\alpha_{\varsigma}^z\otimes\mathrm{Id})\calD_{T_n}$.
In this case, $\{V_k\}_{k\ge0}$ converges to the unique
solution satisfying asymptotic conditions achieving the \textit{d\'evissage} (using the decreasing lexicographical order product):
\begin{eqnarray}\label{lavraie}
\F_{KZ_n}&=&\prod_{l\in\Lyn T_n}^{\searrow}e^{F_{S_l}P_l}\cr
&\times&\underbrace{\Big(1_{\calT_n^*}
+\sum_{v_1,\ldots,v_k\in T_n^*,k\ge1\atop t_1,\ldots,t_k\in\calT_{n-1}}
F_{a(v_1t_1)\halfshuffle\ldots\halfshuffle a(v_kt_k)}r(v_1t_1)\ldots r(v_kt_k)\Big)}_{\mbox{functional expansion of solution of $KZ_{n-1}$}}\cr
&=&\prod_{l\in\Lyn\calT_{n-1}}^{\searrow}e^{F_{S_l}P_l}
\Big(\prod_{l=l_1l_2\atop{l_2\in\Lyn\calT_{n-1},l_1\in\Lyn T_n}}^{\searrow}e^{F_{S_l}P_l}\Big)
\prod_{l\in\Lyn T_n}^{\searrow}e^{F_{S_l}P_l},
\end{eqnarray}

\item as $(\alpha_{\varsigma}^z\otimes\mathrm{Id})\calD_{T_n}\mod[\ncs{\calL ie_{\calH(\calV)}}{T_n},\ncs{\calL ie_{\calH(\calV)}}{T_n}]$
(see also Remarks \ref{nil_app} and \ref{BE} below). In this case, extending the treatment in \cite{drinfeld1} and considered in \eqref{abelianization} below, one gets an approximation of \eqref{lavraie}:
\begin{eqnarray}
\F_{KZ_n}&\equiv&e^{\sum_{t\in T_n}F_tt}\Big(1_{\calT_n^*}\cr
&+&\sum_{v_1,\ldots,v_k\in T_n^*,k\ge1\atop t_1,\ldots,t_k\in\calT_{n-1}}
F_{a(\hat v_1t_1)\halfshuffle(\ldots\halfshuffle(a(\hat v_kt_k))\ldots)}r(v_1t_1)\ldots r(v_kt_k)\Big),
\end{eqnarray}
where, for $w=t_1\ldots t_m\in\calT_n^*,\hat w=t_1\shuffle\ldots\shuffle t_m$.
\end{itemize}
Specializing the convergent case to \eqref{normalization}, it will illustrate, in Section \ref{Appendices},
with the cases of $KZ_4$ and, in a similar way, $KZ_3$ (achieving Example \ref{$KZ_3$bis}).
\end{itemize}

The organization of this paper is as follows
\begin{itemize}
\item In Section \ref{combinatorialframeworks}, some algebraic combinatorics of the diagonal series, on the concatenation-shuffle bialgebra and on a Loday's generalized bialgebra, will be recalled briefly by Theorem \ref{diagonal}. In particular, we will insist on Lazard and  Sch\"utzenberger monoidal factorizations leading to various dual topological bases on which will base the computations of the next sections.

\item In Section \ref{universalequation}, various expansions of Chen series will be provided by Propositions \ref{MRSBTT}--\ref{volterraexp}, Theorem \ref{Chen_braids} and Corollary \ref{finitefactorization} to obtain grouplike solutions of \eqref{NCDE} in the factorized forms, over $\ncs{\calH(\calV)}{\calT_n}$ and then over $\ncs{\calH(\calV)}{\calT_n}/\calJ_{\calR_n}$. In particular, by \eqref{split}, finite factorization is similar to \textit{d\'evissage}\footnote{See Note \ref{encore} below and the descryption in the begining of Section \ref{intro}.} of $KZ_n$.

\item In Section \ref{free}, some consequences for grouplike solutions of \eqref{KZn}, satisfying
asymptotic conditions, will be examined by Theorem \ref{KZsol} and Corollaries \ref{P}--\ref{C}.
\begin{example}
Grouplike solution of $KZ_3$ admits polylogarithms as local coordinates and solutions of $KZ_2$
(admitting elementary transcendental functions $\{\log(z_i-z_j)\}_{1\le i<j\le n}$ as  coordinates)
as in Example \ref{$KZ_3$}.
\end{example}
\end{itemize}

\section{Combinatorial frameworks}\label{combinatorialframeworks}
\subsection{Algebraic combinatorics on noncommutative series}\label{combinatorics}
Now, for fixed  $n$ and $T_k:=\{t_{j,k}\}_{1\le j\le k-1}$ ($2\le k\le n$), by \eqref{split} let us consider\footnote{In terms of cardinality, one has $\sharp{\calT_n}=n(n-1)/2$ and $\sharp{T_n}=n-1$. If $n\ge4$ then $\sharp{\calT_{n-1}}\ge\sharp{T_n}$.} $\calT_k=T_k\sqcup\calT_{k-1}$.
\begin{example}
\begin{enumerate}
\item $\calT_5=\{t_{1,2},t_{1,3},t_{1,4},t_{1,5},t_{2,3},t_{2,4},t_{2,5},t_{3,4},t_{3,5},t_{4,4}\}$, one has $T_5=\{t_{1,5},t_{2,5},t_{3,5},t_{4,5}\}$ and $\calT_4$,
\item $\calT_4=\{t_{1,2},t_{1,3},t_{1,4},t_{2,3},t_{2,4},t_{3,4}\}$, one has $T_4=\{t_{1,4},t_{2,4},t_{3,4}\}$ and $\calT_3$,
\item $\calT_3=\{t_{1,2},t_{1,3},t_{2,3}\}$, one has $T_3=\{t_{1,3},t_{2,3}\}$ and $\calT_2=\{t_{1,2}\}$.
\end{enumerate}
\end{example}

Let us consider the following total order $\calT_n$ and then over the sets of Lyndon words \cite{lothaire,reutenauer} $\Lyn\calT$ and $\Lyn\calT_n$ as follows (for $2\le k\le n$)
\begin{eqnarray}\label{orderLyndon}
T_2\succ\ldots\succ T_n,&t_{1,k}\succ\ldots\succ t_{k-1,k},&\Lyn T_2\succ\ldots\succ\Lyn T_n.
\end{eqnarray}
According to the Chen-Fox-Lyndon theorem \cite{lothaire,reutenauer,viennotgerard}, with the ordering in \eqref{orderLyndon}, there is a unique way to get the standard factorization of $l\in\Lyn\calT_n$, \textit{i.e.}
$st(l)=(l_1,l_2)$, where $l_2$ is the longest nontrivial proper right factor of $l$ or equivalently its smallest such for the lexicographic ordering \cite{lothaire}. Then
\begin{eqnarray}\label{orderLyndon3}
\Lyn\calT_{n-1}\succ\Lyn T_n.\Lyn\calT_{n-1}\succ\Lyn T_n,
\end{eqnarray}
More generally, for any $(t_1,t_2)\in T_{k_1}\times T_{k_2},2\le k_1<k_2\le n$, one also has
\begin{eqnarray}\label{orderLyndon4}
t_2t_1\in\Lyn\calT_{k_2}\subset\Lyn\calT_n&\mbox{and}&t_2\prec t_2t_1\prec t_1.
\end{eqnarray}

Hence, as consequences of \eqref{orderLyndon}--\eqref{orderLyndon3}, one obtains
\begin{itemize}
\item If $l\in\Lyn T_{k-1}$ and $t\in T_k,2\le k\le n$ then $tl\in\Lyn\calT_n$ and $t\prec tl\prec l$.
\item If $l_1\in\Lyn T_{k_1}$ and $l_2\in\Lyn T_{k_2}$ (for $2\le k_1<k_2\le n$) then
$l_2l_1\in\Lyn\calT_{k_2}\subset\Lyn\calT_n$ and $l_2\prec l_2l_1\prec l_1$.
\item If $l_1\in\Lyn T_k$ and $l_2\in\Lyn\calT_{k-1}$ (for $2\le k_1<k_2\le n$) then
$l_1l_2\in\Lyn\calT_n$ and $l_1\prec l_1l_2\prec l_2$.
\end{itemize}

In this Section, $\calA$ is a commutative integral ring containing $\Q$ and, by notations in \cite{berstel,lothaire,reutenauer},
$(\calT_n^*,1_{\calT_n^*})$ is the free monoid generated by $\calT_n$, for the concatenation denoted by $\tt conc$
(and it will be omitted when there is non ambiguity). The set of polynomials (resp. series) over
$\calT_n$ is denoted by $\ncp{\calA}{\calT_n}$ (resp. $\ncs{\calA}{\calT_n}$) and
$\ncs{\calA}{\calT_n}=\ncp{\calA}{\calT_n}^{\vee}$ (\textit{i.e} $\ncs{\calA}{\calT_n}$ is dual to
$\ncp{\calA}{\calT_n}$), via the following pairing
\begin{eqnarray}\label{pairing}
\ncs{\calA}{\calT_n}\otimes_{\calA}\ncp{\calA}{\calT_n}\longrightarrow\calA,&
T\otimes_{\calA}P\longmapsto\scal{T}{P}:=\Sum_{w\in\calT_n^*}\scal{T}{w}\scal{P}{w}.
\end{eqnarray}
In the sequel, all algebras, linear maps and tensor signs that appear in the following are over $\calA$ unless specified otherwise. The set of Lie polynomials (resp. Lie series), over $\calT_n$ with coefficients in $\calA$, is denoted by $\ncp{\calL ie_{\calA}}{\calT_n}$ (resp. $\ncs{\calL ie_{\calA}}{\calT_n}$). For convenience, the set of exponentials of Lie series will be denoted by
$e^{\ncs{\calL ie_{\calA}}{\calT_n}}=\{e^C\}_{C\in\ncs{\calL ie_{\calA}}{\calT_n}}$. The smallest algebra containing $\ncp{\calA}{\calT_n}$ and closed by rational operations (\textit{i.e.} addition, concatenation, Kleene star) is denoted by $\ncs{\calA^{\mathrm{rat}}}{\calT_n}$. Any $S\in\ncs{\calA^{\mathrm{rat}}}{\calT_n}$ is said to be rational and, by a Sch\"utzenberger's theorem \cite{berstel}, there is a linear representation $(\beta,\mu,\eta)$ of dimension $k\ge0$ such that (and conversely)
\begin{eqnarray}\label{recognized}
S=\beta((\mathrm{Id}\otimes\mu){\calD}_{\calT_n})\eta=\Sum_{w\in\calT_n^*}(\beta\mu(w)\eta)w,
\end{eqnarray}
where $\mu$ is the morphism of monoids from $X^*$ to ${\calM}_{k,k}(\calA)$, mappping each letter to a $k\times k$-matrix, $\beta$ is a column matrix in ${\calM}_{k,1}(\calA)$ and $\eta$ is a raw matrix in ${\calM}_{1,k}(\calA)$.

\begin{example}[\cite{orlando}]\label{automate}
To simplify, let $X$ be the alphabet $\{x_0,x_1\}$. The rational series $(t^2x_0x_1)^*$ and $(-t^2x_0x_1)^*$ admit, respectively, $(\nu_1,\{\mu_1(x_0),\mu_1(x_1)\},\eta_1)$ and $(\nu_2,\{\mu_2(x_0),\mu_2(x_1)\},\eta_2)$ as the linear representations given by
$$\begin{array}{rlrl}
\nu_1=\begin{pmatrix}1&0\end{pmatrix},
&\mu_1(x_0)=\begin{pmatrix}0&t\cr0&0\end{pmatrix},
&\mu_1(x_1)=\begin{pmatrix}0&0\cr t&0\end{pmatrix},
&\eta_1=\begin{pmatrix}1\cr0\end{pmatrix},\cr
\nu_2=\begin{pmatrix}1&0\end{pmatrix},
&\mu_2(x_0)=\begin{pmatrix}0&\mathrm{i}t\cr0&0\end{pmatrix},
&\mu_2(x_1)=\begin{pmatrix}0&0\cr\mathrm{i}t&0\end{pmatrix},
&\eta_2=\begin{pmatrix}1\cr0\end{pmatrix}.
\end{array}$$
\end{example}

Recall that $\ncs{\calA^{\mathrm{rat}}}{\calT_n}$ is also closed by shuffle which is denoted by $\shuffle$ and defined recursively, for any letters $x,y\in\calT_n$ and words $u,v\in\calT_n^*$, as follows \cite{berstel}
\begin{eqnarray}\label{recursion}
u\shuffle 1_{\calT_n^*}=1_{\calT_n^*}\shuffle u=u&\mbox{and}&(xu)\shuffle(yv)=x(u\shuffle yv)+y(v\shuffle xu).
\end{eqnarray}

\begin{example}[\cite{orlando}]
With the notations in Example \ref{automate}, one has (see \cite{orlando})
\begin{eqnarray*}
(-t^2x_0x_1)^*\shuffle(t^2x_0x_1)^*=(-4t^4x_0^2x_1^2)^*
\end{eqnarray*}
and $(-4t^4x_0^2x_1^2)^*$ admits $(\nu,\{\mu(x_0),\mu(x_1)\},\eta)$ as the linear representations given by
$$\nu=\begin{pmatrix}1&0&0&0\end{pmatrix},
\mu(x_0)=\begin{pmatrix}
0&\mathrm{i}t&t&0\cr
0&0&0&t\cr
0&0&0&\mathrm{i}t\cr
0&0&0&0
\end{pmatrix},\\
\mu(x_1)=\begin{pmatrix}
0&0&0&0\cr
\mathrm{i}t&0&0&0\cr
t&0&0&0\cr
0&t&\mathrm{i}t&0
\end{pmatrix},
\eta=\begin{pmatrix}1\cr0\cr0\cr0\end{pmatrix}.$$
\end{example}

By a Radford's theorem \cite{reutenauer}, the shuffle algebra, over $\calT_n$ and with coefficients in $\calA$, admits $\Lyn\calT_n$ as pure transcendence basis and then
\begin{eqnarray}
\mathrm{Sh}_{\calA}(\calT_n):=(\ncp{\calA}{\calT_n},\shuffle)\simeq(\calA[\{l\}_{l\in\Lyn\calT_n}],\shuffle).
\end{eqnarray}

Recall also that the following co-products (of $\tt conc$ and $\shuffle$)
\begin{eqnarray}
\Delta_{{\tt conc}}\mbox{ and }\Delta_{\shuffle}:\ncp{\calA}{\calT_n}&\longrightarrow&\ncp{\calA}{\calT_n}\otimes\ncp{\calA}{\calT_n}
\end{eqnarray}
are defined respectively, for any $u,v,w\in\calT_n^*$, as follows
\begin{eqnarray}
\scal{\Delta_{{\tt conc}}w}{u\otimes v}=\scal{w}{uv}
&\mbox{and}&
\scal{\Delta_{\shuffle}w}{u\otimes v}=\scal{w}{u\shuffle v}.
\end{eqnarray}
It follows, for any  $w\in\calT_n^*$, that \cite{cartier_Hopf}
\begin{eqnarray}
\Delta_{{\tt conc}}w=\sum_{u,v\in\calT_n^*,uv=w}u\otimes v
&\mbox{and}&
\Delta_{\shuffle}w=\sum_{u,v\in\calT_n^*}\scal{w}{u\shuffle v}u\otimes v.
\end{eqnarray}
\begin{example}
For any $t_1$ and $t_2\in\calT_n$, one has
\begin{eqnarray*}
\Delta_{{\tt conc}}(t_1t_2)&=&t_1t_2\otimes1_{\calT_n^*}+t_1\otimes t_2+t_1t_2\otimes1_{\calT_n^*},\cr
\Delta_{\shuffle}(t_1t_2)&=&t_1t_2\otimes1_{\calT_n^*}+t_1\otimes t_2+t_2\otimes t_1+1_{\calT_n^*}\otimes t_1t_2.
\end{eqnarray*}
\end{example}
In particular, $\Delta_{{\tt conc}}1_{\calT_n^*}=1_{\calT_n^*}\otimes1_{\calT_n^*}$ and $\Delta_{\shuffle}1_{\calT_n^*}=1_{\calT_n^*}\otimes1_{\calT_n^*}$. For any $t\in\calT_n$,
one also has $\Delta_{{\tt conc}}t=t\otimes1_{\calT_n^*}+1_{\calT_n^*}\otimes t$ and $\Delta_{\shuffle}t=t\otimes1_{\calT_n^*}+1_{\calT_n^*}\otimes t$. Hence, letters are primitive, for $\Delta_{{\tt conc}}$ and $\Delta_{\shuffle}$.

Both the products ${\tt conc}$ and $\shuffle$ and the co-products $\Delta_{{\tt conc}}$ and $\Delta_{\shuffle}$ are extended, for any $S$ and $R\in\ncs{\calA}{\calT_n}$, by
($SR$ and $S\shuffle R\in\ncs{\calA}{\calT_n}$ and, on the other hand, $\Delta_{{\tt conc}}S$ and $\Delta_{\shuffle}\in\ncs{\calA}{\calT_n^*\otimes\calT_n^*}$)
\begin{eqnarray}\label{extenoverseries}
SR=\sum_{u,v\in\calT_n^*\atop uv=w\in\calT_n^*}
\scal{S}{u}\scal{R}{v}
&\mbox{and}&
\Delta_{{\tt conc}}S=\sum_{w\in\calT_n^*}\scal{S}{w}\Delta_{{\tt conc}}w,\\
S\shuffle R=\sum_{u,v\in\calT_n^*}\scal{S}{u}\scal{R}{v}u\shuffle v
&\mbox{and}&
\Delta_{\shuffle}S=\sum_{w\in\calT_n^*}\scal{S}{w}\Delta_{\shuffle}w.
\end{eqnarray}

\begin{remark}[\cite{FPSAC95,FPSAC96,CM}]\label{Sweedler}
Let $(\beta,\mu,\eta)$ be a linear representation of dimension $k$ of $S\in\ncs{\calA^{\mathrm{rat}}}{\calT_n}$ which is also associated to the linear representations $(\beta,\mu,e_i)$ and $({}^te_i,\mu,\eta)$ of dimension $k$ of the rational series $\{L_i\}_{1\le i\le k}$ and $\{R_i\}_{1\le i\le k}$, where
\begin{eqnarray*}
e_i\in{\calM}_{1,k}(\calA)&\mbox{and}&{}^te_i=\begin{matrix}(0&\ldots&0&1\!\!_{_{_{\displaystyle i}}}&0&\ldots&0)\end{matrix}.
\end{eqnarray*}
By \eqref{recognized}, it follows that, for any $x,y\in\calT_n$, one has
\begin{eqnarray*}
\scal{S}{xy}=\beta\mu(x)\mu(y)\eta=\sum_{i=1}^k(\beta\mu(x)e_i)({}^te_i\mu(y)\eta)=\sum_{i=1}^k\scal{L_i}{x}\scal{R_i}{y},\cr
\scal{\Delta_{\tt conc}S}{x\otimes y}=\scal{S}{xy}=\sum_{i=1}^k\scal{L_i}{x}\scal{R_i}{y}=\sum_{i=1}^k\scal{L_i\otimes R_i}{x\otimes y}.
\end{eqnarray*}
\end{remark}

With these products and co-products, any series $S$ in $\ncs{\calA}{\calT_n}$ is said to be
\begin{itemize}
\item A character for $\tt conc$ (resp. $\shuffle$) if and only if, for $u,v\in\calT_n^*$,
\begin{eqnarray}\label{Freidrichs}
\scal{S}{uv}=\scal{S}{u}\scal{S}{v}&\mbox{(resp.}&\scal{S}{u\shuffle v}=\scal{S}{u}\scal{S}{v}).
\end{eqnarray}
Or equivalently, it is grouplike series for $\Delta_{\tt conc}$ (resp. $\Delta_{\shuffle}$) if and only if
\begin{eqnarray}
\scal{S}{1_{\calT_n^*}}=1\mbox{ and }\Delta_{\tt conc}(S)=\Phi(S\otimes S)\mbox{ (resp. }\Delta_{\shuffle}(S)=\Phi(S\otimes S)),
\end{eqnarray}
where $\Phi:\begin{tikzcd}[column sep=1.8em]
\ncs{\calA}{\calT_n}^{\vee}\otimes\ncs{\calA}{\calT_n}^{\vee}\ar[hook]{r}&(\ncs{\calA}{\calT_n}\otimes\ncs{\calA}{\calT_n})^{\vee}
\end{tikzcd}$ is injective.

\item An infinitesimal character, for $\tt conc$ (resp. $\shuffle$) if and only, for $w,v\in\calT_n^*$,
\begin{eqnarray}
\scal{S}{wv}&=&\scal{S}{w}\scal{v}{1_{\calT_n^*}}+\scal{w}{1_{\calT_n^*}}\scal{S}{v},\cr
(\mbox{resp. }
\scal{S}{w\shuffle v}&=&\scal{S}{w}\scal{v}{1_{\calT_n^*}}+\scal{w}{1_{\calT_n^*}}\scal{S}{v}).
\end{eqnarray}
Or equivalently, $S$ is a primitive series for $\Delta_{\tt conc}$ (resp. $\Delta_{\shuffle}$) if and only if
\begin{eqnarray}
\Delta_{\tt{conc}}S=1_{\calT_n^*}\otimes S+S\otimes1_{\calT_n^*}
&(\mbox{resp. }
\Delta_{\shuffle}S=1_{\calT_n^*}\otimes S+S\otimes1_{\calT_n^*}).
\end{eqnarray}
By a Ree's theorem \cite{reutenauer}, a Lie series is primitive for $\Delta_{\shuffle}$ and {\it vice versa}. For $\Delta_{\shuffle}$, when $\Phi$ is injective, if $S$ is grouplike then $\log S$ is primitive and, conversely, if $S$ is primitive then $e^S$ is grouplike. The sets of primitive polynomials, for $\Delta_{\shuffle}$ is $\mathrm{Prim}_{\shuffle}(\calT_n)=\ncp{\calL ie_{\calA}}{\calT_n}$ and $\mathrm{Prim}_{{\tt conc}}(\calT_n)=\calA.\calT_n$.
\end{itemize}

Finally, on the one hand, by\footnote{CQMM is an abbreviation of P. Cartier, D. Quillen, J. Milnor and J. Moore.} CQMM theorem, one has (see \cite{reutenauer})
\begin{eqnarray}
H_{{\tt conc}}(\calT_n)=&(\ncp{\calA}{\calT_n},{\tt conc},1_{\calT_n^*},\Delta_{\shuffle})&\simeq\calU(\ncp{\calL ie_{\calA}}{\calT_n}),\cr
H_{\shuffle}(\calT_n)=&(\ncp{\calA}{\calT_n},\shuffle,1_{\calT_n^*},\Delta_{{\tt conc}})&\simeq\calU(\ncp{\calL ie_{\calA}}{\calT_n})^{\vee},
\end{eqnarray}
and, on the other hand, the Sweedler's dual of $H_{\shuffle}(\calT_n)$ is followed \cite{reutenauer}
\begin{eqnarray}
H_{\shuffle}^{\circ}(\calT_n)=(\ncs{\calA^{\mathrm{rat}}}{\calT_n},\shuffle,1_{\calT_n^*},\Delta_{\tt conc}).
\end{eqnarray}
The last dual is defined, for any $S\in\ncs{\calA}{\calT_n}$, as follows \cite{reutenauer}
\begin{eqnarray}
S\in H_{\shuffle}^{\circ}(\calT_n)\iff\Delta_{\tt conc}(S)=\sum_{i\in I}L_i\otimes R_i,
\end{eqnarray}
where $I$ is finite and, by Remark \ref{Sweedler}, $\{L_i,R_i\}_{i\in I}$ can be selected in $\ncs{\calA^{\mathrm{rat}}}{\calT_n}$.

\begin{remark}
With the notations in Remark \ref{Sweedler}, one also has
\begin{eqnarray*}
S\in\ncs{\calA^{\mathrm{rat}}}{\calT_n}\iff\Delta_{\tt conc}(S)=\sum_{i\in I}L_i\otimes R_i.
\end{eqnarray*}
\end{remark}

Let $\abs{v}$ (resp.  $\abs{v}_t$) be the lenght (resp. number of occurrences of a letter $t$) of (resp. in) the word $v=t_1\ldots t_m$, associating to its mirroir $\tilde v=t_m\ldots t_1$ and to the following polynomials
\begin{eqnarray}\label{barv}
\bar v=t_1\shuffle\ldots\shuffle t_m=\abs{v}!\shuffle_{t\in\calT_n}t^{\abs{v}_t}&\mbox{and}&\hat v=\frac{\bar v}{\abs{v}!}=\shuffle_{t\in\calT_n}t^{\abs{v}_t}.
\end{eqnarray}
Let also $a$ be the injective linear endomorphism defined by $a(1_{\calT_n^*})=1_{\calT_n^*}$ and by $a(v)=(-1)^{\abs{v}}\tilde v$ ($v\in\calT_n^+$), being involutive and extended over $\ncs{\calA}{\calT_n}$ as follows
\begin{eqnarray}\label{antipode}
\forall S\in\ncs{\calA}{\calT_n},&
a(S)=\Sum_{w\in\calT_n^*}\scal{S}{w}a(w)=\Sum_{w\in\calT_n^*}(-1)^{\abs{w}}\scal{S}{w}\tilde w
\end{eqnarray}
and then
\begin{eqnarray}\label{antipode_shuffle}
\forall S,R\in\ncs{\calA}{\calT_n},&a(SR)=a(R)a(S),&a(S\shuffle R)=a(S)\shuffle a(R).
\end{eqnarray}
Moreover, if $S$ is such that $\scal{S}{1_{\calT_n^*}}=1$ then $a(S)$ is its inverse, $S^{-1}$, for $\tt conc$:
\begin{eqnarray}\label{inverse}
Sa(S)=a(S)S=1_{\calT_n^*}&\mbox{and then}&\forall L\in\ncs{\calL ie_{\calA}}{\calT_n}, a(e^L)=e^{-L}.
\end{eqnarray}

Ending this section, let us also consider the following product\footnote{\label{half}It is more general than the one used in \cite{these,orlando,legsloops,CM}
(denoted by $\circ$, for iterated integrals associated to polynomials) and is called \textit{half-shuffle}, denoted by $\prec$ in \cite{Loday} and
\textit{demi-shuffle} in \cite{Nakamura} (see Corollary \ref{Chensolution} below in which involve iterated integrals associated to series).},
$\halfshuffle$, defined for any $t\in\calT_n,R\in\ncs{\calA}{\calT_n},H\in\ncs{\calA}{T_n}$, by (see \cite{these,orlando,legsloops,livre,CM})
\begin{eqnarray}\label{halfshuffle}
1_{\calT_n^*}\halfshuffle(tH)=0&\mbox{and}&
(tH)\halfshuffle R=\left\{
\begin{array}{ccc}
tH&\mbox{if}&R=1_{\calT_n^*},\cr
t(H\shuffle R)&\mbox{if}&R\neq 1_{\calT_n^*}.
\end{array}\right.
\end{eqnarray}

\begin{example}\label{Excomposite1}
Using the second part of \eqref{halfshuffle} (with $t=t_{1,3},H=t_{1,2}$ and $R=t_{2,3}$)
\begin{eqnarray*}
(t_{1,3}t_{1,2})\halfshuffle t_{2,3}=t_{1,3}(t_{1,2}\shuffle t_{2,3})=t_{1,3}(t_{1,2}t_{2,3}+t_{2,3}t_{1,2})=t_{1,3}t_{1,2}t_{2,3}+t_{1,3}t_{2,3}t_{1,2}
\end{eqnarray*}
and, since $a\shuffle b^*=b^*ab^*$ ($a,b\in\calT_n$) then (with $t=t_{1,3},H=t_{1,2}^*$ and $R=t_{2,3}$)
\begin{eqnarray*}
(t_{1,3}t_{1,2}^*)\halfshuffle t_{2,3}=t_{1,3}(t_{1,2}^*\shuffle t_{2,3})=t_{1,3}(t_{1,2}^*t_{2,3}t_{1,2}^*)=t_{1,3}t_{1,2}^*t_{2,3}t_{1,2}^*.
\end{eqnarray*}
\end{example}

This product corresponds to the chronological product involved in quantum electrodynamic \cite{dyson}. It is not associative but satisfies the following identity
\begin{eqnarray}
\forall R,S,T\in\ncs{\calA}{\calT_n},&(R\halfshuffle S)\halfshuffle T=R\halfshuffle(S\halfshuffle T)+R\halfshuffle(T\halfshuffle S).
\end{eqnarray}
$(\ncs{\calA}{\calT_n},\halfshuffle)$ is a Zinbiel algebra \cite{Loday} and $\shuffle$ is a symmetrised product
of $\halfshuffle$, \textit{i.e.} for any $x,y\in\calT_n,u,v\in\calT_n^*$ and $R,S,T\in\ncs{\calA}{\calT_n}$,
\begin{eqnarray}\label{recursionbis}
(xu)\shuffle(yv)=(xu)\halfshuffle(yv)+(yv)\halfshuffle(xu)&\mbox{and}&R\shuffle S=R\halfshuffle S+S\halfshuffle R.
\end{eqnarray}

\begin{example}\label{Excomposite2}
For any $t_1,t_2\in\calT_n,w_1,w_2\in\calT_n^+$, by the recursion \eqref{recursion} one has
\begin{eqnarray*}
(t_1w_1)\shuffle(t_2w_2)=t_1(w_1\shuffle(t_2w_2))+t_2(w_2\shuffle(t_1w_1))
=(t_1w_1)\halfshuffle(t_2w_2)+(t_2w_2)\halfshuffle(t_1w_1),\cr
(t_1w_1^*)\shuffle(t_2w_2^*)=t_1(w_1^*\shuffle(t_2w_2^*))+t_2(w_2^*\shuffle(t_1w_1^*))=(t_1w_1^*)\halfshuffle(t_2w_2^*)+(t_2w_2^*)\halfshuffle(t_1w_1^*).
\end{eqnarray*}
\end{example}

The Zinbiel bialgebra and its dual are Loday's generalized bialgebras \cite{Loday}, \textit{i.e.}
\begin{eqnarray}
Z_{\halfshuffle}(\calT_n)=(\ncp{\calA}{\calT_n},\halfshuffle,1_{\calT_n^*},\Delta_{\tt conc}),&
Z_{\tt conc}(\calT_n)=(\ncp{\calA}{\calT_n},{\tt conc},1_{\calT_n^*},\Delta_{\halfshuffle}),
\end{eqnarray}
where $\Delta_{\halfshuffle}:\ncp{\calA}{\calT_n}\longrightarrow\ncp{\calA}{\calT_n}\otimes\ncp{\calA}{\calT_n}$
is defined by $\Delta_{\halfshuffle}1_{\calT_n^*}=1_{\calT_n^*}\otimes1_{\calT_n^*}$ and
\begin{itemize}
\item for any $t\in\calT_n,w\in\calT_n^*,\Delta_{\halfshuffle}t=t\otimes1_{\calT_n^*}$ and $\Delta_{\halfshuffle}(tw)=(\Delta_{\halfshuffle}t)(\Delta_{\shuffle}w)$,
\item for any $P\in\ncp{\calA}{\calT_n},
\Delta_{\halfshuffle}P=\scal{P}{1_{\calT_n^*}}1_{\calT_n^*}\otimes1_{\calT_n^*}+\Sum_{v\in\calT_n^+}\scal{P}{v}\Delta_{\halfshuffle}v$.
\end{itemize}

The co-product $\Delta_{\halfshuffle}$ is also extended, for any $S\in\ncs{\calA}{\calT_n}$, as follows
\begin{eqnarray}
\Delta_{\halfshuffle}S=\Sum_{w\in\calT_n^*}\scal{S}{w}\Delta_{\halfshuffle}w\in\ncs{A}{\calT_n^*\otimes\calT_n^*}.
\end{eqnarray}

\subsection{Diagonal series in concatenation-shuffle bialgebra}\label{diagonalseries}
In all the sequel, the characteristic series \cite{berstel} of $T_k$ and $\calT_k$ (resp. $T_k^*$ and $\calT_k^*$) are Lie polynomials, still denoted by $T_k$ and $\calT_k$ (resp. rational series $T_k^*$ and $\calT_k^*$), for $2\le k\le n$.

Let $\nabla S$ denote $S-1_{\calT_k^*}$ (resp. $S-1_{\calT_k^*}\otimes1_{\calT_k^*}$), for
$S\in\widehat{\ncp{A}{\calT_k}}$ (resp. $\ncp{A}{\calT_k}\hat\otimes\ncp{A}{\calT_k}$).
If $\scal{S}{1_{\calT_k^*}}=0$ (resp. $\scal{S}{1_{\calT_k^*}\otimes1_{\calT_k^*}}=0$)
then the Kleene star of $S$ is defined by
\begin{eqnarray}
S^*:=1+S+S^2+\cdots&\mbox{and}&S^+:=S^*S=SS^*
\end{eqnarray}
In the same way, for any $2\le k\le n$, the diagonal series is defined as follows
\begin{eqnarray}\label{diaser}
&&\calD_{\calT_k}=\calM_{\calT_k}^*\mbox{ and }\calD_{T_k}=\calM_{T_k}^*,\mbox{ where }
\calM_{\calT_k}=\sum_{t\in\calT_k}t\otimes t\mbox{ and }\calM_{T_k}=\sum_{t\in T_k}t\otimes t.
\end{eqnarray}
One also defines
\begin{eqnarray}\label{M}
\calM_{\calT_k}^+=\calD_{\calT_k}\calM_{\calT_k}=\calM_{\calT_k}\calD_{\calT_k}
&\mbox{and}&\calM_{T_k}^+=\calD_{T_k}\calM_{T_k}=\calM_{T_k}\calD_{T_k}
\end{eqnarray}
and, expanding \eqref{diaser}, one also has
\begin{eqnarray}\label{diaserbis}
&\calD_{\calT_k}=\Sum_{w\in\calT_k^*}w\otimes w=\Sum_{w\in\calT_k^*\atop\abs{w}=m,m\ge0}w\otimes w,&
\calD_{T_k}=\Sum_{w\in T_k^*}w\otimes w=\Sum_{w\in T_k^*\atop\abs{w}=m,m\ge0}w\otimes w.
\end{eqnarray}

If $S\in\widehat{\ncp{A}{\calT_k}}$ such that $\scal{S}{1_{\calT_k^*}}=0$ then $S^*$ is the unique solution of
$\nabla S=\calT_kS$ and $\nabla S=S\calT_k$.
In the same way, $\calD_{\calT_k}$ (resp. $\calD_{T_k}$) is the unique solution of
$\nabla S=\calM_{\calT_k}S$ and $\nabla S=S\calM_{\calT_k}$(resp. $\nabla S=\calM_{T_k}S$ and $\nabla S=S\calM_{T_k}$), for $2\le k\le n$.

Let us recall that $\calT_n=T_n\sqcup\calT_{n-1}$ and
\begin{itemize}
\item For any $a_1,\ldots,a_{n-1}\in\calA$, one has
\begin{eqnarray}
\Big(\sum_{i=1}^{n-1}a_it_{i,n}\Big)^*=\shuffle_{i=1}^{n-1}(a_it_{i,n})^*
&\mbox{and}&
T_n^*=\sum_{c_1,\ldots,c_{n-1}\ge0}\big(\shuffle_{i=1}^{n-1}t_{i,n}^{c_i}\big).
\end{eqnarray}
Thus, as $\calA$-modules, $\calT_{n-1}^m\shuffle T_n^*$ and $T_n^*\shuffle\calT_{n-1}^m$ are generated
by the series of the following form ($t_{i_1,j_1},\ldots,t_{i_m,j_m}$ are the letters in $\calT_{n-1}$)
\begin{eqnarray}\label{notation}
\Big(\sum_{c_{0,1},\ldots,c_{0,n-1}\ge0}\big(\shuffle_{i=1}^{n-1}t_{i,n}^{c_{0,i}}\big)\Big)
t_{i_1,j_1}\Big(\sum_{c_{1,1},\ldots,c_{1,n-1}\ge0}\big(\shuffle_{i=1}^{n-1}t_{i,n}^{c_{1,i}}\big)\Big)\phantom{.}&\cr
\ldots
t_{i_m,j_m}\Big(\sum_{c_{m,1},\ldots,c_{m,n-1}\ge0}\big(\shuffle_{i=1}^{n-1}t_{i,n}^{c_{j,i}}\big)\Big),&
\end{eqnarray}
and similarly for $\calT_{n-1}^*\shuffle T_n^m$ and $T_n^m\shuffle\calT_{n-1}^*$.

By Lazard factorization, \textit{i.e.} $\calT_n^*=T_n^*(\calT_{n-1}T_n^*)^*=(T_n^*\calT_{n-1})^*T_n^*$, or equi\-valently,
$\calT_n^*=\calT_{n-1}^*(T_n\calT_{n-1}^*)^*=(\calT_{n-1}^*T_n)^*\calT_{n-1}^*$ \cite{lothaire,viennotgerard} and
\begin{eqnarray}
&&\calT_n^*=\sum_{m\ge0}\calT_{n-1}^m\shuffle T_n^*=\sum_{m\ge0}T_n^*\shuffle\calT_{n-1}^m
\sum_{m\ge0}\calT_{n-1}^*\shuffle T_n^m=\sum_{m\ge0}T_n^m\shuffle\calT_{n-1}^*.
\end{eqnarray}
Then, by \eqref{diaserbis}, it follows that
\begin{eqnarray}\label{diaserter}
\calD_{\calT_n}=\sum_{m\ge0}\sum_{w\in\calT_{n-1}^m\shuffle T_n^*}w\otimes w.
\end{eqnarray}

\item Let the free Lie algebra $\ncp{\calL ie_{\calA}}{\calT_n}$ be endowed the basis $\{P_l\}_{l\in\Lyn\calT_n}$
over which are constructed, for the enveloping algebra ${\calU}(\ncp{\calL ie_{\calA}}{\calT_n})$, the PBW
basis $\{P_w\}_{w\in\calT_n^*}$ and its dual, $\{S_w\}_{w\in\calT_n^*}$ containing $\{S_l\}_{l\in\Lyn\calT_n}$
which is a pure transcendence basis of the shuffle algebra $\mathrm{Sh}_{\calA}(\calT_n)$ \cite{reutenauer}:
\begin{eqnarray}
\ncp{\calL ie_{\calA}}{\calT_n}=\mathrm{span}_{\calA}\{P_l\}_{l\in\Lyn\calT_n},
&\mathrm{Sh}_{\calA}(\calT_n)=\calA[\{S_l\}_{l\in\Lyn\calT_n}],\\
\forall l,\lambda\in\Lyn\calT_n,\scal{P_l}{S_\lambda}=\delta_{l,\lambda},
&\forall u,v\in\calT_n^*,\scal{P_u}{S_v}=\delta_{u,v}.
\end{eqnarray}
Homogenous in weight polynomials\footnote{\label{weight}For any $w\in\calT_n^*$, the weight of $P_w$ and $S_w$
are equal to the length of $w$, \textit{i.e.} $\abs{w}$.} $\{P_w\}_{w\in\calT_n^*},\{S_w\}_{w\in\calT_n^*}$
are constructed algorithmically and recursively ($P_{1_{\calT_n^*}}=1_{\calT_n^*}=S_{1_{\calT_n^*}}$) as follows \cite{lothaire}
\begin{eqnarray}\label{basisP}
\left\{\begin{array}{ll}
P_t=t,&\mbox{for }t\in\calT_n,\\
P_l=[P_{l_1},P_{l_2}],&\mbox{for }l\in\Lyn\calT_n\setminus\calT_n,\ st(l)=(l_1,l_2),\\
P_w=P_{l_1}^{i_1}\ldots P_{l_k}^{i_k},
&{\displaystyle\mbox{for }w=l_1^{i_1}\ldots l_k^{i_k},\mbox{ with}\hfill
\atop\displaystyle l_1,\ldots,l_k\in\Lyn\calT_n,\ l_1\succ\ldots\succ l_k,}
\end{array}\right.
\end{eqnarray}
and, by duality, {\it i.e.} $\scal{P_u}{S_v}=\delta_{u,v}$ (for $u,v\in\calT_n^*$) \cite{reutenauer}
\begin{eqnarray}\label{basisS}
\left\{\begin{array}{ll}
S_t=t,&\mbox{for }t\in\calT_n,\\
S_l=tS_{l'},
&\mbox{for }l=tl'\in\Lyn\calT_n,\\
S_w=\Frac{S_{l_1}^{\shuffle i_1}\shuffle\ldots\shuffle S_{l_k}^{\shuffle i_k}}{i_1!\ldots i_k!},
&{\displaystyle\mbox{for }w=l_1^{i_1}\ldots l_k^{i_k},\mbox{ with}\hfill
\atop\displaystyle l_1,\ldots,l_k\in\Lyn\calT_n,l_1\succ\ldots\succ l_k.}
\end{array}\right.
\end{eqnarray}

\begin{remark}
Or equivalently, $P_w=P_{l_1}\ldots P_{l_k}$ and $S_w=S_{l_1}\shuffle\ldots\shuffle S_{l_k}$,
for $w=l_1\ldots l_k$ with $l_1\succeq\ldots\succeq l_k$ and $l_1,\ldots,l_k\in\Lyn\calT_n$.
\end{remark}
\end{itemize}

By \eqref{diaser}, one gets in the bialgebra $H_{\shuffle}(\calT_k)$ \cite{reutenauer} (and also in $H_{\shuffle}(T_k)$)
\begin{eqnarray}\label{D}
\qquad\qquad
\calD_{\calT_k}
&=&\sum_{v\in\calT_k^*}S_v\otimes P_v=\sum_{i_1,\ldots,i_m\ge0\atop{l_1,\ldots,l_k\in\Lyn\calT_k\atop l_1\succ\ldots\succ l_m,m\ge0}}
\frac{S_{l_1}^{i_1}\shuffle\ldots\shuffle S_{l_k}^{i_k}}{i_1!\ldots i_m!}\otimes P_{l_1}^{i_1}\ldots P_{l_m}^{i_m},\\
\log\calD_{\calT_k}&=&\sum_{w\in\calT_k^*}w\otimes\pi_1(w),\label{logD}
\end{eqnarray}
where $\pi_1(w)$ is the projection on the set of primitive elements (see also \eqref{M}):
\begin{eqnarray}\label{pi1}
\pi_1(w)=\sum_{m\ge1}\frac{(-1)^{m-1}}m\sum_{u_1,\ldots,u_m\in\calT_k^+}\scal{w}{u_1\shuffle\ldots\shuffle u_m}u_1\ldots u_m.
\end{eqnarray}

\subsection{More about diagonal series in concatenation-shuffle bialgebra and in a Loday's generalized bialgebra}\label{moreabout}
One defines the adjoint endomorphism, as being a derivation of $\ncs{\calL ie_{\calA}}{\calT_n}$,
for any $S\in\ncs{\calL ie_{\calA}}{\calT_n}$, as follows
\begin{eqnarray}\label{adjoint}
\ad_S:\ncs{\calL ie_{\calA}}{\calT_n}\longrightarrow\ncs{\calL ie_{\calA}}{\calT_n},&
R\longmapsto\ad_SR=[S,R]
\end{eqnarray}
determining the so-called adjoint representation of Lie algebra \cite{bourbaki,dixmier}:
\begin{eqnarray}\label{adjoint2}
\ad:\ncs{\calL ie_{\calA}}{\calT_n}\longrightarrow\operatorname{End}(\ncs{\calL ie_{\calA}}{\calT_n}),&
S\longmapsto\ad_S.
\end{eqnarray}
To $\ad$ corresponds to the right normed bracketing (bracketing from right to left) which is the injective linear endomorphism of $\ncs{\calA}{\calT_n}$
defined by\footnote{In \cite{bourbaki}, $r$ is denoted by $\varphi$ and is proved to be an isomorphism of Lie sub algebras.}
$r(1_{\calT_n^*})=0$ and, for any $t_1,\ldots,t_{m-1},t_m\in\calT_n$, by \cite{bourbaki,reutenauer}
\begin{eqnarray}\label{rho}
r(t_1\ldots t_{m-1}t_m)=[t_1,[\ldots,[t_{m-1},t_m]\ldots]]=\ad_{t_1}\circ\ldots\circ\ad_{t_{m-1}}t_m.
\end{eqnarray}

\begin{remark}
\begin{enumerate}
\item The coadjoint endomorphism is defined as follows
\begin{eqnarray*}
\forall S\in\ncs{\calL ie_{\calA}}{\calT_n},&
\mathrm{coad}_S:\ncs{\calL ie_{\calA}}{\calT_n}\longrightarrow\ncs{\calL ie_{\calA}}{\calT_n},&R\longmapsto\mathrm{coad}_SR=[R,S].
\end{eqnarray*}

\item The adjoint endomorphism of $r$, denoted by $\check r$, is defined by \cite{reutenauer}
\begin{eqnarray*}
\sum_{w\in\calT_n^*}w\otimes r(w)=\sum_{w\in\calT_n^*}\check r(w)\otimes w,
\end{eqnarray*}
or equivalently, $\scal{r(v)}{w}=\scal{v}{\check r(w)}$ ($v,w\in\calT_n^*$) satisfying
\begin{eqnarray*}
\forall w\in\calT_n^+,&\abs{w}w=\Sum_{u,v\in\calT_n^*,uv=w}\check r(w)\shuffle w.
\end{eqnarray*}
It can be also defined recursively by $\check r(1_{\calT_n^*})=0$ and
\begin{eqnarray*}
\forall t_1,t_2\in\calT_n,w\in\calT_n^*,&\check r(t_1)=t_1,&\check r(t_1wt_2)=t_1\check r(wt_2)-t_2\check r(t_1w).
\end{eqnarray*}
\end{enumerate}
\end{remark}

With Notations in \eqref{barv}, let $g$ be the endomorphism of $(\ncp{\calA}{\calT_n},{\tt conc})$
defined by $g(1_{\calT_n^*})=1_{\calT_n^*}$ and, for any $w\in\calT_n^+$, by $g(w)=a(w)$ such that
\begin{eqnarray}\label{g}
\forall t\in\calT_n,&g(w)(t)=-ta(w)=a(wt).
\end{eqnarray}
Similarly, let us also associate $r$ to $f:(\ncp{\calA}{\calT_n},{\tt conc})\longrightarrow(\operatorname{End}(\ncs{\calL ie_{\calA}}{\calT_n}),\circ)$
defined by $f(1_{\calT_n^*})=1_{\mathrm{\operatorname{End}(\ncs{\calL ie_{\calA}}{\calT_n})}}$ and, for any $t_1,\ldots,t_{m-1}\in\calT_n$, as follows
\begin{eqnarray}\label{f}
f(t_1\ldots t_{m-1})=\ad_{t_1}\circ\ldots\circ\ad_{t_{m-1}}.
\end{eqnarray}

\begin{example}\label{example}
Denoting, for any $a,b\in\ncs{\calL ie_{\calA}}{\calT_n}$ and $j>0$, $\ad_a^0b=b$ and \cite{bourbaki,lothaire}
\begin{eqnarray*}
\ad^j_ab=[a,\ad^{j-1}_ab]=\sum_{i=0}^j(-1)^i{j\choose i}a^iba^{j-i}=r(a^jb)=f(a^j)(b),
\end{eqnarray*}
\begin{enumerate}
\item one has, by the ordering \eqref{orderLyndon} and the dual bases in \eqref{basisP}--\eqref{basisS},
for any $t\in T_n$ and $x\in\calT_{n-1}$ and $j\ge0$, $t\prec x$ and $t^jx\in\Lyn\calT_n$ and then, by
induction, $P_{t^jx}=\ad_t^jx=f(t^j)(x)$ and $S_{t^jx}=t^jx$.

\item for $\calT_3=\{t_{1,2},t_{1,3},t_{2,3}\}$, if $t_{1,2}\prec t_{1,3}\prec t_{2,3}$ then $t_{1,2}^jt_{i,3}\in\Lyn\calT_3$ and then
$P_{t_{1,2}^jt_{i,3}}=\ad_{t_{1,2}}^jt_{i,3}=f(t_{1,2}^j)(t_{i,3})$ and $S_{t_{1,2}^jt_{i,3}}=t_{1,2}^jt_{i,3},k\ge0,i=1$ or $2$.
\end{enumerate}
\end{example}

Now, by the partitions of $\calT_n$, let $\calI_n$ be the sub Lie algebra of $\ncp{\calL ie_{\calA}}{\calT_n}$ generated
by $\{\ad^k_{-T_n}t\}_{t\in\calT_{n-1}}^{k\ge0}$. By the Lazard's elimination \cite{bourbaki,lazard}, one has
\begin{itemize}
\item as Lie algebras and then by duality,
\begin{eqnarray}\label{produitsemidirecte}
\ncp{\calL ie_{\calA}}{\calT_n}=\ncp{\calL ie_{\calA}}{T_n}\ltimes{\calI}_n,
&\ncp{\calL ie_{\calA}}{\calT_n}^{\vee}=\ncp{\calL ie_{\calA}}{T_n}^{\vee}\rtimes{\calI}_n^{\vee},
\end{eqnarray}
\item as being modules and then by duality,
\begin{eqnarray}\label{sommedirecte}
\ncp{\calL ie_{\calA}}{\calT_n}=\ncp{\calL ie_{\calA}}{T_n}\oplus{\calI}_n,
&\ncp{\calL ie_{\calA}}{\calT_n}^{\vee}=\ncp{\calL ie_{\calA}}{T_n}^{\vee}\oplus{\calI}_n^{\vee},
\end{eqnarray}
\item and, by taking the enveloping algebras \cite{JurisichWilson} and then by duality,
\begin{eqnarray}\label{sommedirecte2}
&\calU(\ncp{\calL ie_{\calA}}{\calT_n})=\calU(\ncp{\calL ie_{\calA}}{T_n})\calU({\calI}_n),\\
&\calU(\ncp{\calL ie_{\calA}}{\calT_n})^{\vee}=\calU(\ncp{\calL ie_{\calA}}{T_n})^{\vee}\shuffle\calU({\calI}_n)^{\vee}.
\end{eqnarray}
\end{itemize}

$\calI_n$ can be also obtained as image by $r$ of the free Lie algebra generated by $(-T_n)^*\calT_{n-1}$,
on which the restriction of $r$ is an isomorphism of free Lie algebras.

In other terms, let $Y_{T_n^*\calT_{n-1}}:=\{y_w\}_{w\in T_n^*\calT_{n-1}}$ be the new alphabet in which letters $y_w$ are encoded by words $w$ in $T_n^*\calT_{n-1}$.
Then, with this  lphabet and the recursive constructions given in \eqref{basisP}--\eqref{basisS}, the families $\{P_w\}_{w\in Y_{T_n^*\calT_{n-1}}^*}$ and $\{S_w\}_{w\in Y_{T_n^*\calT_{n-1}}^*}$ form linear bases of $\calU(\ncp{\calL ie_{\calA}}{Y_{T_n^*\calT_{n-1}}})$ and $\calU(\ncp{\calL ie_{\calA}}{Y_{T_n^*\calT_{n-1}}})^{\vee}$, respectively, and their images form linear bases of $\calU(\calI_n)$ and $\calU(\calI_n)^{\vee}$.

\begin{example}\label{examplesimple}
For $X=\{x_0,x_1\}=\{x_0\}\sqcup\{x_1\}$ and $Y_{x_0^*x_1}=\{y_w\}_{w\in x_0^*x_1}$, this construction is classically illustrated in \cite{lothaire}. The bases $\{P_w\}_{w\in Y^*}$ and $\{S_w\}_{w\in Y^*}$ (or $\{P_w\}_{w\in Y_{x_0^*x_1}^*}$ and $\{S_w\}_{w\in Y_{x_0^*x_1}^*}$) are constructed according to \eqref{basisP}--\eqref{basisS}. In particular,
$P_{x_0^{s_1-1}x_1\cdots x_0^{s_r-1}x_1}
=(\ad_{x_0}^{s_1-1}x_1)\cdots\allowbreak(\ad_{x_0}^{s_r-1}x_1)
=r(x_0^{s_1-1}x_1)\cdots r(x_0^{s_r-1}x_1)$, for $s_1>\cdots>s_r$.
Note also that each letter $y_{x_0^{s-1}x_1}$ of $Y_{x_0^*x_1}$ can be also encoded by the letter $y_s$ of the alphabet $Y=\{y_s\}_{s\ge1}$ and then each word $x_0^{s_1-1}x_1\cdots x_0^{s_r-1}x_1$ in $X^*$ correspnds to the word $y_{s_1}\cdots y_{s_r}$ in $Y^*$ (see \cite{CM}).
\end{example}

\begin{example}
For $\calT_3=\{t_{1,2},t_{1,3},t_{2,3}\}=T_3\sqcup\calT_2$, where $T_3=\{t_{1,3},t_{2,3}\}$ and $\calT_2=\{t_{1,2}\}$, let $T_3$ (resp. $\calT_2$) play the r\^ole of $\{x_0\}$ (resp. $\{x_1\}$) of Example \ref{examplesimple}. In this case, the free monoid $\{t_{1,3},t_{2,3}\}^*$ (equipping the set of Lyndon words $\Lyn(\{t_{1,3},t_{2,3}\})$) plays the r\^ole of $x_0^*$.
More generally, for the partition of the alphabet $\calT_n$, $T_n$ (resp. $\calT_{n-1}$) plays the r\^ole of $\{x_0\}$ (resp. $\{x_1\}$) of Example \ref{examplesimple}. In this case, the free monoid $T_n^*$ (equipping $\Lyn(T_n)$) plays the r\^ole of $x_0^*$.
\end{example}

\begin{definition}\label{dual_bases}
For any $k\ge1$, let $\hat T_n^k:=\{\hat v\in T_n^*,\abs{v}=k\}$. One defines
\begin{eqnarray*}
\calB&:=&\{\ad_{-T_n}^{k_1}t_1\ldots\ad_{-T_n}^{k_p}t_p\}_{t_1,\ldots,t_p\in\calT_{n-1}}^{k_1,\ldots,k_p\ge0,p\ge1},\cr
\calB^{\vee}&:=&\{(-t_1T_n^{k_1})\shuffle\cdots\shuffle(-t_pT_n^{k_p})\}_{t_1,\ldots,t_p\in\calT_{n-1}}^{k_1,\ldots,k_p\ge0,p\ge1},\cr
\hat\calB&:=&\{-t_1(\hat T_n^{k_1}\shuffle(\cdots\shuffle(-t_p\hat T_n^{k_p})\ldots))\}_{t_1,\ldots,t_p\in\calT_{n-1}}^{k_1,\ldots,k_p\ge0,p\ge1}.
\end{eqnarray*}
\end{definition}

\begin{remark}\label{identites}
For any $k\ge0$, expanding $T_n^k$ and $\hat T_n^k$, it is immediate that
\begin{eqnarray*}
\calB
&=&\{(-1)^{\abs{v_1\ldots v_k}}r(v_1t_1)\cdots r(v_kt_p)\}_{v_1,\ldots,v_p\in T_n^*\atop t_1,\ldots,t_p\in\calT_{n-1}}^{p\ge1},\cr
\calB^{\vee}
&=&\{(-t_1u_1)\halfshuffle(\cdots\halfshuffle(-t_pu_p)\ldots))\}_{u_1,\ldots,u_p\in T_n^*\atop t_1,\ldots,t_p\in\calT_{n-1}}^{p\ge1}\cr
&=&\{a(v_1t_1)\halfshuffle(\cdots\halfshuffle(v_pt_p)\ldots))\}_{v_1,\ldots,v_p\in T_n^*\atop t_1,\ldots,t_p\in\calT_{n-1}}^{p\ge1},\cr
\hat\calB
&=&\{-t_1(\hat v_1\shuffle(\cdots\shuffle(-t_p\hat v_p)\ldots))\}_{v_1\in T_n^{k_1},\ldots,v_p\in T_n^{k_p}\atop t_1,\ldots,t_p\in\calT_{n-1}}^{k_1,\ldots,k_p\ge0,p\ge1}\cr
&=&\{(-t_1\hat v_1)\halfshuffle(\cdots\halfshuffle(-t_p\hat v_p)\ldots))\}_{v_1\in T_n^{k_1},\ldots,v_p\in T_n^{k_p}\atop t_1,\ldots,t_p\in\calT_{n-1}}^{k_1,\ldots,k_p\ge0,p\ge1}.
\end{eqnarray*}
\end{remark}

Furthermore, according to \cite{Loday}, as Lie algebra, $\calI_n$ is obviously a Leibniz algebra generated by $\{\ad^k_{-T_n}t\}_{t\in\calT_{n-1}}^{k\ge0}$
and $\calI_n^{\vee}$ is the Zinbiel subalgebra of $(\ncp{\calA}{\calT_n},\halfshuffle)$ generated by $\{-tT_n^k\}_{t\in\calT_{n-1}}^{k\ge0}$.
These constitute the Zinbiel bialgebra $Z_{\halfshuffle}(\calT_n)$.

\begin{lemma}\label{factorization}
Let $\{b_i\}_{i\ge0}$ and $\{\check b_i\}_{i\ge0}$ (resp. $\{c_i\}_{i\ge0}$ and $\{\check c_i\}_{i\ge0}$)
be a pair of (non necessary ordered) dual linear bases of $\calU({\calI}_n)$ and $\calU({\calI}_n)^{\vee}$
(resp. $\calU(\ncp{\calL ie_{\calA}}{T_n})$ and $\calU(\ncp{\calL ie_{\calA}}{\calT_n})^{\vee}$).
Then the diagonal series is factorized as follows
\begin{eqnarray*}
\calD_{\calT_n}=\Big(\sum_{i\ge0}\check c_i\otimes c_i\Big)\Big(\sum_{i\ge0}\check b_i\otimes b_i\Big),
\end{eqnarray*}
\end{lemma}

\begin{proof}
The Lazard's elimination described in \eqref{produitsemidirecte}--\eqref{sommedirecte2}, and $\{r(P_w)\}_{w\in Y_{T_n^*\calT_{n-1}}^*}$
and $\{r(S_w)\}_{w\in Y_{T_n^*\calT_{n-1}}^*}$ (resp. $\{P_w\}_{w\in T_n^*}$ and $\{S_w\}_{w\in T_n^*}$), generating freely
$\calU({\calI}_n)$ and $\calU({\calI}_n)^{\vee}$ (resp. $\calU(\ncp{\calL ie_{\calA}}{T_n})$ and $\calU(\ncp{\calL ie_{\calA}}{\calT_n})^{\vee}$),
yield the expected result.
\end{proof}

\begin{proposition}[dual bases]\label{dual_algebras}
\begin{enumerate}
\item $\scal{a(v_1t_1)}{r(v_2t_2)}=\delta_{v_1,v_2}\delta_{t_1,t_2}$,
for $v_1,v_2\in T_n^*$ and $t_1,t_2\in\calT_{n-1}$. Hence, as modules,
$\calI_n\simeq(\mathrm{span}_{\calA}\{r(vt)\}_{v\in T_n^*\atop t\in\calT_{n-1}},[,])$
and, by duality, $\calI_n^{\vee}\simeq
\allowbreak(\mathrm{span}_{\calA}\{-tu\}_{u\in T_n^*\atop t\in\calT_{n-1}},\shuffle)
\simeq(\mathrm{span}_{\calA}\{a(vt)\}_{v\in T_n^*\atop t\in\calT_{n-1}},\halfshuffle)$.

\item $\scal{a(v_1t_1)\halfshuffle(\cdots\halfshuffle a(v_pt_p)\ldots))}{r(v_1t_1)\ldots r(v_pt_p)}=1$,
for $v_1,\ldots,v_p\in T_n^*$ and $t_1,\ldots,t_p\in\calT_{n-1}$. Hence,
\begin{eqnarray*}
\calU(\calI_n)&\simeq&\mathrm{span}_{\calA}\{(-1)^{\abs{v_1\ldots v_k}}r(v_1t_1)\cdots
r(v_pt_p)\}_{v_1,\ldots,v_p\in T_n^*\atop t_1,\ldots,t_p\in\calT_{n-1}}^{p\ge1},\\
\calU(\calI_n)^{\vee}&\simeq&\mathrm{span}_{\calA}\{a(u_1t_1)\shuffle\cdots\shuffle a(u_pt_p)\}_{u_1,\ldots,u_p\in T_n^*\atop t_1,\ldots,t_p\in\calT_{n-1}}^{p\ge1}\cr
&\simeq&\mathrm{span}_{\calA}\{a(v_1t_1)\halfshuffle(\cdots\halfshuffle a(v_pt_p)\ldots))\}_{v_1,\ldots,v_p\in T_n^*\atop t_1,\ldots,t_p\in\calT_{n-1}}^{p\ge1}.
\end{eqnarray*}

\item $T_n^*\calB$ (resp. $T_n^*\shuffle\calB^{\vee}$) is linear basis of
$\calU(\ncp{\calL ie_{\calA}}{\calT_n})$ (resp. $\calU(\ncp{\calL ie_{\calA}}{\calT_n})^{\vee}$).
\end{enumerate}
\end{proposition}

\begin{proof}
\begin{enumerate}
\item Let $u=\tilde v\in T_n^*$. By \eqref{barv}, $-tv=(-1)^{\abs{u}}a(ut)$ and then
$\{\ad^k_{-T_n}t\}_{t\in\calT_{n-1}}^{k\ge0}\allowbreak=r((-T_n)^*\calT_{n-1})=\{(-1)^{\abs{v}}r(vt)\}_{v\in T_n^*,t\in\calT_{n-1}}$ and $\{-tT_n^k\}_{t\in\calT_{n-1}}^{k\ge0}\allowbreak=-\calT_{n-1}T_n^*\allowbreak=\{a(ut)\}_{u\in T_n^*\atop t\in\calT_{n-1}}$. By \eqref{recursion} and \eqref{halfshuffle}, it follows then the expected result.

\item Since $\{(-1)^{\abs{v}}r(vt)\}_{v\in T_n^*\atop t\in\calT_{n-1}}$ is $\calA$-linearly free and any $r(vt)$ is primitive for $\Delta_{\shuffle}$ (by definition) then, basing on previous item and using PBW and CQMM theorems, $\calB$ and $\calB^{\vee}$ generate freely
$\calU(\calI_n)$ and $\calU(\calI_n)^{\vee}$. It follows then the expected results (see also Remark \ref{identites}).

\item It is a consequence of the Lazard's elimination described in \eqref{produitsemidirecte}--\eqref{sommedirecte2}.
\end{enumerate}
\end{proof}

\begin{definition}\label{Lambda}
\begin{enumerate}
\item Let $\lambda_r:(\ncp{\calA}{\calT_{n-1}},{\tt conc})\longrightarrow(\ncs{\calA}{\calT_n},{\tt conc})$ be the
${\tt conc}$-morphism and let $\lambda_l$ and $\hat\lambda_l$ be the morphisms, from the Cauchy algebra $(\ncp{\calA}{\calT_{n-1}},{\tt conc})$ to the Zinbiel
algebra $(\ncs{\calA}{\calT_n},\halfshuffle)$, defined over letters by
\begin{eqnarray*}
\lambda_r(t)=r((-T_n)^*t)=\Sum_{v\in T_n^*}(-1)^{\abs{v}}r(vt),&\\
\lambda_l(t)=a((-T_n)^*t)=\Sum_{v\in T_n^*}(-1)^{\abs{v}}a(vt),&\hat\lambda_l(t)=\Sum_{v\in T_n^*}(-1)^{\abs{v}}a(\hat vt).
\end{eqnarray*}

\item Let $\lambda,\hat\lambda:(\ncp{\calA}{\calT_{n-1}}\hat\otimes\ncp{\calA}{\calT_{n-1}},{}_{{\tt conc}}\otimes_{\tt conc})
\longrightarrow(\ncp{\calA}{\calT_n}{}_{\halfshuffle}\hat\otimes_{\tt conc}\ncp{\calA}{\calT_n},{}_{\halfshuffle}\otimes_{\tt conc})$
be the morphisms of algebras\footnote{Using ${}_{\halfshuffle}\otimes_{\tt conc}$ (resp. ${}_{{\tt conc}}\otimes_{\tt conc}$) with
$\halfshuffle$ (resp. ${\tt conc}$) on the left and ${\tt conc}$ on the right of $\otimes$.
For convenience, they are also denoted by $\otimes$.} defined over letters by
\begin{eqnarray*}
&\lambda(t\otimes t)=\mathrm{diag}(\lambda_l\otimes\lambda_r)(t\otimes t)=\Sum_{v\in T_n^*}a(vt){}_{\halfshuffle}\otimes_{\tt conc}r(vt),\cr
&\hat\lambda(t\otimes t)=\mathrm{diag}(\hat\lambda_l\otimes\lambda_r)(t\otimes t)=\Sum_{v\in T_n^*}a(\hat vt){}_{\halfshuffle}\otimes_{\tt conc}r(vt).
\end{eqnarray*}
\end{enumerate}
\end{definition}

\begin{proposition}\label{expression}
\begin{enumerate}
\item With the notations in \eqref{basisP}--\eqref{basisS} and \eqref{g}--\eqref{f}, one has  (using the decreasing lexicographical order product)
\begin{eqnarray*}
\lambda
=(g\otimes f)\calD_{T_n}
=\sum_{w\in T_n^*}g(w)\otimes f(w)
=\Prod_{l\in\Lyn T_n}^{\searrow}e^{g(S_l)\otimes f(P_l)}
=\Prod_{l\in\Lyn T_n}^{\searrow}e^{a(S_l)\otimes ad_{P_l}}.
\end{eqnarray*}

\item With the notations in Proposition \ref{dual_algebras}, one also has 
\begin{eqnarray*}
\lambda(\calM_{\calT_{n-1}}^+)=(\lambda(\calM_{\calT_{n-1}}))^+,&\mbox{where}&
\lambda(\calM_{\calT_{n-1}})=\Sum_{v\in T_n^*,t\in\calT_{n-1}}a(vt){}_{\halfshuffle}\otimes_{\tt conc}r(vt),\cr
\hat\lambda(\calM_{\calT_{n-1}}^+)=(\hat\lambda(\calM_{\calT_{n-1}}^*))^+,&\mbox{where}&
\hat\lambda(\calM_{\calT_{n-1}})=\Sum_{v\in T_n^*,t\in\calT_{n-1}}a(\hat vt){}_{\halfshuffle}\otimes_{\tt conc}r(vt),
\end{eqnarray*}
and explicitly:
\begin{eqnarray*}
\lambda(\calM_{\calT_{n-1}}^+)=\Sum_{k\ge1}\Sum_{v_1,\ldots,v_k\in T_n^*\atop t_1,\ldots,t_k\in\calT_{n-1}}
a(v_1t_1)\halfshuffle(\cdots\halfshuffle a(v_kt_k)\ldots))\otimes r(v_1t_1)\ldots r(v_kt_k),\cr
\hat\lambda(\calM_{\calT_{n-1}}^+)=\Sum_{k\ge1}\Sum_{v_1,\ldots,v_k\in T_n^*\atop t_1,\ldots,t_k\in\calT_{n-1}}
a(\hat v_1t_1)\halfshuffle(\cdots\halfshuffle a(\hat v_kt_k)\ldots))\otimes r(v_1t_1)\ldots r(v_kt_k).
\end{eqnarray*}
\end{enumerate}
\end{proposition}

\begin{proof}
\begin{enumerate}
\item By \eqref{antipode} (resp. \eqref{adjoint}), the restriction of $g$ (resp. $f$) on $\ncp{\mathrm{Sh}_{\calA}}{T_n}$ (resp. $\ncp{\calL ie_{\calA}}{T_n}$) is a morphism of algebras.
Then $\lambda(t\otimes t)=((g\otimes f)\calD_{T_n})(t\otimes t)$, for $t\in\calT_{n-1}$.

\item By the previous item, one deduces the expected expressions for $\lambda(\calM_{\calT_{n-1}})$ and
$\lambda(\calM_{\calT_{n-1}}^+)$ (and similarly for $\hat\lambda(\calM_{\calT_{n-1}})$ and $\hat\lambda(\calM_{\calT_{n-1}}^+)$:
\begin{eqnarray*}
\lambda(\calM_{\calT_{n-1}})=&\lambda\Big(\Sum_{t\in\calT_{n-1}}t\otimes t\Big)&=\Sum_{t\in\calT_{n-1}}\lambda(t\otimes t),\cr
\lambda(\calM_{\calT_{n-1}}^+)=&(\lambda(\calM_{\calT_{n-1}}))^+&=\Big(\Sum_{v\in T_n^*,t\in\calT_{n-1}}a(vt){}_{\halfshuffle}\otimes_{\tt conc}r(vt)\Big)^+.
\end{eqnarray*}
\end{enumerate}
\end{proof}

\begin{theorem}[factorized diagonal series]\label{diagonal}
With the bases in \eqref{basisP}--\eqref{basisS}, Definitions \ref{dual_bases}--\ref{Lambda}, Lemma \ref{factorization} and Propositions \ref{dual_algebras}--\ref{expression}, 
the diagonal series $\calD_{\calT_n}$ is factorized, using the decreasing lexicographical order product, as follows
\begin{eqnarray*}
\calD_{\calT_n}&=&\prod_{l\in\Lyn\calT_n}^{\searrow}e^{S_l\otimes P_l}
=\calD_{\calT_{n-1}}\Big(\Prod_{l=l_1l_2\atop{l_2\in\Lyn\calT_{n-1},l_1\in\Lyn T_n}}^{\searrow}e^{S_l\otimes P_l}\Big)\calD_{T_n},\cr
\calD_{\calT_n}&=&\calD_{T_n}\Big(1_{\calT_n^*}\otimes1_{\calT_n^*}\\
&+&\sum_{k\ge1}\sum_{v_1,\ldots,v_k\in T_n^*\atop t_1,\ldots,t_k\in\calT_{n-1}}
a(v_1t_1)\halfshuffle(\cdots\halfshuffle a(v_kt_k)\ldots))\otimes r(v_1t_1)\ldots r(v_kt_k)\Big).
\end{eqnarray*}
\end{theorem}

Any $S\in\ncs{\calA}{\calT_k}$ can be expressed as image by $S\otimes\mathrm{Id}$ of $\calD_{\calT_k}$ (resp. $\log\calD_{\calT_k}$) by (and also in $\ncs{\calA}{T_k}$)
\begin{eqnarray}
S
&=&\Big(\sum_{w\in T_k^*}\scal{S}{w}w\Big)\\
&\times&\Big(\sum_{v_1,\ldots,v_s\in T_k^*,k\ge0\atop t_1,\ldots,t_s\in\calT_{k-1}}
\scal{S}{a(v_1t_1)\halfshuffle\cdots\halfshuffle a(v_st_s)}r(v_1t_1)\ldots r(v_st_s)\Big),\cr
\log S&=&\sum_{w\in\calT_k^*}\scal{S}{w}\pi_1(w).
\end{eqnarray}
If $S$ is grouplike then it can be put in the MRS form \cite{reutenauer} and, by \eqref{inverse}, since $S^{-1}=a(S)$ then (and similarly in $\ncs{\calA}{T_k}$):
\begin{eqnarray}
&&S=\sum_{w\in\calT_k^*}\scal{S}{S_w}P_w=\prod_{l\in\Lyn\calT_k}^{\searrow}e^{\scal{S}{S_l}P_l}\quad{\mbox{(decreasing lexicographical}\atop\mbox{ordered product).}}\label{grouplike}\\
&&S^{-1}=\prod_{l\in\Lyn\calT_k}^{\nearrow}a(e^{\scal{S}{S_l}P_l})=\prod_{l\in\Lyn\calT_k}^{\nearrow}e^{-\scal{S}{S_l}P_l}
{\mbox{(increasing lexicographical}\atop\mbox{ordered product).}}\label{a(S)}
\end{eqnarray}

\begin{proposition}
In the Loday's generalized bialgebra $Z_{\halfshuffle}(\calT_k)$ (and also in $Z_{\halfshuffle}(T_k)$), 
\begin{eqnarray*}
\shuffle\limits_{i=1}^mu_i&=&\sum_{\sigma\in\mathfrak{S}_m}u_{\sigma(1)}\halfshuffle(\ldots(\halfshuffle u_{\sigma(m)})),\\
\shuffle\limits_{i=1}^mS_{l_i}&=&\sum_{\sigma\in\mathfrak{S}_m}S_{l_{\sigma(1)}}\halfshuffle(\ldots(\halfshuffle S_{l_{\sigma(m)}})).
\end{eqnarray*}
\end{proposition}

\begin{proof}
These results are obvious for $m=1$. Suppose it holds, for any $1\le i\le m-1$. Next, for $u_i=t_iu_i'\in\calT_k^+$ and $l_i=t_il_i'\in\Lyn\calT_k$, by induction hypothesis
and by \eqref{recursion} and \eqref{halfshuffle} and \eqref{basisS}, one successively obtains 
\begin{eqnarray*}
\shuffle\limits_{i=1}^mu_i
&=&\Sum_{\sigma\in\mathfrak{S}_m}t_{\sigma(m)}(u'_{\sigma(m)}\shuffle\shuffle\limits_{i=1}^{m-1}u_{\sigma(i)})
=\Sum_{\sigma\in\mathfrak{S}_m}u_{\sigma(m)}\halfshuffle(\shuffle\limits_{i=1}^{m-1}u_{\sigma(i)})\\
&=&\Sum_{\sigma\in\mathfrak{S}_m}u_{\sigma(m)}\halfshuffle\Sum_{\rho\in\mathfrak{S}_{m-1}}
u_{\rho\circ\sigma(1)}\halfshuffle(\ldots(\halfshuffle u_{\rho\circ\sigma(m-1)})\ldots),\cr
\shuffle\limits_{i=1}^mS_{l_i}
&=&\Sum_{\sigma\in\mathfrak{S}_m}t_{\sigma(m)}(S_{l'_{\sigma(m)}}\shuffle\shuffle\limits_{i=1}^{m-1}S_{l_{\sigma(i)}})
=\Sum_{\sigma\in\mathfrak{S}_m}S_{l_\sigma(m)}\halfshuffle(\shuffle\limits_{i=1}^{m-1}S_{l_{\sigma(i)}})\\
&=&\Sum_{\sigma\in\mathfrak{S}_m}S_{l_\sigma(m)}\halfshuffle\Sum_{\rho\in\mathfrak{S}_{m-1}}
S_{l_{\rho\circ\sigma(1)}}\halfshuffle(\ldots(\halfshuffle S_{l_{\rho\circ\sigma(m-1)}})\ldots).
\end{eqnarray*}
For any $\sigma\in\mathfrak{S}_m,\rho\in\mathfrak{S}_{m-1}$, $\rho$ belongs also $\mathfrak{S}_m$,
for which $\rho(m)=m$ and then $\rho\circ\sigma\in\mathfrak{S}_m$. It follows then the expected results.
\end{proof}

\section{Solutions of universal differential equation}\label{universalequation}
\subsection{Iterated integrals and Chen series}\label{Iterated integrals}
In all the sequel, $\calV$ is the simply connected manifold on $\C^n$.
The pushforward (resp. pullback) of any diffeomorphism $g$ on $\calV$ is denoted by $g_*$ (resp. $g^*$).
The ring of holomorphic functions over $\calV$ is denoted by $(\calH(\calV),*,1_{\calH(\calV)})$
and the differential ring $(\calH(\calV),\partial_1,\ldots,\partial_n)$ by $\calA$.

\begin{itemize}
\item $\calC$ denotes the sub differential ring of $\calA$ (\textit{i.e.} $\partial_i\calC\subset\calC$, for $1\le i\le n$).

\item $d$ denotes the total differential defined by
\begin{eqnarray}\label{df}
\forall f\in\calH(\calV),&&df=(\partial_1f)dz_1+\ldots+(\partial_nf)dz_n,
\end{eqnarray}
where $\partial_i$, for $i=1,\ldots,n$, denotes the partial derivative operator
$\partial/\partial{z_i}$ defined, for any $a=(a_1,\ldots,a_n)\in\calH(\calV)$, as follows
\begin{eqnarray}
\qquad
(\partial_i f)(a)=\frac{\partial f(a)}{\partial{z_i}}=\lim_{z\to a}
\frac{f(z_1,\ldots,z_i,\ldots,z_n)-f(a_1,\ldots,a_i,\ldots,a_n)}{z_i-a_i}.
\end{eqnarray}

\begin{example}\label{partial}
For any $u\in\calH(\calV)$, if $f$ satisfies the differential equation
$\partial_i f=u f$ then $f=Ce^{\log u}\in\calH(\calV)$, where $C$ is a constant.
\end{example}

\item $\Omega(\calV)$ denotes the space of holomorphic forms over $\calV$ being graded as follows
\begin{eqnarray}
\Omega(\calV)=\bigoplus_{p\ge0}\Omega^p(\calV),
\end{eqnarray}
where $\Omega^p(\calV)$ (specially, $\Omega^0(\calV)=\calH(\calV)$) is the space of holomorphic
$p$-forms over $\calV$. Equipped the wedge product, $\wedge$, $\Omega$ is a graded algebra such that, for any $\omega_1\in\Omega^{p_1}$ and $\omega_2\in\Omega^{p_2}$, one has
$\omega_1\wedge\omega_2=(-1)^{p_1p_2}\omega_2\wedge\omega_1$.

\item Over $\ncs{\calA}{\calT_n}$ (resp. $\ncs{\Omega^p(\calV)}{\calT_n},p\ge0$), the derivative operators $d,\partial_1,\ldots,\partial_n$ are extended as follows (see also \eqref{df})
\begin{eqnarray}\label{dS}
\forall S=\sum_{w\in\calT_n^*}\scal{S}{w}w,&{\bf d}S=\Sum_{w\in\calT_n^*}(d\scal{S}{w})w=\sum_{i=1}^n(\mathbf{\partial}_iS)\;dz_i.
\end{eqnarray}

\begin{example}
Let $t_{i,j}\in\calT_n$ and $U_{i,j}(z)=t_{i,j}(z_i-z_j)^{-1}$, for $0\le i<j\le n$. Any solution of $\partial_i F=U_{i,j}F$ is of
the form $F(z)=e^{t_{i,j}\log(z_i-z_j)^{-1}}C=(z_i-z_j)^{-t_{i,j}}C$, where $C\in\ncs{\C}{\calT_n}$ (see also Example \ref{partial}).
\end{example}

\item $\varsigma\path z$ is a path over $\calV$ with fixed endpoints $(\varsigma,z)$, \textit{i.e.} the curve $\gamma:[0,1]\longrightarrow\calV$ such that
$\gamma(0)=\varsigma=(\varsigma_1,\ldots,\varsigma_n)$ and $\gamma(1)=z=(z_1,\ldots,z_n)$.

For any $i,j\in\N,1\le i<j\le n$, let $\xi_{i,j}\in\calC$ and let $\omega_{i,j}:=d\xi_{i,j}$
be holomorphic $1$-form belonging to $\Omega^1(\calV)$. By \eqref{df}, one also has
\begin{eqnarray}
d\xi_{i,j}=\sum_{k=1}^n(\partial_k\xi_{i,j})dz_k.
\end{eqnarray}

\begin{example}\label{calC0}
For $\xi_{i,j}=\log(z_i-z_j)$, for $1\le i<j\le n$, let us denote the sub differential ring, of $\C(z)$,
$\C[\{(\partial_1\xi_{i,j})^{\pm1},\ldots,(\partial_n\xi_{i,j})^{\pm1}\}_{1\le i<j\le n}]$ by $\calC_0$.
\end{example}
\end{itemize}
The holomorphic function\footnote{If $f\in\calH(\calV)\equiv\Omega^0(\calV)$ and $\omega\in\Omega^1(\calV)$
then $\omega\wedge f\in\Omega^1(\calV)$ and $d(\omega\wedge f)=(d\omega)\wedge f+\omega\wedge(df)$.}
$\xi_{i,j}\in\calH(\calV)$ is a primitive for $\omega_{i,j}$ which is a exact form and then is a closed, \textit{i.e.} $d\omega_{i,j}=0$. Thus, iterated integrals
and the Chen series, of $\{\omega_{i,j}\}_{1\le i<j\le n}$ and along $\varsigma\path z$, in Definition \ref{Chenseriesdef} below are a homotopy invariant \cite{Chen1954}.

\begin{definition}[see \cite{these}]\label{growth}
\begin{enumerate}
\item Let $a\in\Q$ and $\chi_a$ be a real morphism $\calT_n^*\longrightarrow\R_{\ge0}$. The series $S\in\ncs{\calA}{\calT_n}$
is said satisfy the $\chi_a$-growth condition if and only if, choosing a compact $K$ on $\calA$,
\begin{eqnarray*}
\exists c\in\R_{\ge0},k\in\N,&\forall w\in\calT_n^{\ge k},&\absv{\scal{S}{w}}_K\le c\chi(w)\abs{w}!^{-a}.
\end{eqnarray*}

\item For $i=1$ or $2$, let $S_i\in\ncs{\calA}{\calT_n}$ and $K_i$ be a compact on $\calA$ such that
\begin{eqnarray*}
\sum_{w\in\calT_n^*}\absv{\scal{S_1}{w}}_{K_1}\absv{\scal{S_2}{w}}_{K_2}<+\infty.
\end{eqnarray*}
Then one defines
\begin{eqnarray*}
\sscal{S_1}{S_2}:=\sum_{w\in\calT_n^*}\scal{S_1}{w}\scal{S_2}{w}.
\end{eqnarray*}
\end{enumerate}
\end{definition}

\begin{lemma}[see \cite{these}]\label{condition}
Let $a_1,a_2\in\Q$ such that $a_1+a_2<1$. Let $\chi_{a_1},\chi_{a_2}$ be morphisms of monoids $\calT_n^*\longrightarrow\R_{\ge0}$.
For any $i=1,2$, let $S_i\in\ncs{\calA}{\calT_n}$ satisfying the $\chi_{a_i}$-growth condition.
If $\sum_{t\in\calT_n}\chi_{a_1}(t)\chi_{a_2}(t)<1$ then $\sscal{S_1}{S_2}$ is well defined.
\end{lemma}

\begin{proof}
By assumption, the expected result is due to the fact that
\begin{eqnarray*}
\absv{\sum_{w\in\calT_n^*}\scal{S_1}{w}\scal{S_2}{w}}&\le&\sum_{w\in\calT_n^*}\absv{\scal{S_1}{w}}_{K_1}\absv{\scal{S_2}{w}}_{K_2}\cr
&\le&c_1c_2\sum_{w\in\calT_n^*}\frac{\chi_{a_1}(w)\chi_{a_2}(w)}{\abs{w}!^{a_1+a_2}}\cr
&\le&c_1c_2\sum_{w\in\calT_n^*}\chi_{a_1}(w)\chi_{a_2}(w)\cr
&=&c_1c_2\Big(\sum_{t\in\calT_n}\chi_{a_1}(t)\chi_{a_2}(t)\Big)^*.
\end{eqnarray*}
\end{proof}

\begin{remark}
With Notations in Lemma \ref{condition} and, for any $i=1,2$,
\begin{eqnarray*}
\mathrm{Dom}(S_i):=\{R\in\ncs{\calA}{\calT_n}\|\sum_{k\ge0}\scal{S_i}{[R]_k}\mbox{ converges in }K_i\},
&[R]_k=\displaystyle\sum_{w\in\calT_n^k}\scal{R}{w}w,
\end{eqnarray*}
($\mathrm{Dom}(S_i)$ can be void), one has $S_1\in\mathrm{Dom}(S_2)$ and $S_2\in\mathrm{Dom}(S_1)$ because
\begin{eqnarray*}
\Big(\sum_{t\in\calT_n}\chi_{a_1}(t)\chi_{a_2}(t)\Big)^*
=\sum_{w\in\calT_n^*}\chi_{a_1}(w)\chi_{a_2}(w)
=\sum_{k\ge0}\sum_{w\in\calT_n^*\atop\abs{w}=k}\chi_{a_1}(w)\chi_{a_2}(w)<+\infty.
\end{eqnarray*}
\end{remark}

\begin{definition}\label{Chenseriesdef}
The iterated integral, of the holomorphic $1$-forms $\{\omega_{i,j}\}_{1\le i<j\le n}$ and along the path $\varsigma\path z$ over $\calV$, is given by $\alpha_{\varsigma}^z(1_{\calT_n^*})=1_{\calH(\calV)}$ and, for any $w=t_{i_1,j_1}\ldots t_{i_k,j_k}\in\calT_n^*$ and subdivision $(\varsigma,s_1\ldots,s_k,z)$ of the path $\varsigma\path z$ over $\calV$, by
\begin{eqnarray*}
\alpha_{\varsigma}^z(w)=\int_{\varsigma}^z\omega_{i_1,j_1}(s_1)\int_{\varsigma}^{s_1}
\omega_{i_2,j_2}(s_2)\ldots\int_{\varsigma}^{s_{k-1}}\omega_{i_k,j_k}(s_k)\in\calH(\calV).
\end{eqnarray*}
The Chen series, of
$\{\omega_{i,j}\}_{1\le i<j\le n}$ and along $\varsigma\path z$, is the following series
\begin{eqnarray*}
C_{\varsigma\path z}:=\sum_{w\in\calT_n^*}\alpha_{\varsigma}^z(w)w\in\ncs{\calA}{\calT_n}.
\end{eqnarray*}
\end{definition}

\begin{proposition}[see \cite{these}]\label{melange0}
With Notations in Definition \ref{Chenseriesdef},
\begin{enumerate}
\item $C_{\varsigma\path z}$ satisfies the $\chi_a$-growth condition.

\item Let $(\beta,\mu,\eta)$ be linear representation of $S\in\ncs{\calA^{\mathrm{rat}}}{\calT_n}$.
Then
\begin{eqnarray*}
\sscal{C_{\varsigma\path z}}{S}=
\alpha_{\varsigma}^z(S)=\Sum_{w\in\calT_n^*}(\beta\mu(w)\eta)\alpha_{\varsigma}^z(w).
\end{eqnarray*}

\item\label{shuffle} Let $S_i\in\ncs{\calA^{\mathrm{rat}}}{\calT_n}$, for $i=1,2$.
Then $\alpha_{\varsigma}^z(S_1\shuffle S_2)=\alpha_{\varsigma}^z(S_1)\alpha_{\varsigma}^z(S_2)$.
\end{enumerate}
\end{proposition}

\begin{proof}
\begin{enumerate}
\item By induction on the length of $w\in\calT_n^*$ and by use the length of the path $\varsigma\path z$, denoted by $\ell$.
one proves that $C_{\varsigma\path z}$ satisfies the $\chi_1$-growth condition, with $\chi_1(y)=\ell$, for $t\in\calT_n$.

\item Since $\scal{S}{w}=\beta\mu(w)\eta$, for $w\in\calT_n^*$, then $S$ satisfies the $\chi_2$-growth condition,
with $\chi_2(t)=\absv{\mu(t)}$, for $t\in\calT_n$ (using of norm on matrices with coefficients in $\calA$).
By Lemma \ref{condition}, it follows then the expected result.

\item The recursion \eqref{recursion} yields $\alpha_{\varsigma}^z(u\shuffle v)=\alpha_{\varsigma}^z(u)\alpha_{\varsigma}^z(v)$,
for $u,v\in\calT_n^*$ (a Chen's lemma, \cite{Chen1954})
and then the expected result, by extending to $\ncs{\calA^{\mathrm{rat}}}{\calT_n}$.
\end{enumerate}
\end{proof}

\begin{definition}\label{calK}
Let $\calK:=(\mathrm{span}_{\calA}\{\alpha_{\varsigma}^z(R)\}_{R\in\ncs{\calC^{\mathrm{rat}}}{\calT_n}},\times)$ and then
$\calC\subset\calA\subset\calK$.
\end{definition}

\begin{remark}\label{expad_example}
\begin{enumerate}
\item Using \eqref{adjoint}, for any $S\in\ncs{\calL ie_{\calK}}{T_n}$, let $\varphi_s=e^{\ad_S}$. One has
\begin{eqnarray*}
\forall R\in\ncs{\calL ie_{\calA}}{\calT_{n-1}},&
\varphi_S(R)=e^{\ad_S}R=\Sum_{k\ge0}\Frac{1}{k!}\ad_S^kR\in\ncs{\calL ie_{\calK}}{\calT_n}.
\end{eqnarray*}
In particular, for $S\in\ncp{\calL ie_{\calK}}{T_n},R\in\ncp{\calL ie_{\calK}}{\calT_{n-1}}$ and then $S\in T_n,R\in\calT_{n-1}$.
Using \eqref{basisP}, if $\varphi_{P_l}=e^{\ad_{P_l}}$ with $l\in\Lyn T_n$ then,
for $q=P_{\ell}$ with $\ell\in\calT_{n-1}$,  and using \eqref{orderLyndon}-\eqref{orderLyndon4},
one obtains $l\ell\in\Lyn\calT_n$ and then (see \eqref{basisP})
\begin{eqnarray*}
\varphi_{P_l}(P_{\ell})=e^{\ad_{P_l}}P_{\ell}=\sum_{k\ge0}\frac{1}{k!}\ad_{P_l}^kP_{\ell}=\sum_{k\ge0}\frac1{k!}P_{l^k\ell}.
\end{eqnarray*}
In particular, if $P_l=l\in T_n$ and $P_{\ell}=\ell\in\calT_{n-1}$ then (see \eqref{basisP}--\eqref{basisS})
\begin{eqnarray*}
\varphi_l(\ell)=e^{\ad_l}\ell=\sum_{k\ge0}\frac{1}{k!}\ad_l^k\ell=\sum_{k\ge0}\frac{r(l^k\ell)}{k!}
&\mbox{and by duality}&\check\varphi_l(\ell)=\sum_{k\ge0}\frac{l^k\ell}{k!}=e^l\ell.
\end{eqnarray*}
\end{enumerate}
\end{remark}

\begin{corollary}[see \cite{these}]\label{melange}
Let $t_{i,j}\in\calT_n,k\ge1$ and $\alpha_{\varsigma}^z:(\ncs{\calC^{\mathrm{rat}}}{\calT_n},\shuffle,1_{\calT_n^*})\longrightarrow(\calK,\times,1_{\calC})$.
\begin{enumerate}
\item One has $\alpha_{\varsigma}^z(t_{i,j}^k)=(\alpha_{\varsigma}^z(t_{i,j}))^k/{k!}$ and then $\alpha_{\varsigma}^z(t_{i,j}^*)=\exp(\alpha_{\varsigma}^z(t_{i,j}))$.

\item For any $R\in\ncs{\calC^{\mathrm{rat}}}{\calT_n}$ and $H\in\ncs{\calC^{\mathrm{rat}}}{T_n}$,
\begin{eqnarray*}
\alpha_{\varsigma}^z((t_{i,j}H)\halfshuffle R)=\left\{
\begin{array}{ccc}
\alpha_{\varsigma}^z(t_{i,j}H)&\mbox{if}&R=1_{\calT_n^*},\cr
\displaystyle\int_{\varsigma}^z\omega_{{i,j}}(s)\alpha_{\varsigma}^s(H)\alpha_{\varsigma}^s(R)&\mbox{if}&R\neq 1_{\calT_n^*}.
\end{array}\right.
\end{eqnarray*}
\end{enumerate}
\end{corollary}

\begin{proof}
By Proposition \ref{melange0} and, on the one hand, since $t_{i,j}^k=t_{i,j}^{\shuffle k}/k!$ then it follows the first result and, on the other hand, by \eqref{halfshuffle}, it follows the last result.
\end{proof}

\begin{remark}[\cite{ACA,orlando}]\label{Hypergeometric}
Developping the idea of \textit{universality}, for simplification, let $C_{\varsigma\path z}$ be the Chen series, along $\varsigma\path z$ and of $\omega_0(z)=dz/z$ and $\omega_1(z)=dz/(1-z)$.

Let $a,b,c$ be real parameters and let $S\in\ncs{\C^{\mathrm{rat}}}{x_0,x_1}$ be the rational
series admitting the triplet $(\beta,\mu,\eta)$ as parametrized linear representation \cite{ACA}:
\begin{eqnarray*}
\beta={}^t\eta=\begin{pmatrix}1&0\end{pmatrix},
&\mu(x_0)=-\begin{pmatrix}0&0\cr ab& c\end{pmatrix},
&\mu(x_1)=-\begin{pmatrix}0&1\cr0&c-a-b\end{pmatrix}.
\end{eqnarray*}

One can consider the following hypergeometric equation
\begin{eqnarray*}
z(1-z)\ddot y(z)+[c-(a+b+1)z]\dot y(z)-ab y(z)=0,
\end{eqnarray*}
in which putting $q_1(z)=-y(z)$ and $q_2(z)=(1-z)\dot y(z)$, the state vector $q$ satisfies the
following linear differential equation associated to $(\beta,\mu,\eta)$ \cite{fliess1,fliess2}
\begin{eqnarray*}
\dot q(z)=\begin{pmatrix}\dot q_1\cr\dot q_2\end{pmatrix}
=\left(\frac{\mu(x_0)}z+\frac{\mu(x_1)}{1-z}\right)
\begin{pmatrix}q_1\cr q_2\end{pmatrix},&\begin{pmatrix}q_1(0)\cr q_2(0)\end{pmatrix}=\begin{pmatrix}1\cr0\end{pmatrix}.
\end{eqnarray*}
Or equivalently, considering two following parametrized linear vector fields \cite{fliess1,fliess2}
\begin{eqnarray*}
A_0=-(abq_1+cq_2){\partial}/{\partial q_2}&\mbox{and}&
A_1=-q_2{\partial}/{\partial q_1}-(c-a-b)q_2{\partial}/{\partial q_2},
\end{eqnarray*}
$q$ satisfies then the following differential equation \cite{fliess1,fliess2}
\begin{eqnarray*}
\dot q(z)=z^{-1}A_0(q)+(1-z)^{-1}A_1(q)&\mbox{and}&y(z)=-q_1(z).
\end{eqnarray*}

By Proposition \ref{melange0}, one has $\langle{C_{0\path z}}\Vert{S}\rangle=\alpha_0^z(S)=q_1(z)=-y(z)$.
\end{remark}

\subsection{Noncommutative differential equations}\label{eqdiffnoncommutatice}
Getting back to \eqref{NCDE}--\eqref{split}, let us consider the Chen series $C_{\varsigma\path z}$, of the holomorphic $1$-forms
$\{\omega_{i,j}\}_{1\le i<j\le n}$ and along the path $\varsigma\path z$ over the simply connected manifold $\calV$. Let $g$ be a
diffeomorphism on $\calV$ and $C_{g_*\varsigma\path z}$ be the Chen series, of $\{g^*\omega_{i,j}\}_{1\le i<j\le n}$ and along$\varsigma\path z$, or equivalently, of $\{\omega_{i,j}\}_{1\le i<j \le n}$ and along $g_*\varsigma\path z$ \cite{Chen1954}:
\begin{eqnarray}\label{Chen_g*}
C_{g_*\varsigma\path z}
&=&\sum_{m\ge0}\sum_{t_{i_1,j_1}\ldots t_{i_m,j_m}\in\calT_n^*}
\int_{\varsigma}^zg^*\omega_{i_1,j_1}(s_1)\ldots\int_{\varsigma}^{s_{m-1}}g^*\omega_{i_m,j_m}(s_m)\cr
&&\hskip68mm t_{i_1,j_1}\ldots t_{i_m,j_m}\cr
&=&\sum_{w\in\calT_n^*}\alpha_{g(\varsigma)}^{g(z)}(w)w.
\end{eqnarray}
$C_{g_*\varsigma\path z}$ is obtained by the Picard's iteration, as in \eqref{picard0}, and is convergent
\begin{eqnarray}\label{picard0*}
&F^*_0(\varsigma,z)=1_{\calA},&F^*_i(\varsigma,z)=F_{i-1}^*(\varsigma,z)+\int_{\varsigma}^zM_n^*(s)F^*_{i-1}(s),i\ge1,
\end{eqnarray}
where $M_n^*:=g^*M_n$, associated to ${\bf d}S=M_n^*S$.

\begin{definition}\label{G}
By Definition \ref{calK}, let $\G:=\{e^{\ad_S}\}_{S\in\ncs{\calL ie_{\calK}}{T_n}}$.
\end{definition}

For any $\phi\in\G$, let $\check\phi$ be its adjoint to $\phi$ and let us consider the Picard's iterations with
initial condition $F_0^{\phi}$, according to following recursion similar to \eqref{picard0} (for $i\ge1$):
\begin{eqnarray}\label{picard1}
F_i^{\phi^{(\varsigma,z)}}(\varsigma,z)=F_{i-1}^{\phi^{(\varsigma,z)}}(\varsigma,z)+
\int_{\varsigma}^zM_{n-1}^{\phi^{(\varsigma,s)}}(s)F_{i-1}^{\phi^{(\varsigma,s)}}(\varsigma,s).
\end{eqnarray}
where
\begin{eqnarray}\label{NCDE1}
M_{n-1}^{\phi}:=\phi(M_{n-1})&\mbox{associated to}&{\bf d}F^{\phi}=M_{n-1}^{\phi}F.
\end{eqnarray}

\begin{proposition}\label{MRSBTT}
Let $S\in\ncs{\calA}{\calT_n}$ be a grouplike solution of \eqref{NCDE}. Then
\begin{enumerate}
\item If $H\in\ncs{\calA}{\calT_n}$ is another grouplike solution for \eqref{NCDE}
then there exists $C\in\ncs{\calL ie_{\calA}}{\calT_n}$ such that $S=He^C$ (and conversely).
\item\label{item} The following assertions are equivalent
\begin{enumerate}
\item\label{1} The family $\{\scal{S}{w}\}_{w\in\calT_n^*}$ is $\calC$-linearly free.
\item\label{2} The family $\{\scal{S}{l}\}_{l\in\Lyn\calT_n}$ is $\calC$-algebraically free.
\item\label{3} The family $\{\scal{S}{t}\}_{t\in\calT_n}$ is $\calC$-algebraically free.
\item\label{4} The family $\{\scal{S}{t}\}_{t\in\calT_n\cup\{1_{\calT_n^*}\}}$ is $\calC$-linearly free.
\item\label{5} The family $\{\omega_{i,j}\}_{1\le i<j\le n}$ is such that, for any
$(c_{i,j})_{1\le i<j\le n}\in\C^{(\calT_n)}$ and $f\in\mathrm{Frac}(\calC)$, one has
\begin{eqnarray*}
\sum_{1\le i<j\le n}c_{i,j}\omega_{i,j}=df&\Longrightarrow&(\forall 1\le i<j\le n)(c_{i,j}=0).
\end{eqnarray*}
\item\label{6} $\{\omega_{i,j}\}_{1\le i<j\le n}$ is $\calC$-free and
$d\mathrm{Frac}(\calC)\cap\mathrm{span}_{\C}\{\omega_{i,j}\}_{1\le i<j\le n}=\{0\}$.
\end{enumerate}
\end{enumerate}
\end{proposition}

\begin{proof}[Sketch]
\begin{enumerate}
\item The proof is similarly treated in \cite{orlando}: since ${\bf d}(SS^{-1})={\bf d}(\mathrm{Id})=0$
then, by the Liebniz rule, $({\bf d}S)S^{-1}+S({\bf d}S^{-1})=0$ and then (see also \eqref{grouplike})\\
${\bf d}S^{-1}=-S^{-1}({\bf d}S)S^{-1}=-S^{-1}(M_nS)S^{-1}=-S^{-1}M_n(SS^{-1})=-S^{-1}M_n$
and then
${\bf d}(S^{-1}H)=S^{-1}({\bf d}H)+({\bf d}S^{-1})H=S^{-1}(M_nH)-(S^{-1}M_n)H=0$.
Thus,  $S^{-1}H$ is a constant series. Since the inverse and the product
of grouplike elements are grouplike then it follows the expected result.

\item This is a grouplike version of the abstract form of Theorem 1 of \cite{Linz}. It goes as follows
\begin{itemize}
\item due to the fact that $\calA$ is without zero divisors, using the fields of fractions of $\calC$ and $\calA$, we have the embeddings $\calC\subset\mathrm{Frac}(\calC)\subset\mathrm{Frac}(\calA)$. $\mathrm{Frac}(\calA)$ is a differential field, and its differential operator can still be denoted by $d$ as it induces the previous one on $\calA$. The same holds for $\ncs{\calA}{\calT_n}\subset\ncs{\mathrm{Frac}(\calA)}{\calT_n}$ and ${\bf d}$. Hence, equation \eqref{NCDE} can be transported in
$\ncs{\mathrm{Frac}(\calA)}{\calT_n}$ and $M_n$ satisfies the same condition as previously. 

\item Equivalence between \ref{1}-\ref{4} comes from the fact that $\calC$ is without zero divisors and then,
by denominator chasing, linear independances with respect to $\calC$ and $\mathrm{Frac}(\calC)$ are equivalent.
In particular, supposing condition \ref{4}, the family $\{\scal{S}{x}\}_{x\in\calT_n\cup\{1_{\calT_n^*}\}}$
(basic triangle) is $\mathrm{Frac}(\calC)$-linearly independent which imply, by Theorem 1 of \cite{Linz},
condition \ref{5}.

\item Still by Theorem 1 of \cite{Linz}, \ref{5}--\ref{6} are equivalent and then $\{\scal{S}{w}\}_{w\in\calT_n^*}$
is $\mathrm{Frac}(\calC)$-linearly free which induces $\calC$-linear independence (\textit{i.e.} \ref{1}).
\end{itemize}  
\end{enumerate}
\end{proof}

In the sequel, with the notations in Definition \ref{calK}, let
\begin{itemize}
\item $\calF(S):=\mathrm{span}_{\calC}\{\scal{S}{w}\}_{w\in\calT_n^*}$, for $S\in\ncs{\calA}{\calT_n}$,
\item $g$ be the diffeomorphism on $\calV$ acting by pullback on $\{\omega_{i,j}\}_{1\le i<j\le n}$ as follows
\begin{eqnarray}\label{diffeomorphism1}
g^*\omega_{i,j}=\sum_{1\le k<l\le n}\omega_{k,l}h_{i,j\atop k,l},&\mbox{for}&h_{i,j\atop k,l}\in\calK,
\end{eqnarray}

\item $\psi$ be the morphism of algebras $(\ncp{\calC}{\calT_n},{\tt conc})\longrightarrow(\ncs{\calC^{\mathrm{rat}}}{\calT_n},\halfshuffle)$ defined, for any $t_{i,j}\in\calT_n$, as follows\footnote{$\forall t_{i_1,j_1}\ldots t_{i_r,j_r}\calT_n^*,\psi(t_{i_1,j_1}\ldots t_{i_r,j_r})
=\psi(t_{i_1,j_1})\halfshuffle(\psi(t_{i_2,j_2})\ldots(\halfshuffle\psi(t_{i_r,j_r})))$.} (see also \eqref{halfshuffle} for the half-shuffle)
\begin{eqnarray}\label{diffeomorphism2}
\psi(t_{i,j})=\sum_{1\le k<l\le n}t_{k,l}H_{i,j\atop k,l},&\mbox{for}&H_{i,j\atop k,l}\in\ncs{\calC^{\mathrm{rat}}}{\calT_n}.
\end{eqnarray}
\end{itemize}

\begin{example}\label{simple}
For
$\omega_{1,2}(z)=-d\log(z_1-z_2),\omega_{1,3}(z)=-d\log(z_1-z_3),\omega_{2,3}(z)=-d\log(z_2-z_3)$ and $\calT_3=\{t_{1,2},t_{1,3},t_{2,3}\}$, let
\begin{itemize}
\item $g$ be the diffeomorphism on $\tilde{\C_3^*}$ acting on $\{\omega_{i,j}\}_{1\le i<j\le n}$ as follows
\begin{eqnarray*}
g^*\begin{pmatrix}
\omega_{1,2}\cr\omega_{1,3}\cr\omega_{2,3}
\end{pmatrix}
=\begin{pmatrix}
1&0&0\cr0&(z_1-z_2)^{-1}\log((z_2-z_3)^{-1}))&0\cr0&0&1
\end{pmatrix}
\begin{pmatrix}
\omega_{1,2}\cr\omega_{1,3}\cr\omega_{2,3}
\end{pmatrix},
\end{eqnarray*}
\item $\psi:(\ncp{\calC}{\calT_3},{\tt conc})\longrightarrow(\ncs{\calC^{\mathrm{rat}}}{\calT_3},\halfshuffle)$ be the morphism of algebras defined by
$\psi(t_{1,2})=t_{1,2}t_{1,2}^*$ and $\psi(t_{1,3})=t_{1,3}t_{1,2}^*$ and $\psi(t_{2,3})=t_{2,3}t_{2,3}^*$.
\item With the data in previous items, by Example \ref{Excomposite1} and Proposition \ref{melange}, one has
\begin{eqnarray*}
\alpha_{z\varsigma}^z(\psi(t_{1,3})\halfshuffle t_{2,3})
&=&\alpha_{z\varsigma}^z((t_{1,3}t_{1,2}^*)\halfshuffle t_{2,3})\cr
&=&\alpha_{\varsigma}^z(t_{1,3}(t_{1,2}^*\shuffle t_{2,3}))\cr
&=&\int_{z\varsigma}^z-d\log(z_1-z_3)\,\frac{\log((z_2-z_3)^{-1})}{z_1-z_2}\cr
&=&\int_{z\varsigma}^zg^*\omega_{1,3}\cr
&=&\alpha_{g(\varsigma)}^{g(z)}(t_{1,3}).
\end{eqnarray*}
\end{itemize}
\end{example}

Proposition \ref{MRSBTT} holds, in particular, for $C_{\varsigma\path z}$. Hence, one deduces that

\begin{corollary}\label{Chensolution}
\begin{enumerate}
\item The following assertions are equivalent\footnote{In particular, $\calC=\calC_0$
(see Example \ref{calC0}) yielding $\F_{KZ_n}$ in Definition \ref{F}, Corollaries \ref{P}--\ref{C} below.}
\begin{enumerate}
\item The restricted $\shuffle$-morphism $\alpha_{\varsigma}^z$, on $\ncp{\calC}{\calT_n}$, is injective.
\item The family $\{\alpha_{\varsigma}^z(w)\}_{w\in\calT_n^*}$ is $\calC$-linearly free.
\item The family $\{\alpha_{\varsigma}^z(l)\}_{l\in\Lyn\calT_n}$ is $\calC$-algebraically free.
\item The family $\{\alpha_{\varsigma}^z(t)\}_{t\in\calT_n}$ is $\calC$-algebraically free.
\item The family $\{\alpha_{\varsigma}^z(t)\}_{t\in\calT_n\cup\{1_{\calT_n^*}\}}$ is $\calC$-linearly free.
\item $\begin{array}{lll}
\forall E\in e^{\ncs{\calL ie_{\calC}}{\calT_n}},
&\exists\phi\in\mathrm{Aut}(\calF(C_{\varsigma\path z})),
&\phi(C_{\varsigma\path z})=C_{\varsigma\path z}E.\end{array}$
\end{enumerate}

\item The following assertions are equivalent (see Notations in \eqref{Chen_g*}, \eqref{diffeomorphism1}--\eqref{diffeomorphism2})
\begin{enumerate}
\item For any $1\le i<j\le n$ and $1\le k<l\le n$, one has $h_{i,j\atop k,l}(z)=\alpha_{\varsigma}^z(H_{i,j\atop k,l})$.

\item The restricted $\shuffle$-morphism $\alpha_{\varsigma}^z$, on $\ncp{\calC}{\calT_n}$, is injective.

\item The Chen series, of $\{\omega_{i,j}\}_{1\le i<j\le n}$ and along $g_*\varsigma\path z$, satisfies
\begin{eqnarray*}
C_{g_*\varsigma\path z}=\sum_{w\in\calT_n^*}\alpha_{\varsigma}^z(\psi(w))w
=C_{\varsigma\path z}E,&\mbox{where}&E\in e^{\ncs{\calL ie_{\calC}}{\calT_n}}.
\end{eqnarray*}
\end{enumerate}

\item For any $\phi\in\G$, there exists a diffeomorphism $g$ on $\calV$ such that the Chen series,
of $\{\omega_{i,j}\}_{1\le i<j\le n-1}$ along $g_*\varsigma\path z$, can be expressed as follows
\begin{eqnarray*}
C'_{g_*\varsigma\path z}
:=\sum_{w\in\calT_{n-1}^*}\alpha_{g(\varsigma)}^{g(z)}(w)w
=\sum_{w\in\calT_{n-1}^*}\alpha_{\varsigma}^z(w)\phi^{(\varsigma,z)}(w).
\end{eqnarray*}
\end{enumerate}
\end{corollary}

\begin{proof}
The first item is a consequence of Proposition \ref{MRSBTT}.
Applying Propositions \ref{melange0}--\ref{MRSBTT} and Corollary \ref{melange}, one gets the second item.
By duality, one gets
\begin{eqnarray*}
\sum_{w\in\calT_{n-1}^*}\alpha_{\varsigma}^z(w)\phi^{(\varsigma,z)}(w)=\sum_{w\in\calT_{n-1}^*}\alpha_{\varsigma}^z(\check\phi^{(\varsigma,z)}(w))w.
\end{eqnarray*}
Applying the second item with $\psi=\check\phi$, it follows the last item.
\end{proof}

In Proposition \ref{MRSBTT}, the Hausdorff group of $H_{\shuffle}(\calT_n)$ plays the r\^ole of the differential Galois group of \eqref{NCDE} + grouplike solutions, \textit{i.e.} $\mathrm{Gal}(M_n)=e^{\ncs{\calL ie_{\C}}{\calT_n}}$, mapping grouplike solution to another grouplike solution and then leading to the definitions, on the one hand, of the system fundamental of \eqref{NCDE} as $\{C_{\varsigma\path z}\}$ and, on the other hand, of the PV extension related to  \eqref{NCDE} as $\widehat{\calC.\calT_n}\{C_{\varsigma\path z}\}$ \cite{orlando}.

\subsection{Explicit solutions of noncommutative differential equations}\label{Chenseries}
In the sequel, $\{V_k\}_{k\ge0}$ and $\{\hat V_k\}_{k\ge0}$ denote the sequences of series in $\ncs{\calA}{\calT_n}$,
satisfying the recursion in \eqref{S_k} with the following starting conditions being grouplike series:
\begin{eqnarray}
&&V_0(\varsigma,z):=(\alpha_{\varsigma}^z\otimes\mathrm{Id})\calD_{T_n}=\prod\limits_{l\in\Lyn T_n}^{\searrow}e^{\alpha_{\varsigma}^z(S_l)P_l}{\mbox{(decreasing lexicographical}\atop\mbox{ordered product),}}\label{V_0}\\
&&\hat V_0(\varsigma,z):=e^{\sum\limits_{t\in T_n}\alpha_{\varsigma}^z(t)t}
=V_0(\varsigma,z)\mod[\ncs{\calL ie_{\calA}}{T_n},\ncs{\calL ie_{\calA}}{T_n}].\label{hatV_0}
\end{eqnarray}

\begin{remark}\label{nil_app}
\begin{itemize}
\item $V_0$ is the Chen series, of $\{\omega_{k,n}\}_{1\le k\le n-1}$ and along $\varsigma\path z$,
and satisfies the $\chi_a$-growth condition (see by Proposition \eqref{melange0}).
It can be obtained by using the following Picard's iteration, analogous to \eqref{picard0},
which is convergent for the discrete topology but does
not mean that $V_0$ satisfies ${\bf d}S=\bar M_nS$ (see Remark \ref{BE} below)
\begin{eqnarray*}
&F_0(\varsigma,z)=1_{\calH(\calV)},&F_i(\varsigma,z)=
F_{i-1}(\varsigma,z)+\int_{\varsigma}^z\bar M_n(s)F_{i-1}(s),i\ge1.
\end{eqnarray*}

\item  With data in \eqref{data} below, $V_0$ will behave, for\footnote{See Note \ref{NOTE}.} $z_n\to z_{n-1}$,
as the generating series of hyperlogarithms (see \eqref{sghyperlog}--\eqref{DEN} below) and, of course, as the generating series of polylogarithms for $n=3$ (see \eqref{sgLi} below).

\item $\hat V_0$ satisfies the partial differential equation $\partial_nf=\bar M_nf$ and \eqref{hatV_0} is equivalent to a nilpotent structural approximation of order $1$ of $V_0$ \cite{hoangjacoboussous},
\textit{i.e.} $\log\hat V_0=\log V_0\mod[\ncs{\calL ie_{\calA}}{T_n},\ncs{\calL ie_{\calA}}{T_n}]$
 (see also Remark \ref{BE} below).
\end{itemize}
\end{remark}

\begin{definition}\label{phi}
\begin{enumerate}
\item Let $\varphi_{T_n}$ and $\hat\varphi_{T_n}\in\G$ be the ${\tt conc}$-morphisms,
depending on $\varsigma\path z$ subdived by $(\varsigma,s_1,\ldots,s_k,z)$, of $\ncp{\calA}{\calT_n}$ such that $\varphi_{T_n}\equiv\varphi_n\equiv\mathrm{Id}$, over $T_n^*$,
and by\footnote{For any $a,b\in\ncs{\calL ie_{\calA}}{\calT_n}$, one has $e^{-a}be^a=e^{\ad_{-a}}b$ \cite{bourbaki}.} (using the decreasing lexicographical order product)
\begin{eqnarray*}
\varphi_{T_n}^{(\varsigma,z)}=\prod_{l\in\Lyn T_n}^{\searrow}e^{\ad_{-\alpha_{\varsigma}^{s_k}(S_l)P_l}}
&\mbox{and}&\hat\varphi^{(\varsigma,z)}_{T_n}=e^{\sum\limits_{t\in T_n}\ad_{-\alpha_{\varsigma}^{s_k}(S_l)P_l}},
\end{eqnarray*}
over $\calT_{n-1}^*$. They are chronologically defined, for $t_{i_1,j_1}\ldots t_{i_k,j_k}\in\calT_{n-1}^*$ by
\begin{eqnarray*}
\varphi_{T_n}^{(\varsigma,z)}(t_{i_1,j_1}\ldots t_{i_k,j_k})=\varphi_{T_n}^{(\varsigma,s_1)}(t_{i_1,j_1})\cdots\varphi_{T_n}^{(\varsigma,s_k)}(t_{i_k,j_k}),\cr
\hat\varphi_{T_n}^{(\varsigma,z)}(t_{i_1,j_1}\ldots t_{i_k,j_k})=\hat\varphi_{T_n}^{(\varsigma,s_1)}(t_{i_1,j_1})\cdots\hat\varphi_{T_n}^{(\varsigma,s_k)}(t_{i_k,j_k}).
\end{eqnarray*}

\item Let $\varphi_n$ and $\hat\varphi_n$ be the morphisms of $\ncp{\calA}{\calT_n}$ defined, for any $t\in\calT_n$, by
\begin{eqnarray*}
\varphi_n(t)=\varphi_{T_n}(t)\mod\calJ_n&\mbox{and}&\hat\varphi_n(t)=\hat\varphi_{T_n}(t)\mod\calJ_n,
\end{eqnarray*}
where $\calJ_n$ is the ideal of relators on $\{t_{i,j}\}_{1\le i<j\le n}$.
\end{enumerate}
\end{definition}
 
\begin{proposition}\label{volterraexp}
With Notations in Definitions \ref{Chenseriesdef}--\ref{phi} and \eqref{V_0}--\eqref{hatV_0}, one has
\begin{eqnarray*}
\varphi_{T_n}^{(\varsigma,z)}(t_{i_k,j_k})=e^{\ad_{-V_0(\varsigma,s_k)}}
t_{i_k,j_k}&\mbox{and}&\hat\varphi_{T_n}^{(\varsigma,z)}(t_{i_k,j_k})=e^{\ad_{-\hat V_0(\varsigma,s_k)}}t_{i_k,j_k}
\end{eqnarray*}
and there is, on the one hand, $\{\kappa_w\}_{w\in\calT_{n-1}^*}$ and $\{\hat\kappa_w\}_{w\in\calT_{n-1}^*}$, on the other hand, $H$ and $\hat H$ in $\ncs{\calA}{\calT_n}$ satisfying \eqref{NCDE1} such that
\begin{eqnarray*}
\forall w\in\calT_{n-1}^*,\quad
\kappa_w=V_0\varphi_{T_n}(w)&\mbox{and}&\hat\kappa_w=\hat V_0\hat\varphi_{T_n}(w),\\
\sum_{k\ge0}V_k=V_0H&\mbox{and}&\sum_{k\ge0}\hat V_k=\hat V_0\hat H.
\end{eqnarray*}
Moreover, for any $k\ge1$, one has
\begin{eqnarray*}
V_k(\varsigma,z)
=\sum_{w=t_{i_1,j_1}\ldots,t_{i_k,j_k}\in\calT_{n-1}^*}\int_{\varsigma}^z\omega_{i_1,j_1}(s_1)\cdots\int_{\varsigma}^{s_{k-1}}\omega_{i_k,j_k}(s_k)\kappa_w(z,s),\cr
\hat V_k(\varsigma,z)
=\sum_{w=t_{i_1,j_1}\ldots,t_{i_k,j_k}\in\calT_{n-1}^*}\int_{\varsigma}^z\omega_{i_1,j_1}(s_1)\cdots\int_{\varsigma}^{s_{k-1}}\omega_{i_k,j_k}(s_k)\hat\kappa_w(z,s).
\end{eqnarray*}

Reducing by $\calJ_n$, one gets analogous results using respectively $\varphi_n$ and $\hat\varphi_n$ (and then, in this case, one has $\kappa_w=V_0\varphi_n(w)$ and $\hat\kappa_w=\hat V_0\hat\varphi_n(w)$, for $w\in\calT_{n-1}^*$).
\end{proposition}

\begin{proof}
The first result is a consequence of \eqref{a(S)} and \eqref{V_0}--\eqref{hatV_0}.
According to \eqref{S_k}, iterative computations by \eqref{picard1} yield the expected expressions with
\begin{eqnarray*}
H(\varsigma,z)
&=&1_{\calT_n^*}+\sum_{k\ge1}\sum_{t_{i_1,j_1}\ldots t_{i_k,j_k}\in\calT_{n-1}^*}\\
&&\int_{\varsigma}^z\omega_{i_1,j_1}(s_1)\varphi_{T_n}^{(\varsigma,s_1)}(t_{i_1,j_1})\ldots\int_{\varsigma}^{s_{k-1}}
\omega_{i_k,j_k}(s_k)\varphi_{T_n}^{(\varsigma,s_k)}(t_{i_k,j_k})\cr
&=&1_{\calT_n^*}+\sum_{k\ge1}\sum_{t_{i_1,j_1}\ldots t_{i_k,j_k}\in\calT_{n-1}^*}
\int_{\varsigma}^z\omega_{i_1,j_1}(s_1)\ldots\int_{\varsigma}^{s_{k-1}}\omega_{i_k,j_k}(s_k)\\
&&\hskip6cm\varphi_{T_n}^{(\varsigma,z)}(t_{i_1,j_1}\ldots t_{i_k,j_k}),\cr
\hat H(\varsigma,z)
&=&1_{\calT_n^*}+\sum_{k\ge1}\sum_{t_{i_1,j_1}\ldots t_{i_k,j_k}\in\calT_{n-1}^*}\\
&&\int_{\varsigma}^z\omega_{i_1,j_1}(s_1)\hat\varphi_{T_n}^{(\varsigma,s_1)}(t_{i_1,j_1})\ldots\int_{\varsigma}^{s_{k-1}}
\omega_{i_k,j_k}(s_k)\hat\varphi_{T_n}^{(\varsigma,s_k)}(t_{i_k,j_k})\cr
&=&1_{\calT_n^*}+\sum_{k\ge1}\sum_{t_{i_1,j_1}\ldots t_{i_k,j_k}\in\calT_{n-1}^*}
\int_{\varsigma}^z\omega_{i_1,j_1}(s_1)\ldots\int_{\varsigma}^{s_{k-1}}\omega_{i_k,j_k}(s_k)\cr
&&\hskip6cm\hat\varphi_{T_n}^{(\varsigma,z)}(t_{i_1,j_1}\ldots t_{i_k,j_k}).
\end{eqnarray*}
\end{proof}

\begin{theorem}[Volterra expansion like for Chen series]\label{Chen_braids}
With Notations in Definitions \ref{dual_bases}--\ref{phi}, Theorem \ref{diagonal} and Propositions \ref{MRSBTT}--\ref{volterraexp}, $C_{\varsigma\path z}=V_0(\varsigma,z)H(\varsigma,z)$, one has
\begin{eqnarray*}
H(\varsigma,z)=(\alpha_{\varsigma}^z\otimes\mathrm{Id})\lambda(\calM_{\calT_{n-1}}^*)
=(\alpha_{\varsigma}^z\otimes\mathrm{Id})\mathrm{diag}((\lambda_l\otimes\lambda_r)(\calM_{\calT_{n-1}}^*)),\cr
\hat H(\varsigma,z)=(\alpha_{\varsigma}^z\otimes\mathrm{Id})\hat\lambda(\calM_{\calT_{n-1}}^*)
=(\alpha_{\varsigma}^z\otimes\mathrm{Id})\mathrm{diag}((\hat\lambda_l\otimes\lambda_r)(\calM_{\calT_{n-1}}^*)).
\end{eqnarray*}
Reducing by $\calJ_n$, one gets analogous results using respectively $\varphi_n$ and $\hat\varphi_n$.
\end{theorem}

\begin{proof}
By Proposition \ref{expression}, the images by $\alpha_{\varsigma}^z\otimes\mathrm{Id}$ of $\lambda(t\otimes t)$ and $\hat\lambda(t\otimes t)$,
for $t\in\calT_{n-1}$, are respectively followed (see also Notations in \eqref{barv}, \eqref{antipode_shuffle} and \eqref{rho})
\begin{eqnarray*}
\int_{\varsigma}^z\omega_{i,j}(s)\varphi_{T_n}^{(\varsigma,s)}(t)=(\alpha_{\varsigma}^z\otimes\mathrm{Id})\lambda(t\otimes t)=\sum_{v\in T_n^*}\alpha_{\varsigma}^z(a(vt))r(vt),\cr
\int_{\varsigma}^z\omega_{i,j}(s)\hat\varphi_{T_n}^{(\varsigma,s)}(t)=(\alpha_{\varsigma}^z\otimes\mathrm{Id})\hat\lambda(t\otimes t)=\sum_{v\in T_n^*}\alpha_{\varsigma}^z(a(\hat vt))r(vt).
\end{eqnarray*}
Hence, for any $t_{i_1,j_1}\ldots t_{i_k,j_k}\in\calT_{n-1}^*$, one iteratively obtains
\begin{eqnarray*}
&&\int_{\varsigma}^z\omega_{i_1,j_1}(s_1)\ldots\int_{\varsigma}^{s_{k-1}}\omega_{i_k,j_k}(s_k)
\varphi_{T_n}^{(\varsigma,z)}(t_{i_1,j_1}\ldots t_{i_k,j_k})\cr
&&=\sum_{v_1,\ldots,v_k\in T_n^*\atop t_1,\ldots,t_k\in\calT_{n-1}}
\alpha_{\varsigma}^z(a(v_1t_1)\halfshuffle\cdots\halfshuffle a(v_kt_k))r(v_1t_1)\ldots r(v_kt_k),\\
&&\int_{\varsigma}^z\omega_{i_1,j_1}(s_1)\ldots\int_{\varsigma}^{s_{k-1}}\omega_{i_k,j_k}(s_k)
\hat\varphi_{T_n}^{(\varsigma,z)}(t_{i_1,j_1}\ldots t_{i_k,j_k})\cr
&&=\sum_{v_1,\ldots,v_k\in T_n^*\atop t_1,\ldots,t_k\in\calT_{n-1}}
\alpha_{\varsigma}^z(a(\hat v_1t_1)\halfshuffle\cdots\halfshuffle a(\hat v_kt_k))r(v_1t_1)\ldots r(v_kt_k).
\end{eqnarray*}
By Propositions \ref{expression}, \ref{volterraexp}, summing for $k$ on $\N$, it follows
the expected expressions:
\begin{eqnarray*}
H(\varsigma,z)=1_{\calT_n^*}+\sum_{k\ge1}\sum_{v_1,\ldots,v_k\in T_n^*\atop t_1,\ldots,t_k\in\calT_{n-1}}
\alpha_{\varsigma}^z(a(v_1t_1)\halfshuffle\cdots\halfshuffle a(v_kt_k))r(v_1t_1)\ldots r(v_kt_k),\cr
\hat H(\varsigma,z)
=1_{\calT_n^*}+\sum_{k\ge1}\sum_{v_1,\ldots,v_k\in T_n^*\atop t_1,\ldots,t_k\in\calT_{n-1}}
\alpha_{\varsigma}^{s_1}(a(\hat v_1t_1)\halfshuffle\cdots\halfshuffle a(\hat v_kt_k))r(v_1t_1)\ldots r(v_kt_k).
\end{eqnarray*}
\end{proof}

\begin{remark}\label{delsarte}
\begin{enumerate}
\item In \eqref{picard1}, $\{F_l^{\phi}\}_{k\ge1}$ is image by $\phi$ of $\{F_i\}_{i\ge0}$ in \eqref{picard0}, being viewed as a generalization on noncommutative variable of the Fredholm like transformation, so-called
functional rotation of sequence (of orthogonal functions) with the kernel of rotation $K(s,t)$ \cite{delsarte},
and $M_{n-1}^{\phi}$ is a generalization of such kernel:
\begin{eqnarray*}
\varphi(s)=f(s)+\int_a^bK(s,t)f(t)dt.
\end{eqnarray*}

\item $\sum_{m\ge0}V_m$ is called Volterra expansion (like) of ${\bf d}F=\Omega_nF$ \cite{these,livre}, \textit{i.e.}
\begin{eqnarray*}
\sum_{m\ge0}V_m=V_0H,&\mbox{with the Volterra kernels}&\biggl\{\sum_{w\in\calT_n^m}\kappa_w\biggr\}_{m\ge0}.
\end{eqnarray*}
Replacing letters, in \eqref{NCDE}--\eqref{split}, by vector fields or matrices (see also Remark \ref{Hypergeometric}),
the sequence $\{F_i\}_{i\ge0}$ with matrices in \eqref{picard0} yields the so-called Dyson series associated to \eqref{NCDE} \cite{cartier1,dyson}.
\end{enumerate}
\end{remark}

\begin{corollary}\label{finitefactorization}
With Notations in Definition \ref{Chenseriesdef} and Theorem \ref{Chen_braids}, one has the following
\begin{enumerate}
\item infinite factorization of Chen series:
\begin{eqnarray*}
C_{\varsigma\path z}=\prod_{l\in\Lyn\calT_n}^{\searrow}
e^{\alpha_{\varsigma}^z(S_l)P_l}\in e^{\ncs{\calL ie_{\calA}}{\calT_n}}\quad{\mbox{(decreasing lexicographical}\atop\mbox{ordered product).}}
\end{eqnarray*}
\item finite factorization of Chen series (see also \eqref{V_0} and Remark \ref{nil_app})\footnote{\label{encore}
This can be also considered as \textit{d\'evissage} (see Section \ref{intro}) and recurssively done.}:
\begin{eqnarray*}
C_{\varsigma\path z}=V_0(\varsigma,z)H(\varsigma,z)
\end{eqnarray*}
and then $H(\varsigma,z)\in e^{\ncs{\calL ie_{\calA}}{\calT_n}}$, being $V_0^{-1}(\varsigma,z)C_{\varsigma\path z}$ and satisfying \eqref{NCDE1}.
\end{enumerate}
\end{corollary}

\begin{proof}
These are classic for Chen series (see \cite{these,livre} for example), using
\begin{enumerate}
\item Proposition \ref{melange0}.\ref{shuffle}, the Friedrichs criterion \cite{reutenauer} and \eqref{grouplike}.
\item Theorem \ref{Chen_braids} and then \eqref{V_0}.
\end{enumerate}
\end{proof}

\section{Application to Knizhnik-Zamolodchikov equations}\label{free}
\subsection{Noncommutative generating series of polylogarithms}
For\footnote{As universal differential equation with three singularities, $KZ_3$ leads to the study, substituting letters by matrices of dimension $2$, of hypergeometric functions (and the group $\mathfrak{sl_2}$) \cite{etingof}. In \cite{Terasoma}, matrices in $\calM_{k!,k!}(\C),k\ge2,$ (considered again as letters) lead to Selberg integrals over $k-1$ marked points on the sphere or disk.} $KZ_3$ (see Examples \ref{$KZ_3$}--\ref{$KZ_3$bis}), essentially interested in solutions of \eqref{$(DE)$} over $]0,1[$ and via the involution $s\longmapsto1-s$, Dridfel'd advocated the following solution in $\ncs{\calH(\widetilde{\C_*^3})}{\calT_3}$ \cite{drinfeld1}:
\begin{eqnarray}\label{solDrinfeld}
F(z)=(z_1-z_2)^{(t_{1,2}+t_{1,3}+t_{2,3})/2{\rm i}\pi}G((z_3-z_2)/(z_1-z_2)),
\end{eqnarray}
where $G$, belonging to $\ncs{\calH(\widetilde{\C\setminus\{0,1\}})}{t_{1,2},t_{2,3}}$, satisfies the
noncommuative differential using the connection $N_2$ determined in Example \ref{$KZ_3$}
\begin{eqnarray}\label{$(DE)$}
dG(s)=N_2(s)G(s).
\end{eqnarray}
Without explaining any method to obtain\footnote{In \cite{drinfeld1},
neither be constructed such expression of $\Phi_{KZ}$ nor be made explicit $G_0$ or $G_1$.

A proof that \eqref{solDrinfeld} is the limit of $\{V_l\}_{l\ge0}$ (in Example \ref{$KZ_3$bis}) is provided in Appendix \ref{AppendixB}.

See also \eqref{abelianization} below for an approximation solution of \eqref{$(DE)$}--\eqref{asymcond} and an identification
of the coefficients of $\log\Phi_{KZ}$ in \cite{drinfeld1}.} \eqref{solDrinfeld}, he stated that \eqref{$(DE)$} admits a unique
solution, $G_0$ (resp. $G_1$), satisfying the following asymptotic condition \cite{drinfeld1}
\begin{eqnarray}\label{asymcond}
G_0(s)\sim_0e^{x_0\log(s)}=s^{x_0}&\mbox{(resp. }G_1(s)\sim_1e^{-x_1\log(1-s)}=(1-s)^{-x_1}),
\end{eqnarray}
and there is unique grouplike series $\Phi_{KZ}\in\ncs{\R}{X}$ such that $G_0=G_1\Phi_{KZ}$.
This series satisfies a system of algebraic relations (duality, hexagonal and pentagonal) \cite{cartier2,drinfeld1},
so-called Drinfel'd series or Drinfel'd associator \cite{cartier2}.

In \cite{drinfeld1}, the coefficients $\{c_{k,l}\}_{k,l\ge0}$ of $\log\Phi_{KZ}$ are identified as follows
\begin{itemize}
\item Setting $A:=t_{1,2},B:=t_{2,3}$ and supposing that $[A,B]=0$, Drinfel'd proposed $z^{A/2{\rm i}\pi}(1-z)^{B/2{\rm i}\pi}$
as solution\footnote{In \cite{drinfeld1}, solution for \eqref{$(DE)$}--\eqref{asymcond} and method providing \eqref{solDrinfeld}
was not described.} of \eqref{$(DE)$}, over $]0,1[$, satisfying standard asymptotic conditions \eqref{asymcond}.
Such approximation solution of $KZ_3$ (a grouplike series on $\ncs{\calH(\widetilde{\C_*^3})}{\calT_3}$) for which the logarithm
belongs then to the following partial abelianization (see also Remark \eqref{BE} below)
\begin{eqnarray}\label{abelianization}
\ncs{\calL ie_{\calH(\widetilde{\C_*^3})}}{t_{1,2},t_{1,3},t_{2,3}}/
[\ncs{\calL ie_{\calH(\widetilde{\C_*^3})}}{t_{1,2},t_{2,3}},\ncs{\calL ie_{\calH(\widetilde{\C_*^3})}}{t_{1,2},t_{2,3}}]
\end{eqnarray}
and will be examined, as application of \eqref{S_k} and \eqref{hatV_0}, in Section \ref{Knizhnik-Zamolodchikov}.

\item Then setting $\bar A=A/2{\rm i}\pi$ and $\bar B=B/2{\rm i}\pi$, he also proposed, over $]0,1[$,
the standard solutions $G_0=z^{\bar A}(1-z)^{\bar B}V_0(z)$ and $G_1=z^{\bar A}(1-z)^{\bar B}V_1(z)$,
where $V_0$ and $V_1$ have continuous extensions to $]0,1[$ and is grouplike solution of the following
noncommutative differential equation, with $V_0(0)=V_1(1)=1$
in the topological free Lie algebra, $\mathfrak p:=\mathrm{span}\{\ad_A^k\ad_B^l[A,B]\}_{k,l\ge0}$,
\begin{eqnarray}\label{approx_Drinfeld}
{\bf d}S(z)=Q(z)S(z),&\mbox{where}&Q(z):=e^{\ad_{-\log(1-z)\bar B}}e^{\ad_{-\log(z)\bar A}}\Frac{\bar B}{z-1}\in\mathfrak p.
\end{eqnarray}

\item Since $G_0=G_1\Phi_{KZ}$ then $\Phi_{KZ}=V(0)V(1)^{-1}$, where $V$ is a solution of \eqref{approx_Drinfeld}
and then, by identification 
in the abelianization $\mathfrak p/[\mathfrak p,\mathfrak p]$, as follows
\begin{eqnarray}\label{logPhiKZ}
\log\Phi_{KZ}=\sum_{k,l\ge0}c_{k,l}B^{k+1}A^{l+1}=\int_0^1Q(z)dz\mod[\mathfrak p,\mathfrak p]\cr
=\int_0^1e^{\ad_{-\log(1-z)\bar B}}e^{\ad_{-\log(z)\bar A}}\frac{\bar Bdz}{z-1}\mod[\mathfrak p,\mathfrak p]
\end{eqnarray}
and by serial expansions of exponentials, one deduces that
\begin{eqnarray}
&&\log\Phi_{KZ}=\sum_{k,l\ge0}\frac1{l!k!}\int_0^1\log^l\frac1{1-z}\log^k\Big(\frac1z\Big)
\ad_{\bar B^k\bar A^l}\bar B\frac{dz}{z-1}\mod[\mathfrak p,\mathfrak p].
\end{eqnarray}

\item The following divergent (iterated) integral is regularized\footnote{The readers are invited to consult \cite{CM} for a
comparison of these regularized values yielding expressions of $\Phi_{KZ}$ and $\log\Phi_{KZ}$, in which involve polyzetas.} by
\begin{eqnarray}
&\displaystyle
c_{k,l}=\frac{1}{(2{\rm i}\pi)^{k+l+2}(k+1)!l!}\int_0^1\log^l\Big(\frac1{1-z}\Big)\frac{dz}{z-1}
&\Big(\bar B^k\bar A^l\bar B=\frac{B^kA^lB}{(2{\rm i}\pi)^{k+l+1}}\Big)
\end{eqnarray}
and, by a Legendre's formula\footnote{\textit{i.e.}
the Taylor expansion of $\log\Gamma(1-z)$ involving only the real numbers $\{\zeta(k)\}_{k\ge2}$ and $\gamma$
(as regularized value of the harmonic series $1+2^{-1}+3^{-1}+\ldots$).}, Drinfel'd stated that previous process is
equivalent to the following identification\footnote{Note that the summation on right side starts with $n=2$
and then $\gamma$ could not be appeared in the regularization proposed in \cite{drinfeld1}.} \cite{drinfeld1}:
\begin{eqnarray}
1+\sum_{k,l\ge0}c_{k,l}B^{k+1}A^{l+1}=\exp\sum_{n\ge2}\frac{\zeta(n)}{(2{\rm i}\pi)^nn}(B^n+A^n-(B+A)^n).
\end{eqnarray}
\end{itemize}

With $X=\{x_0,x_1\}$ ($x_0\prec x_1$), the series $\Phi_{KZ}$ is completely studied using polylogarithms defined by $\Li_{1_{X^*}}=1_{\calH(\widetilde{\C\setminus\{0,1\}})},\allowbreak\Li_{x_0}(s)=\log(s),\Li_{x_1}(s)=\log(1-s)$ and, for any $x_iw\in\Lyn X\setminus X$, (see \cite{CM})
\begin{eqnarray}
\Li_{x_iw}(s)=\int_0^s\omega_i(\sigma)\Li_w(\sigma),
&\mbox{where}&\left\{\begin{array}{l}\omega_0(s)=s^{-1}ds,\cr\omega_1(s)=(1-s)^{-1}ds.\end{array}\right.\label{Li}
\end{eqnarray}
In particular, $\{\Li_l\}_{l\in\Lyn X}$ (resp. $\{\Li_w\}_{w\in X^*}$) is algebraically (resp. linearly) free, over $\C$, and the noncommutative series of $\{\Li_w\}_{w\in X^*}$ is grouplike (see Proposition \ref{MRSBTT}), as being the actual solution of \eqref{$(DE)$} satisfying the asymptotic conditions \eqref{asymcond} \cite{FPSAC98,CM} (using the decreasing lexicographical order product)
\begin{eqnarray}
&&\L:=\sum_{w\in X^*}\Li_ww\label{L}=\prod_{l\in\Lyn X}^{\searrow}e^{\Li_{S_l}P_l}
\mbox{ and }\left\{\begin{array}{lcc}
\lim\limits_{s\to0}\L(s)e^{-x_0\log(s)}=1_{X^*},\cr
\lim\limits_{s\to1}e^{x_1\log(1-s)}\L(s)=\Phi_{KZ},
\end{array}\right.\label{sgLi}
\end{eqnarray}
where $\{P_l\}_{l\Lyn X}$ (resp. $\{S_l\}_{l\Lyn X}$) is linear basis of $\ncp{\calL ie_{\Q}}{X}$ (resp. $\mathrm{Sh}_{\Q}(X)$) and
\begin{eqnarray}\label{Z}
\Phi_{KZ}:=\prod_{l\in\Lyn X\setminus X}^{\searrow}e^{\Li_{S_l}(1)P_l},&\mbox{with}&
\left\{\begin{array}{lcr}x_0&=&t_{1,3}/2{\rm i}\pi,\cr x_1&=&-t_{2,3}/2{\rm i}\pi,
\end{array}\right.
\end{eqnarray}
admitting $\{\Li_l(1)\}_{l\in\Lyn X\setminus X}$ as convergent\footnote{For this point, Lyndon words are more efficient for
checking the convergence of $\{\Li_w(1)\}_{w\in X^*}$ (see \cite{CM}) using a Radford's theorem \cite{reutenauer}.}
 coordinates and the coordinates $\{\scal{\Phi_{KZ}}{w}\}_{w\in X^*}$ as the finite parts\footnote{These coefficients
are convergent and regularized divergent polyzetas \cite{CM,lemurakami}.} of the singular expansions at $z=1$ of $\{\Li_w\}_{w\in X^*}$
in the comparison scale $\{(1-z)^{-a}\log^b(1-z)\}_{a,b\in\N}$ (see \eqref{sgLi}). Moreover, in virtue of
\eqref{sgLi}, $\L(({z_3-z_2})/({z_1-z_2}))$ is grouplike solution of $KZ_3$. So does \eqref{solDrinfeld}, for which any other
grouplike solution of $KZ_3$ can be deduced by right multiplication by constant grouplike series as treated in Appendix \ref{AppendixB} below.

\subsection{Noncommutative generating series of hyperlogarithms}
Recall also that, after $KZ_3$, Dridfel'd proposed asymptotic solutions, for $KZ_4$, on different zones in the
region $\{z\in\R^4|z_1<z_2<z_3<z_4\}$ \cite{drinfeld1} and exact solutions, as in \eqref{solDrinfeld},
are not provided yet. It was a break with respect to the strategy in previous cases.
Several works tried to advance on the resolution of $KZ_n$ (for $n\ge4$). Indeed, it was studied the Dirichlet
functions $\{\mathrm{Di}_w(F;s)\}_{w\in X}$ (and their parametrization)
indexed by words in $X=\{x_i\}_{0\le i\le N}$ (totally ordered by $x_0\prec\ldots\prec x_N$),
\textit{i.e.} iterated integrals of the following holomorphic $1$-forms \cite{FPSAC95,FPSAC96}
\begin{eqnarray}\label{omega}
\omega_0(s)=\frac{ds}s,&\omega_i(s)=F_i(s)ds,&\mbox{where }F_i(s)=\sum_{k\ge1}f_{i,k}z^k,0\le i\le N.
\end{eqnarray}
In particular, for singularities in $\Sigma_N=\{0,a_1,\ldots,a_N\}$ (in bijection with $X$) and
\begin{eqnarray}\label{omega1}
F_i(s)=(s-a_i)^{-1},&0\le i\le N,
\end{eqnarray}
these correspond to Lappo-Danilevsky's hyperlogarithms\footnote{and, of course, colored polylogarithms for the case of roots of unity, \textit{i.e} $a_i=e^{2\mathrm{i}\pi}/N$ \cite{legsloops}.} \cite{Linz}
Moreover, abuse ratings for convenience, hyperlogarithms are defined by $\Li_{1_{X^*}}=1_{\calH(\widetilde{\C\setminus\Sigma_N})}$ and $\Li_{x_i}(s)=\log(s-a_i)$ ($1\le i\le N$) and, for any Lyndon work $x_iw\in\Lyn X\setminus X$, by
\begin{eqnarray}\label{hyperlogarithms}
\Li_{x_iw}(s)=\int_0^s\omega_i(\sigma)\Li_w(\sigma),&\mbox{where}&\omega_i(s)=\frac{ds}{s-a_i}.
\end{eqnarray}

These hyperlogarithms $\{\Li_l\}_{l\in\Lyn X}$ (resp. $\{\Li_w\}_{w\in X^*}$) are algebraically (resp. linearly) free over $\C$ \cite{Linz},
\textit{i.e.} the character $\Li_{\bullet}$ of $(\ncp{\C}{X},\shuffle,1_{X^*})$ (see \eqref{hyperlogarithms}) is injective and its graph,
viewed as noncommutative generating series, is grouplike and can be put (using the decreasing lexicographical order product) in the MRS form as follows \cite{Linz}
(see also Proposition \ref{MRSBTT} below)
\begin{eqnarray}\label{sghyperlog}
\L:=\sum_{w\in X^*}\Li_ww=\prod_{l\in\Lyn X}^{\searrow}e^{\Li_{S_l}P_l}.
\end{eqnarray}
This series belongs to $\ncs{\calH(\widetilde{\C\setminus\Sigma_N})}{X}$ (while, as already said, solutions of \eqref{KZn} belong
to $\ncs{\calH(\widetilde{\C_*^n})}{\calT_n}$) and, by \eqref{omega}--\eqref{omega1}, satisfies the following differential equation
\begin{eqnarray}\label{DEN}
d\L(s)=(x_0\omega_0(s)+x_1\omega_1(s)+\ldots x_N\omega_N(s))\L(s),
\end{eqnarray}
and quite involves in the resolution of \eqref{KZn} according to \eqref{N}--\eqref{Nbis}.
Indeed, for $N=n-2,a_k=z_k$ ($1\le k\le n-2$) and substituting
$x_0={t_{n-1,n}}/{2{\rm i}\pi},x_k=-{t_{k,n}}/{2{\rm i}\pi}$
(for $k=1,..,n-2$), $\bar M_n$ in \eqref{NCDE} induces the following simpler expression for $N_{n-1}$ (given in \eqref{Nbis})
as the connection of \eqref{DEN} satisfied by $\L$ (given in \eqref{hyperlogarithms}--\eqref{sghyperlog}):
\begin{eqnarray}\label{Nter}
N_{n-1}(s)=x_0\frac{ds}{s}+\sum_{k=1}^{n-2}x_k\frac{ds}{a_k-s}&\mbox{and then}&d\L(s)=N_{n-1}(s)\L(s).
\end{eqnarray}
This showed, in fact, the grouplike series $\L$ in \eqref{sghyperlog} (resp. \eqref{sgLi}) is not but normalizes
the Chen series, of $\{\omega_i\}_{0\le i\le N}$ in \eqref{omega1} (resp. $\{\omega_i\}_{0\le i\le1}$ in \eqref{Li})
and along $0\path z$, in which the integral $\int\limits_0^z\omega_0(s)$, for example, is not defined.

\subsection{Knizhnik-Zamolodchikov equations}\label{Knizhnik-Zamolodchikov}
Ending this note, let $p$ be the projection $\widetilde{\C_*^n}\longrightarrow\C_*^n$ and let us consider the following affine plans
\begin{eqnarray}
(P_{i,j}):z_i-z_j=1,&\mbox{for}&1\le i<j\le n.
\end{eqnarray}

Let us consider
\begin{eqnarray}\label{data}
\left\{\begin{array}{lcl}
u_{i,j}(z)=(z_i-z_j)^{-1},&\mbox{for}&1\le i,j\le n,\cr
\omega_{i,j}(z)=u_{i,j}(z)d(z_i-z_j),&\mbox{for}&1\le i<j\le n,
\end{array}\right.\end{eqnarray}
and then the Chen series $C_{z^0\path z}$, of the holomorphic $1$-forms $\{d\log(z_i-z_j)\}_{1\le i<j\le n}$ and
along the path $z^0\path z$ over $\calV:=\widetilde{\C_*^n}$. As in Section \ref{intro}, let $\calA:=\calH(\calV)$.

\begin{remark}\label{rk2}
Let $k\ge1,t_{i,j}\in\calT_n,z^0\in P_{i,j}$. Then\footnote{
$\begin{array}{lcr}
\log(z_i-z_j)=\sum_{k\ge1}(-1)^{k-1}((z_i-z_j)-1)^k/k,&\mbox{for}&\abs{z_i-z_j}<1.
\end{array}$} $\alpha_{z^0}^z(t_{i,j}^k)=\log^k(z_i-z_j)/{k!}$.
\end{remark}

\begin{definition}[normalized Chen series]\label{F}
Let $F_{\bullet}:(\ncp{\C}{\calT_n},\shuffle,1_{\calT_n^*})\longrightarrow(\calA,*,1_{\calA})$
is the character defined by $F_{1_{\calT_n^*}}=1_{\calA}$ and $F_{t_{i,j}}(z)=\log(z_i-z_j)$ ($t_{i,j}\in\calT_n$) and, for any $t_{i,j}w\in\Lyn\calT_n\setminus\calT_n$ and $z^0$ moving towards $0$, by
\begin{eqnarray*}
F_{t_{i,j}w}(z)=\int_{z^0}^z\omega_{i,j}(s)F_w(s).
\end{eqnarray*}
Let $\F_{KZ_n}$ be the graph of $F_{\bullet}$ (\textit{i.e.}
the noncommutative generating series of $\{F_w\}_{w\in\calT_n^*}$).
\end{definition}

\begin{remark}
\begin{enumerate}
\item If $F\in\calA$ and $F$ is expanded as follows\footnote{The coefficients
$f(n_{i,j};1\le i<j\le n)$'s are indexed by integers $n_{i,j}>0$, for $1\le i<j\le n$.}
\begin{eqnarray*}
F(z)=\sum_{n_{i,j}\ge1\atop1\le i<j\le n}f(n_{i,j};1\le i<j\le n)\prod_{1\le i<j\le n}(z_i-z_j)^{n_{i,j}}
\end{eqnarray*}
then, for any $k\ge0$ and $(i_0,j_0)$ such that $1\le i_0<j_0\le n$, one has
\begin{eqnarray*}
\lim_{z_{j_0}\to z_{i_0}}(z_{i_0}-z_{j_0})^kF(z)=0.
\end{eqnarray*}

\item By a Radford's theorem \cite{reutenauer}, $F_w,w\in\calT_n^*,$ is polynomial on
$\{F_l\}_{l\in\Lyn\calT_n}$ and depends on the differences $\{z_i-z_j\}_{1\le i<j\le n}$.
In particular, for $w\in\calT_n^+$, by induction on $\abs{w},F_w$ can be expanded by (see the previous item)
\begin{eqnarray*}
F_w(z)=\sum_{n_{i,j}\ge1\atop1\le i<j\le n}f_w(n_{i,j};1\le i<j\le n)\prod_{1\le i<j\le n}(z_i-z_j)^{n_{i,j}}
\end{eqnarray*}
and $F_{t_{i,j}^k}(z)=\alpha_{z^0}^z(t_{i,j}^k)$, for $z^0\in P_{i,j},t_{i,j}\in\calT_n,k\ge1$ (see also Remark \ref{rk2}).

\item By \eqref{Freidrichs} and Proposition \ref{MRSBTT}, multiplying on the right of the Chen series,
of $\{d\log(z_i-z_j)\}_{1\le i<j\le n}$ and along $z^0\path z$ over $\widetilde{\C_*^n}$, by
$\F_{KZ_n}(z^0)\in\{e^C\}_{C\in\ncs{\calL ie_{\C}}{\calT_n}}$, $\F_{KZ_n}(z)$ normalizes
$C_{z^0\path z}$ and satisfies \eqref{KZn}.
\end{enumerate}
\end{remark}

According to \eqref{normalization0}--\eqref{normalization} and Theorem \ref{diagonal},
the image of $\calD_{\calT_n}$ by $F_{\bullet}\otimes\mathrm{Id}$ yields

\begin{proposition}[factorizations of normalized Chen series]\label{sol}
\begin{enumerate}
\item One has
\begin{eqnarray*}
\F_{KZ_n}&=&\prod_{l\in\Lyn\calT_{n-1}}^{\searrow}e^{F_{S_l}P_l}
\Big(\prod_{l=l_1l_2\atop{l_2\in\Lyn\calT_{n-1},l_1\in\Lyn T_n}}^{\searrow}e^{F_{S_l}P_l}\Big)
\prod_{l\in\Lyn T_n}^{\searrow}e^{F_{S_l}P_l}\cr
&=&\prod_{l\in\Lyn T_n}^{\searrow}e^{F_{S_l}P_l}\cr
&\times&\underbrace{\Big(1_{\calT_n^*}
+\sum_{v_1,\ldots,v_k\in T_n^*,k\ge1\atop t_1,\ldots,t_k\in\calT_{n-1}}
F_{a(v_1t_1)\halfshuffle\ldots\halfshuffle a(v_kt_k)}r(v_1t_1)\ldots r(v_kt_k)\Big)}_{\mbox{functional expansion of solution of $KZ_{n-1}$}},
\end{eqnarray*}
and, as image by $F_{\bullet}\otimes\mathrm{Id}$ of $\log\calD_{\calT_n}$ in \eqref{logD},
$\log\F_{KZ_n}$ is primitive, for $\Delta_{\shuffle}$.

\item Modulo $[\ncs{\calL ie_{1_{\calA}}}{T_n},\ncs{\calL ie_{1_{\calA}}}{T_n}]$, one also has
\begin{eqnarray*}
\F_{KZ_n}&\equiv&e^{\sum_{t\in T_n}F_tt}\big(1_{\calT_n^*}\\
&+&
\sum_{k\ge1}\sum_{v_1,\ldots,v_k\in T_n^*\atop t_1,\ldots,t_k\in\calT_{n-1}}F_{a(\hat v_1t_1)\halfshuffle(\ldots\halfshuffle(a(\hat v_kt_k))\ldots)}r(v_1t_1)\ldots r(v_kt_k)\big).
\end{eqnarray*}
\end{enumerate}
\end{proposition}

\begin{corollary}\label{P}
With Notation in Example \ref{calC0}, one has
\begin{enumerate}
\item The morphism $F_{\bullet}:(\ncp{\calC_0}{\calT_n},\shuffle)\longrightarrow
(\mathrm{span}_{\calC_0}\{F_w\}_{w\in\calT_n^*},\times)$ is injective.

\item Let $\calK_{T_n}$ and $\calK_{\calT_{n-1}}$ be the algebras generated, respectively, by $\{F_l\}_{l\in\Lyn T_n}$
and $\{F_l\}_{l\in\Lyn\calT_{n-1}}$. Then $\calK_{T_n}$ and $\calK_{\calT_{n-1}}$ are $\calC_0$-algebraically disjoint.

\item There exists $E\in e^{\ncs{\calL ie_{\calK_{T_n}}}{\calT_{n-1}}}$ such that, for $z^0\to0$,
\begin{eqnarray*}
\F_{KZ_{n-1}}(z)E&=&1_{\calT_n^*}+\sum_{k\ge1}\sum_{t_{i_1,j_1}\ldots t_{i_k,j_k}\in\calT_{n-1}^*}
\int_{z^0}^z\omega_{i_1,j_1}(s_1)\ldots\int_{z^0}^{s_{k-1}}\omega_{i_k,j_k}(s_k)\\
&&\hskip58mm\varphi_{T_n}^{({z^0},z)}(t_{i_1,j_1}\ldots t_{i_k,j_k}).\cr
\F_{KZ_n}&=&\Big(\prod_{l\in\Lyn T_n}^{\searrow}e^{F_{S_l}P_l}\Big)\F_{KZ_{n-1}}E\quad{\mbox{(decreasing lexicographical}\atop\mbox{ordered product).}}
\end{eqnarray*}

\item $\{\ad_{-T_n}^{k_1}t_1\ldots\ad_{-T_n}^{k_p}t_p\}_{t_1,\ldots,t_p\in\calT_{n-1}}^{k_1,\ldots,k_p\ge0,p\ge1}$
of $\calU(\calI_N)/[\ncs{\calL ie_{1_{\calA}}}{T_n},\ncs{\calL ie_{1_{\calA}}}{T_n}]$ is dual to
$\{(-t_1\hat T_n^{k_1})\halfshuffle\cdots\halfshuffle(-t_k\hat T_n^{k_p})\}_{t_1,\ldots,t_k\in\calT_{n-1}}^{k_1,\ldots,k_p\ge0,p\ge1}$
of $\calU(\calI_N)^{\vee}$.
\end{enumerate}
\end{corollary}

\begin{proof}
These are consequences of Propositions \ref{MRSBTT}--\ref{sol}, Corollary \ref{Chensolution} and Theorem \ref{Chen_braids}.
\end{proof}

In order to examine grouplike solutions of $KZ_n$ with asymptotic conditions by \textit{d\'evissage}, let us consider again the alphabet
$\calT'_n=\{t_{i,j}\}_{1\le i,j\le n}$ satisfying \eqref{braid} and\footnote{$\{\int_{z_0}^zu_{i,j}(s)d(s_i-s_j)\}_{1\le i,j\le n}$
is not $\C$-linearly free since $u_{i,j}(s)d(s_i-s_j)=u_{j,i}(s)d(s_j-s_i)$.}
\begin{eqnarray}
U_i:=\sum_{j=1,j\ne i}^nt_{i,j}u_{i,j},&1\le i\le n.
\end{eqnarray}
With the split \eqref{split}, \textit{i.e.} $M_n=\bar M_n+M_{n-1}$, and the data in \eqref{data}, one has
\begin{eqnarray}\label{splitbis}
&\bar M_n=\Sum_{k=1}^{n-1}t_{k,n}\frac{d(z_k-z_n)}{z_k-z_n},
&M_n=\Sum_{1\le i<j\le n}t_{i,j}\Frac{d(z_j-z_i)}{z_j-z_i}=\Sum_{i=1}^nU_i(z)\;dz_i.
\end{eqnarray}
Moreover, as in \eqref{N}--\eqref{Nbis}, $\bar M_n$ behaves, for\footnote{See Note \ref{NOTE}.} $z_n\to z_{n-1}$, as the following connection
\begin{eqnarray}\label{Nsans}
N_{n-1}(s)=t_{n-1,n}\frac{ds}{s}-\sum_{k=1}^{n-2}t_{k,n}\frac{ds}{a_k-s},
&\mbox{with}&\left\{\begin{array}{r}s=z_n,\cr a_k=z_k.\end{array}\right.
\end{eqnarray}

\begin{proposition}\label{crochet}
\begin{enumerate}
\item\label{un} The family $\{U_i\}_{1\le i\le n}$ satisfies
\begin{eqnarray*}
\sum_{i=1}^nU_i=0,
&\Sum_{i=1}^nz_iU_i(z)=\sum_{1\le i<j\le n}t_{i,j},
&\mathbf{\partial}_iU_j-\mathbf{\partial}_jU_i=[U_i,U_j]=0.
\end{eqnarray*}

\item If $G$ is solution of \eqref{NCDE} then it satisfies the following identities
\begin{eqnarray*}
\sum_{i=1}^n\mathbf{\partial}_iG(z)=0&\mbox{and}&
\sum_{i=1}^nz_i\mathbf{\partial}_iG(z)=\sum_{1\le i<j\le n}t_{i,j}G(z)
\end{eqnarray*}
and the partial differential equations $\mathbf{\partial}_iG=U_iG$, for $i=1,..,n$.

\item One has $M_n\wedge M_n=0$ and ${\bf d}M_n=0$ and then ${\bf d}\bar M_n=0$.

\item One has ${\bf d}\Omega_n-\Omega_n\wedge\Omega_n=0$ (see \eqref{integrable}) and ${\bf d}\bar\Omega_n=0$.
\end{enumerate}
\end{proposition}

\begin{proof}
\begin{enumerate}
\item Since $u_{i,j}=-u_{j,i}$ then
\begin{eqnarray*}
\sum_{i=1}^nU_i=\sum_{i=1}^n\sum_{1\le j<i\le n}(t_{i,j}-t_{j,i})u_{i,j}.
\end{eqnarray*}
By the infinitesimal braid relations given in \eqref{braid}, we get the first identity.

For the second identity, using a change of indices as follows
\begin{eqnarray*}
\sum_{i=1}^nz_iU_i(z)
=&\Sum_{i=1}^nt_{i,j}\Big(\Sum_{1\le i<j\le n}\Frac{z_i}{z_i-z_j}-\Sum_{1\le j<i\le n}\frac{z_i}{z_j-z_i}\Big)&\cr
=&\Sum_{i=1}^nt_{i,j}\Big(\Sum_{1\le i<j\le n}\Frac{z_i}{z_i-z_j}-\Frac{z_j}{z_i-z_j}\Big)&=\sum_{1\le i<j\le n}t_{i,j}.
\end{eqnarray*}
The third identity is obtained by direct calculations:
\begin{eqnarray*}
\mathbf{\partial}_iU_j-{\bf\partial}_jU_i
&=&\sum_{1\le l\le n\atop l\neq j}t_{j,l}(\partial_iu_{j,l})
-\sum_{1\le k\le n\atop k\neq i}t_{i,k}(\partial_ju_{i,k})\\
&=&-t_{j,i}(z_j-z_i)^{-2}+t_{i,j}(z_i-z_j)^{-2}\cr
[U_i,U_j]
&=&\sum_{1\le k,l\le n\atop i\neq j\neq k\neq l}[t_{i,k},t_{j,l}]u_{i,k}u_{j,l}
+\sum_{1\le k\le n\atop k\neq i,j}[t_{i,k},t_{j,l}]u_{i,k}u_{j,l}\\
&+&\sum_{1\le k\le n\atop k\neq i}[t_{i,j},t_{j,k}]u_{i,j}u_{j,k}
+\sum_{1\le k\le n\atop k\neq j}[t_{i,k},t_{j,i}]u_{i,k}u_{j,i}\cr
&=&\sum_{1\le k,l\le n\atop i\neq j\neq k\neq l}[t_{i,k},t_{j,l}]u_{i,k}u_{j,l}
+\sum_{1\le k\le n\atop k\neq i,j}(z_i[t_{j,k},t_{j,i}+t_{k,l}]\\
&+&z_j[t_{i,k},t_{i,j}+t_{k,j}]+z_k[t_{i,j},t_{i,k}+t_{j,k}])u_{i,k}u_{j,k}u_{j,i}.
\end{eqnarray*}
By infinitesimal braid relations in \eqref{braid}, one gets $\mathbf{\partial}_iU_j-\mathbf{\partial}_jU_i=[U_i,U_j]=0$.

\item The first identities are consequences of the item \ref{un}. By \eqref{splitbis}, one deduces
\begin{eqnarray*}
&&{\bf d}G(z)
=\Big(\sum_{i=1}^nU_i(z)\;dz_i\Big)G(z)
=\sum_{i=1}^n(U_i(z)G(z))\;dz_i
=\sum_{i=1}^n(\mathbf{\partial}_iG(z))\;dz_i
\end{eqnarray*}
and by \eqref{dS}, one obtains the last result.

\item By \eqref{splitbis} and the item \ref{un} of Proposition \ref{crochet}, one obtains
\begin{eqnarray*}
M_n(z)\wedge M_n(z)
=&\Sum_{i,j=1}^nU_i(z)U_j(z)\;dz_i\wedge dz_j\cr
=&\Sum_{1\le i<j\le n}[U_i(z),U_j(z)]\;dz_i\wedge dz_j&=0,\cr
{\bf d}M_n(z)=&\Sum_{i,j=1}^n(\mathbf{\partial}_iU_j(z)-\mathbf{\partial}_jU_i(z))\;dz_i\wedge dz_j&=0.
\end{eqnarray*}
and, on the other hand,
${\bf d}\bar M_n={\bf d}(M_n-M_{n-1})={\bf d}M_n-{\bf d}M_{n-1}=0$.

\item Substituting $t_{i,j}$ by $t_{i,j}/2\mathrm{i}\pi$ on $M_n$ and $\bar M_n$, one gets the expected results.
\end{enumerate}
In all the sequel, as for \eqref{integrable}, the letters in $\calT_n$ satisfy now \eqref{braidbis}.
\end{proof}

\begin{remark}\label{BE}
With data in \eqref{data} and by Proposition \ref{crochet}, $\Omega_n$ is flat and ${\bf d}S=\Omega_nS$ is completely integrable (see also \eqref{integrable}). On the other side, $\bar\Omega_n$ is not flat and ${\bf d}S=\bar\Omega_nS$ is not
completely integrable. Indeed, one has ${\bf d}\bar M_n=0$ and\footnote{Observed by
B. Enriquez, using the $\C$-linear independence of $\{\log(z_i-z_n)\}_{1\le i\le n-1}$.}
\begin{eqnarray*}
\bar M_n\wedge\bar M_n
=&\Sum_{1\le i,j\le n-1}t_{i,n}t_{j,n}\;d\log(z_i-z_n)\wedge d\log(z_j-z_n)\cr
=&\Sum_{1\le i<j\le n-1}[t_{i,n},t_{j,n}]\;d\log(z_i-z_n)\wedge d\log(z_j-z_n)&\neq0.
\end{eqnarray*}

Getting flatness of $\bar M_n$, one could further assume that $\{t_{i,n}\}_{1\le i\le n-1}$ commute,
\textit{i.e.} $[t_{i,n},t_{j,n}]=0$, as done in the definition of $\hat V_0$ in \eqref{hatV_0} and then
in Definition \ref{phi} using $\hat\varphi_{T_n}$ and $\hat\varphi_n$, as done in Propositions \ref{volterraexp}--\ref{sol} and Theorem \ref{Chen_braids} (see also \eqref{abelianization}).
\end{remark}

Now, we are in situation back to \eqref{KZn} and its solutions with asymptotic conditions,
by Definitions \ref{phi}--\ref{F} and Propositions \ref{sol}--\ref{crochet}, to achieve our application.

\begin{theorem}[d\'evissage]\label{KZsol}
With Definition \ref{phi} and data in \eqref{data}, grouplike solution\footnote{For $1\le i<j\le n$,
changing $t_{i,j}$ by $t_{i,j}/2\mathrm{i}\pi$ (thus $\bar M_n$ and $\bar M_{n-1}$ become $\bar\Omega_n$
and $\bar\Omega_{n-1}$, respectively), one deduces results for \eqref{KZn}.} of \eqref{NCDE} can be put in the
form $h(z_n)H(z_1,\ldots,z_{n-1})$ such that, for $z_n\to z_{n-1}$,
\begin{enumerate}
\item $h$ is solution of\footnote{See Note \ref{NOTE} and Remark \ref{nil_app}.} $df=N_{n-1}f$,
where $N_{n-1}$ is the connection determined in \eqref{Nsans}.
Hence, $h(z_n)\sim_{z_n\to z_{n-1}}(z_{n-1}-z_n)^{t_{n-1,n}}$.
\item $H(z_1,\ldots,z_{n-1})$ satisfies ${\bf d}S=M_{n-1}^{\varphi_n}S$, \textit{i.e.} \eqref{NCDE1} with $\phi=\varphi_n$, and
\begin{eqnarray*}
&M_{n-1}^{\varphi_n^{(z^0,z)}}(z)=\Sum_{1\le i<j\le n-1}d\log(z_i-z_j)\varphi_n^{(z^0,z)}(t_{i,j}),&\cr
&\varphi_n^{(z^0,z)}(t_{i,j})\sim_{z_n\to z_{n-1}}e^{\ad_{-\log(z_{n-1}-z_n))t_{n-1,n}}}t_{i,j}\mod\calJ_{\calR_n}.&
\end{eqnarray*}
Moreover, $M_{n-1}^{\varphi_{n-1}}$ exactly coincides with $M_{n-1}$ in $\bigcap_{1\le k<n-1}(P_{k,n-1})$.
\end{enumerate}                                                                                                         

Conversely, for $z_n\to z_{n-1}$, if $h$ satisfies $df=N_{n-1}f$ and $H(z_1,\ldots,z_{n-1})$ satisfies \eqref{NCDE1} then $h(z_n)H(z_1,\ldots,z_{n-1})$ is solution of \eqref{NCDE}.
\end{theorem}

\begin{proof}
For $z_n\to z_{n-1}$, on the one hand, $h\equiv V_0$ and it behaves as generating series of hyperlogarithms (\textit{i.e.}
iterated integrals of holomorphic forms $\{ds/(s-s_k)\}_{1\le k<n}$, with the singularities $s_k=z_n-z_k$, see Remarks
\ref{Hypergeometric} and \ref{delsarte}). It follows then the first assertion. On the other hand, with
$\varphi_n=\varphi_{T_n}\mod\calJ_{\calR_n}$ as in Definition \ref{phi}, the Picard's iteration \eqref{picard1} converges,
for the discrete topology, to a solution of \eqref{NCDE1} having the expected connection:
\begin{eqnarray*}
H(z_1,\ldots,z_{n-1})&=&\sum_{m\ge0}\sum_{t_{i_1,j_1}\ldots t_{i_m,j_m}\in\calT_{n-1}^*}
\int_{z^0}^zd\log(s_{i_1}-s_{j_1})\varphi_n^{({z^0},s_1)}(t_{i_1,j_1})\ldots\\
&&\int_{z^0}^{s_{m-1}}d\log(s_{i_m}-s_{j_m})\varphi_n^{({z^0},s_m)}(t_{i_m,j_m}),\cr
\varphi_n^{({z^0},z)}(t_{i,j})&=&\Prod_{l\in\Lyn T_n}^{\searrow}e^{\ad_{-F_{S_l}(z)P_l}}t_{i,j}\mod\calJ_{\calR_n}\\
&\sim&e^{\ad_{-\log(z_{n-1}-z_n)t_{n-1,n}}}t_{i,j}\mod\calJ_{\calR_n},\quad z_n\to z_{n-1}.
\end{eqnarray*}

Conversely, let $C\in\ncs{\C}{\calT_{n-1}}/\calJ_{\calR_{n-1}}$ such that $\scal{C}{1_{\calT_{n-1}^*}}=1_{\calA}$.
If $HC$ satisfies \eqref{NCDE1} then, by Propositions \ref{MRSBTT}, $V_0HC$ satisfies \eqref{NCDE}.
\end{proof}

Theorem \ref{KZsol} is established for $z_n\to z_{n-1}$ and, for \textit{d\'evissage}, can be performed
recursively. Up to a permutation of $\mathfrak{S}_{n}$, it can be adapted for other cases. Hence, 

\begin{corollary}[solution of $KZ_n$ satisfying asymptotic condition]\label{C}
$\F_{KZ_n}$ is unique group-like solution of \eqref{NCDE} satisfying
\begin{eqnarray*}
\F_{KZ_n}(z)\sim_{z_i\path z_{i-1}\atop1<i\le n}(z_{i-1}-z_i)^{t_{i-1,i}}G_i(z_1,\ldots,z_{i-1},z_{i+1},\ldots,z_n)
\end{eqnarray*}
in $\ncs{\calA}{\calT_n}/\calJ_{\calR_n}$ and $G_i(z_1,\ldots,z_{i-1},z_{i+1},\ldots,z_n)$ satisfies \eqref{NCDE1}.

Moreover, for
\begin{eqnarray*}
y_1=z_1,\ldots,y_{i-1}=z_{i-1},y_i=z_{i+1},\ldots,y_{n-1}=z_n, \end{eqnarray*}
the connection $M_{n-1}^{\varphi_{n-1}}$ is expressed as follows 
\begin{eqnarray*}
M_{n-1}^{\varphi_n^{(y^0,y)}}(y)=\sum_{1\le i<j\le n-1}d\log(y_i-y_j)e^{\ad_{-\log(y_i-y_n)t_{i,n}}}t_{i,j}\mod\calJ_{\calR_n}
\end{eqnarray*}
and exactly coincides with $M_{n-1}$ in $\bigcap_{1\le k<n-1}(P_{k,n-1})$.
\end{corollary}

\begin{remark}
Historically, noncommutative series were introduced in control theory to study functional expansions (in particular, the Volterra's expansion) of nonlinear dynamical systems via so-called Fliess' generating series of dynamical systems \cite{fliess1,fliess2} which is in duality with Chen series \cite{these,livre}, viewed as series in noncommutative indeterminates (see Definitions \ref{growth}--\ref{Chenseriesdef}, Lemma \ref{condition}, Proposition \ref{melange0}).

After that, Sussmann \cite{Sussmann} gave an infinite product for Chen series using the Hall basis \cite{viennotgerard} and also a noncommutative differential equation, analogous to \eqref{NCDE}. In this context, with the controls $\{u_k\}_{1\le k\le N}$, the differential $1$-forms are of the form $\omega_k(z)=u_k(z)dz$, for $k=1,..,N$ (see also \eqref{Mn}--\eqref{Mn2}). These controls are encoded by the alphabet $X=\{x_k\}_{1\le j\le N}$ (see also \eqref{intermediatealphabet}) and are Lebesgue integrable real-valued functions on the interval $[0,T]$ ($T\in\R_{\ge0}$, is so-called the duration of the controls) and then the Chen series of $\{\omega_k\}_{1\le k\le N}$ belongs to $\ncs{L^{\infty}([0,T],\R)}{X}$ \cite{these}.

More systematically, other finite and infinite products (see Theorem \ref{Chen_braids} and Corollary \ref{finitefactorization} below) were also proposed to obtain functional expansions \cite{these,hoangjacoboussous,FPSAC95,FPSAC96,livre} basing on monoidal factorizations (by Lazard and by Sch\"utzenberger) which were intensively studied earlier in \cite{lothaire,viennotgerard} and are widely exploited in the present work using notations of \cite{berstel,reutenauer}.
\end{remark}

\section{Conclusion}
Basing on the Lazard and Sch\"utzenberger factorizations over the monoid generated by the alphabet $\calT_n=\{t_{i,j}\}_{1\le i<j\le n}$ ($n\ge2$), partitionned into $\calT_{n-1}$ and $T_n=\{t_{k,n}\}_{1\le k\le n-1}$ and, on the other hand, the noncommutative symbolic calculus on $\ncs{\calH(\calV)}{\calT_n}$ (\textit{i.e.} the ring of noncommutative series over $\calT_n$, with holomorphic coefficients in $\calH(\calV)$) \cite{livre}, various combinatorics on Chen series, $C_{\varsigma\path z}$, of the holomorphic $1$-forms $\{\omega_{i,j}\}_{1\le i<j\le n}$ and along a path $\varsigma\path z$ over the simply connected manifold $\calV$ were obtained, by extending \cite{CM}, over $\ncs{\calH(\calV)}{\calT_n}$ and then over $\ncs{\calH(\calV)}{\calT_n}/\calJ_n$, where $\calJ_n$ is the ideal of relators on $\{t_{i,j}\}_{1\le i<j\le n}$.
These are used in order to compute by iterations, over $\ncs{\calH(\calV)}{\calT_n}$, the grouplike solutions and the Galois differential group of the universal differential equation ${\bf d}S=M_nS$ (see \eqref{NCDE}) with the universal connection $M_n$, splitting onto $M_{n-1}$ and $\bar M_n$ (see \eqref{split}).

More precisely, it was focus on the sequences of $\ncs{\calH(\calV)}{\calT_n}$, $\{V_k\}_{k\ge0}$ and $\{\hat V_k\}_{k\ge0}$, satisfying the following recursion
\begin{eqnarray*}
S_k(\varsigma,z)=S_0(\varsigma,z)\sum_{t_{i,j}\in\calT_{n-1}}
\int_{\varsigma}^z\omega_{i,j}(s)S_0^{-1}(\varsigma,s)t_{i,j}S_{k-1}(\varsigma,s),
\end{eqnarray*}
with the following starting conditions, as being grouplike series, for $\Delta_{\shuffle}$,
\begin{eqnarray*}
V_0(\varsigma,z)=\Prod_{l\in\Lyn T_n}^{\searrow}e^{\alpha_{\varsigma}^z(S_l)P_l}&\mbox{and}&
\hat V_0=V_0\mod[\ncs{\calL ie_{\calH(\calV)}}{T_n},\ncs{\calL ie_{\calH(\calV)}}{T_n}].
\end{eqnarray*}

Technically and intensively, in Section \ref{combinatorialframeworks}, with the pairs of dual bases (see \eqref{basisP}--\eqref{basisS} and Definition \ref{dual_bases}) and then applying Lemma \ref{factorization}, Propositions \ref{dual_algebras}--\ref{expression} and
Theorem \ref{diagonal}, various expansions of diagonal series (given in \eqref{diaser}) were provided, in the concatenation-shuffle bialgebra and in a Loday's generalized bialgebra:
\begin{eqnarray*}
\calD_{\calT_n}&=&\calD_{\calT_{n-1}}\Big(\Prod_{l=l_1l_2\atop{l_2\in\Lyn\calT_{n-1},l_1\in\Lyn T_n}}^{\searrow}e^{S_l\otimes P_l}\Big)\calD_{T_n}\cr
&=&\calD_{T_n}\big(1_{\calT_n^*}\otimes1_{\calT_n^*}\\
&+&\sum_{k\ge1}\sum_{v_1,\ldots,v_k\in T_n^*\atop t_1,\ldots,t_k\in\calT_{n-1}}
a(v_1t_1)\halfshuffle(\cdots\halfshuffle a(v_kt_k)\ldots))\otimes r(v_1t_1)\ldots r(v_kt_k)\big).
\end{eqnarray*}
After that, in Sections \ref{universalequation}--\ref{free}, basing on Chen series (see Definition \ref{Chenseriesdef})
and their properties (established in Propositions \ref{melange0}--\ref{MRSBTT} and Corollary \ref{melange} for our needs) and then applying Propositions \ref{volterraexp}--\ref{sol}, Theorems \ref{Chen_braids}--\ref{KZsol} and Corollaries
\ref{P}--\ref{C}, it was proved that
\begin{enumerate}
\item $\sum_{k\ge0}V_k$ converges to $C_{\varsigma\path z}$, \textit{i.e.} the limit of the Picard's iteration in \eqref{picard0}.

\item Specializing $\omega_{i,j}=d\log(z_i-z_j)$ and then $\calV=\widetilde{\C_*^n}$ and reducing by $\calJ_{\calR_n}$,
for $z_n\to z_{n-1}$, $h(z_n)H(z_1,\ldots,z_{n-1})$ is grouplike solution of \eqref{NCDE} such that
\begin{enumerate}
\item $h$ is solution of $df=N_{n-1}f$, where $N_{n-1}$ is the connection determined
in \eqref{Nsans}. Hence, $h(z_n)\sim_{z_n\to z_{n-1}}(z_{n-1}-z_n)^{t_{n-1,n}}$.

\item $H(z_1,\ldots,z_{n-1})$ satisfies ${\bf d}S=M_{n-1}^{\varphi_{n-1}}S$, where
\begin{eqnarray*}
&M_{n-1}^{\varphi_n^{(z^0,z)}}(z)=\Sum_{1\le i<j\le n-1}d\log(z_i-z_j)\varphi_n^{(z^0,z)}(t_{i,j}),&\cr
&\varphi_n^{(z^0,z)}(t_{i,j})\sim_{z_n\to z_{n-1}}e^{\ad_{-\log(z_{n-1}-z_n))t_{n-1,n}}}t_{i,j}\mod\calJ_{\calR_n}.&
\end{eqnarray*}
\end{enumerate}

\item The normalized Chen series (see Definition \ref{F}) provides by \textit{d\'evissage}, over
$\ncs{\calH(\widetilde{\C_*^n})}{\calT_n}$ and then over $\ncs{\calH(\widetilde{\C_*^n})}{\calT_n}/\calJ_{\calR_n}$,
the unique solution of \eqref{KZn} satisfying asymptotic conditions, obtained as image of $\calD_{\calT_n}$,
\begin{eqnarray*}
\F_{KZ_n}&=&\prod_{l\in\Lyn T_n}^{\searrow}e^{F_{S_l}P_l}\cr
&\times&\underbrace{\Big(1_{\calT_n^*}
+\sum_{v_1,\ldots,v_k\in T_n^*,k\ge1\atop t_1,\ldots,t_k\in\calT_{n-1}}
F_{a(v_1t_1)\halfshuffle\ldots\halfshuffle a(v_kt_k)}r(v_1t_1)\ldots r(v_kt_k)\Big)}_{\mbox{functional expansion of solution of $KZ_{n-1}$}}\cr
&=&\prod_{l\in\Lyn T_n}^{\searrow}e^{F_{S_l}P_l}\Big(1_{\calT_n^*}
+\sum_{v_1,\ldots,v_k\in T_n^*,k\ge1\atop t_1,\ldots,t_k\in\calT_{n-1}}\cr
&&F_{a(v_1t_1)\halfshuffle\ldots\halfshuffle a(v_kt_k)}r(v_1t_1)\ldots r(v_kt_k)\Big).
\end{eqnarray*}

\item On the other hand, since $\hat V_0$ is a nilpotent approximation of order $1$ of $V_0$ (see Remark \ref{nil_app})
then, by the families of polynomials, in Definition \ref{dual_bases}, the series on $\{\hat V_k\}_{k\ge0}$
approximates $C_{\varsigma\path z}$ yielding then an approximation solution of $KZ_n$, as extension of a treatment
in \cite{drinfeld1} or in \eqref{abelianization}:
\begin{eqnarray*}
\qquad\quad
\F_{KZ_n}&\equiv&e^{\sum_{t\in T_n}F_tt}\big(1_{\calT_n^*}\\
&+&\sum_{v_1,\ldots,v_k\in T_n^*,k\ge1\atop t_1,\ldots,t_k\in\calT_{n-1}}F_{a(\hat v_1t_1)\halfshuffle(\ldots\halfshuffle(a(\hat v_kt_k))\ldots)}r(v_1t_1)\ldots r(v_kt_k)\big).
\end{eqnarray*}
\end{enumerate}

\bibliographystyle{amsplain}

\section{Appendices}\label{Appendices}
\subsection{$KZ_3$, the simplest non-trivial case}\label{AppendixB}
With the notations given in Example \ref{$KZ_3$bis}, solution of $KZ_3$ is explicit as $F=V_0G$, where $V_0(z)=(z_1-z_2)^{t_{1,2}/2{\rm i}\pi}$ and, similarly as in Proposition \ref{volterraexp},
$G$ is expanded via Corollary \ref{melange} as follows
\begin{eqnarray*}
G(z)=\sum_{m\ge0}\sum_{t_{i_1,j_1}\ldots t_{i_m,j_m}\in\{t_{1,3},t_{2,3}\}^*}
\int_0^z\omega_{i_1,j_1}(s_1)\varphi^{s_1}(t_{i_1,j_1})\ldots\int_0^{s_{m-1}}\\
\omega_{i_m,j_m}(s_m)\varphi^{s_m}(t_{i_m,j_m}),
\end{eqnarray*}
where $\omega_{1,3}(z)=d\log(z_1-z_3)$ and $\omega_{2,3}(z)=d\log(z_2-z_3)$ and
\begin{eqnarray*}
\varphi^z
=e^{\ad_{-(t_{1,2}/2{\rm i}\pi)\log(z_1-z_2)}}
=\sum_{k\ge0}\frac{\log^k(z_1-z_2)}{(-2{\rm i}\pi)^kk!}\ad^k_{t_{1,2}}.
\end{eqnarray*}
One also has
$\varphi^{(\varsigma,s_1)}(t_{i_1,j_1})\ldots\varphi^{(\varsigma,s_m)}(t_{i_m,j_m})
=V_0(z)^{-1}\hat\kappa_{t_{i_1,j_1}\ldots t_{i_m,j_m}}(z,s_1,\cdots,s_m)$.

Moreover, Example \ref{example} (equipping the ordering $t_{1,2}\prec t_{1,3}\prec t_{2,3}$), one has 
\begin{eqnarray*}
\varphi^z(t_{i,3})=\Sum_{k\ge0}\Frac{\log^k(z_1-z_2)}{(-2{\rm i}\pi)^kk!}P_{t_{1,2}^kt_{i,3}},&
\check\varphi^z(t_{i,3})=\Sum_{k\ge0}\Frac{\log^k(z_1-z_2)}{(-2{\rm i}\pi)^kk!}S_{t_{1,2}^kt_{i,3}},
\end{eqnarray*}
where $\check\varphi$ is the adjoint to $\varphi$ and is defined by
\begin{eqnarray*}
\check\varphi^z
=\sum_{k\ge0}\frac{\log^k(z_1-z_2)}{(-2{\rm i}\pi)^kk!}t_{1,2}^k
=e^{-(t_{1,2}/2{\rm i}\pi)\log(z_1-z_2)}.
\end{eqnarray*}

Hence, belonging to $\ncs{\calH(\widetilde{\C_*^3})}{\calT_3}$, $G$ satisfies ${\bf d}G(z)=\bar{\Omega}_2(z)G(z)$, where
$\bar{\Omega}_2(z)=(\varphi^z(t_{1,3})d\log(z_1-z_3)+\varphi^z(t_{2,3})d\log(z_2-z_3))/2{\rm i}\pi$.
In the plane $(P_{1,2}):z_1-z_2=1$, one has $\log(z_1-z_2)=0$ and then $\varphi\equiv\mathrm{Id}$.

Changing $x_0=t_{1,3}/2{\rm i}\pi,x_1=-t_{2,3}/2{\rm i}\pi$ and setting $z_1=1,z_2=0,z_3=s$,
${\bf d}G(z)=\bar\Omega_2(z)G(z)$ is similar to \eqref{$(DE)$}, \textit{i.e.}
\begin{eqnarray*}
\bar{\Omega}_2(z)
=(2{\rm i}\pi)^{-1}(t_{1,3}d\log(z_1-z_3)+t_{2,3}d\log(z_2-z_3))
=x_1\omega_1(s)+x_0\omega_0(s),
\end{eqnarray*}
and admits the noncommutative generating series of polylogarithms as the actual solution satisfying
the asymptotic conditions in \eqref{asymcond}. Thus, by $\L$ given in \eqref{L}, and the homographic
substitution $g:z_3\longmapsto(z_3-z_2)/(z_1-z_2)$,
mapping\footnote{Generally, $s\mapsto(s-a)(c-b)(s-b)^{-1}(c-a)^{-1}$ maps the singularities $\{a,b,c\}$ in $\{0,+\infty,1\}$.}
$\{z_2,z_1\}$ to $\{0,1\}$ (see Examples \ref{$KZ_3$}--\ref{$KZ_3$bis}), a particular solution of $KZ_3$,
in $(P_{1,2})$, is $\L((z_3-z_2)/(z_1-z_2))$. So does\footnote{Note also that these solutions could not be obtained by Picard's iteration in Example \ref{$KZ_3$bis}.

$(z_1-z_2)^{(t_{1,2}+t_{2,3}+t_{1,3})/2{\rm i}\pi}=e^{((t_{1,2}+t_{2,3}+t_{1,3})/2{\rm i}\pi)\log(z_1-z_2)}$,
which is grouplike and independent on the variable $z_3=s$, and then belongs to the differential Galois group
of $KZ_3$.} $\L((z_3-z_2)/(z_1-z_2))(z_1-z_2)^{(t_{1,2}+t_{1,3}+t_{2,3})/2{\rm i}\pi}$.

To end with $KZ_3$, by quadratic relations relations given in \eqref{braidbis}, one has $[t_{1,2}+t_{2,3}+t_{1,3},t]=0$,
for $t\in\calT_3$, meaning that $t$ commutes with $(z_1-z_2)^{(t_{1,2}+t_{2,3}+t_{1,3})/2{\rm i}\pi}$ and then 
$(z_1-z_2)^{(t_{1,2}+t_{1,3}+t_{2,3})/2{\rm i}\pi}$ commutes with $\ncs{\calA}{\calT_3}$. Thus, $KZ_3$ also admits
$(z_1-z_2)^{(t_{1,2}+t_{1,3}+t_{2,3})/2{\rm i}\pi}\L((z_3-z_2)/(z_1-z_2))$ as a particular solution in $(P_{1,2})$.

\subsection{$KZ_4$, other simplest non-trivial case}\label{AppendixC}
For $n=4$, one has $\calT_4=\{t_{1,2},t_{1,3},\allowbreak t_{1,4}, t_{2,3},t_{2,4},t_{3,4}\}$ and then
$\calT_3=\{t_{1,2},t_{1,3},t_{2,3}\}$ and $T_4=\{t_{1,4},t_{2,4},t_{3,4}\}$. Then, by Proposition \ref{volterraexp},
$\varphi^{(\varsigma,z)}_{T_4}=e^{\ad_{-\sum_{t\in T_4}\alpha_{\varsigma}^z(t)t}}$ and $\varphi^{(\varsigma,z)}_{t_{\bullet,4}}(t_{i,j})=\varphi_{T_4}^{(\varsigma,z)}(t_{i,j})$, for
$t_{i,j}\in\calT_4$.

If $z_4\to z_3$ then
$F(z)=V_0(z)G(z_1,z_2,z_3)$, where $V_0(z)=e^{\sum_{1\le i\le4}t_{i,4}\log(z_i-z_4)}$ and $G(z_1,z_2,z_3)$ satisfies ${\bf d}S=M_3^{t_{\bullet,4}}S$ with
\begin{eqnarray*}
M_3^{t_{\bullet,4}}(z)
&=&\varphi^{(z^0,z)}_{t_{\bullet,4}}(t_{1,2})d\log(z_1-z_2)
+\varphi^{(z^0,z)}_{t_{\bullet,4}}(t_{1,3})d\log(z_1-z_3)\\
&+&\varphi^{(z^0,z)}_{t_{\bullet,4}}(t_{2,3})d\log(z_2-z_3).
\end{eqnarray*}
In the intersection $(P_{1,3})\cap(P_{2,3})$, one has $\log(z_1-z_3)=\log(z_2-z_3)=0$ and
$\varphi_{t_{\bullet,4}}\equiv\mathrm{Id}$ and then $M_3^{t_{\bullet,4}}$ exactly coincides with $M_3$.

$F=V_0G$ is solution with $V_0(z)=(z_3-z_4)^{t_{3,4}/2{\rm i}\pi}$ and for $\omega_{i,j}(z)=d\log(z_i-z_j)$ ($1\le i<j\le 4$) and
$\varphi^z=e^{\ad_{-(t_{3,4}/2{\rm i}\pi)\log(z_3-z_4)}}$,
similarly to Proposition \ref{MRSBTT}, one has
\begin{eqnarray*}
G(z)=\sum_{m\ge0,t_{i_1,j_1}\ldots t_{i_m,j_m}\in\atop\{t_{1,2},t_{1,3},t_{2,3},t_{1,4},t_{2,4}\}^*}
\int_0^z\omega_{i_1,j_1}(s_1)\varphi^{s_1}(t_{i_1,j_1})\ldots\int_0^{s_{m-1}}
\omega_{i_m,j_m}(s_m)\varphi^{s_m}(t_{i_m,j_m}).
\end{eqnarray*}
One also has
$\varphi^{(\varsigma,s_1)}(t_{i_1,j_1})\ldots\varphi^{(\varsigma,s_m)}(t_{i_m,j_m})
=V_0(z)^{-1}\hat\kappa_{t_{i_1,j_1}\ldots t_{i_m,j_m}}(z,s_1,\cdots,s_m)$.

With the ordering $t_{1,2}\succ t_{1,3}\succ t_{2,3}\succ t_{1,4}\succ t_{2,4}\succ t_{3,4}$
in \eqref{orderLyndon}, one has 
\begin{eqnarray*}
\varphi^z(t_{1,2})=\Sum_{k\ge0}\Frac{\log^k(z_3-z_4)}{(-2{\rm i}\pi)^kk!}P_{t_{3,4}^kt_{1,2}},&
\check\varphi^z(t_{1,2})=\Sum_{k\ge0}\Frac{\log^k(z_3-z_4)}{(-2{\rm i}\pi)^kk!}S_{t_{3,4}^kt_{1,2}},\cr
\varphi^z(t_{1,3})=\Sum_{k\ge0}\Frac{\log^k(z_3-z_4)}{(-2{\rm i}\pi)^kk!}P_{t_{3,4}^kt_{1,3}},&
\check\varphi^z(t_{1,3})=\Sum_{k\ge0}\Frac{\log^k(z_3-z_4)}{(-2{\rm i}\pi)^kk!}S_{t_{3,4}^kt_{1,3}},\cr
\varphi^z(t_{2,3})=\Sum_{k\ge0}\Frac{\log^k(z_3-z_4)}{(-2{\rm i}\pi)^kk!}P_{t_{3,4}^kt_{2,3}},&
\check\varphi^z(t_{2,3})=\Sum_{k\ge0}\Frac{\log^k(z_3-z_4)}{(-2{\rm i}\pi)^kk!}S_{t_{3,4}^kt_{2,3}},\cr
\varphi^z(t_{1,4})=\Sum_{k\ge0}\Frac{\log^k(z_3-z_4)}{(-2{\rm i}\pi)^kk!}P_{t_{3,4}^kt_{1,4}},&
\check\varphi^z(t_{1,4})=\Sum_{k\ge0}\Frac{\log^k(z_3-z_4)}{(-2{\rm i}\pi)^kk!}S_{t_{3,4}^kt_{1,4}},\cr
\varphi^z(t_{2,4})=\Sum_{k\ge0}\Frac{\log^k(z_3-z_4)}{(-2{\rm i}\pi)^kk!}P_{t_{3,4}^kt_{2,4}},&
\check\varphi^z(t_{2,4})=\Sum_{k\ge0}\Frac{\log^k(z_3-z_4)}{(-2{\rm i}\pi)^kk!}S_{t_{3,4}^kt_{2,4}},
\end{eqnarray*}
where $\check\varphi$ is the adjoint to $\varphi$ and is defined by
\begin{eqnarray*}
\check\varphi^{(\varsigma,z)}
=\sum_{k\ge0}\frac{\log^k(z_3-z_4)}{(-2{\rm i}\pi)^kk!}t_{3,4}^k
=e^{-(t_{3,4}/2{\rm i}\pi)\log(z_3-z_4)}.
\end{eqnarray*}

Hence, belonging to $\ncs{\calH(\widetilde{\C_*^4})}{\calT_4}$,
$G$ satisfies ${\bf d}G(z)=\bar{\Omega}_3(z)G(z)$, where
\begin{eqnarray*}
\bar{\Omega}_3(z)&=&(2{\rm i}\pi)^{-1}(\varphi^{(\varsigma,z)}(t_{1,2})d\log(z_1-z_2)+\varphi^z(t_{1,3})d\log(z_1-z_3)\\
&+&\varphi^{(\varsigma,z)}(t_{2,3})d\log(z_2-z_3)+\varphi^{(\varsigma,z)}(t_{1,4}) d \log(z_1-z_4)\\
&+&\varphi^{(\varsigma,z)}(t_{2,4})d\log(z_2-z_4)).
\end{eqnarray*}
In the affine plane $(P_{3,4}):z_3-z_4=1$, one has $\log(z_3-z_4)=0$ and then $\varphi\equiv\mathrm{Id}$.

By the cubic coordinate system on the moduli space $\mathfrak{M}_{0,5}$ \cite{Furusho2} we can put $z_1=xy,z_2=y,z_3=1, z_4 =0$, one has
\begin{eqnarray*}
\bar{\Omega}_3(xy,y,1,0)
&=&(2{\rm i}\pi)^{-1}(t_{12}d\log(y(1-x))+t_{13}d\log(1-xy)\\
&+&t_{23}d\log(1-y)+t_{14}d\log(xy)+t_{24}d\log y)\\
&=&(2{\rm i}\pi)^{-1}(t_{12}d \log(1-x)+t_{13}\log(1-xy)\\
&+&t_{23}d\log(1-y)+t_{14}d\log x+(t_{12}+ t_{14}+t_{24})d\log y).
\end{eqnarray*}
The differential equation 
$dG(x,y)=\bar{\Omega}_3(xy,y,1,0)G(x,y)$ admits the unique solution $G(x,y)$ \cite{drinfeld2} satisfying $G(x,y)\sim_{(0,0)}x^{(2\mathrm{i}\pi)^{-1}t_{1,4}} y^{(2\mathrm{i}\pi)^{-1}(t_{12}+t_{14}+t_{24})}$.
Thus, by the homographic substitution mapping $\{z_3,z_4\}$ to $\{1,0\}$
\begin{eqnarray*}
g:\left\{z_1\longmapsto(z_1-z_4)/(z_2-z_4),z_2\longmapsto(z_2-z_4)/(z_3-z_4)\right\},
\end{eqnarray*}
a particular solution of $KZ_4$ is $G((z_1-z_4)/(z_2-z_4),(z_2-z_4)/(z_3-z_4))$, in $(P_{3,4})$. Since grouplike series $(z_3-z_4)^{(2\mathrm{i}\pi)^{-1}\sum_{1\le i<j\le 4}t_{i,j}}=e^{(2\mathrm{i}\pi)^{-1}\log(z_3-z_4)\sum_{1\le i<j\le 4}t_{i,j}}$ is independent on $\{z_1=xy,z_2=y\}$ and belongs to the differential Galois group of $KZ_4$ then $G((z_1-z_4)/(z_2-z_4),(z_2-z_4)/(z_3-z_4))(z_3-z_4)^{(2\mathrm{i}\pi)^{-1}\sum_{1\le i<j\le 4}t_{i,j}}$ is a particular solution, in $(P_{3,4})$.

Now, for any $t\in\calT_4$, using quadratic relations relations given in \eqref{braidbis},
one has $[\sum_{1\le i<j\le 4}t_{i,j},t]=0$. Thus $t$ commutes with
$(z_3-z_4)^{(2\mathrm{i}\pi)^{-1}\sum_{1\le i<j\le 4}t_{i,j}}$ and then 
$(z_3-z_4)^{(2\mathrm{i}\pi)^{-1}\sum_{1\le i<j\lq 4}t_{i,j}}$ commutes with $\ncs{\calA}{\calT_4}$.
Thus, $KZ_4$ also admits $(z_3-z_4)^{(2\mathrm{i}\pi)^{-1}\sum_{1\le i<j\le 4}t_{i,j}}
G((z_1-z_4)/(z_2-z_4),(z_2-z_4)/(z_3-z_4))$ as solution in $(P_{3,4})$.

\vfill

\textbf{Acknowledgements.}
The author would like to thank D.~Barsky, G.~Duchamp and B.~Enriquez for fruitful interactions and improving suggestions and also J.Y.~Enjalbert, G.~Koshevoy and C.~Tollu for discussions.
\end{document}